\documentclass[12pt]{article}
\usepackage{mathptmx} % For a font similar to Times New Roman
\usepackage[T1]{fontenc}
\usepackage[left=1in, right=1in, top=1in, bottom=1in]{geometry} % includefoot

\usepackage{amsmath}
\usepackage{amsthm} 
\usepackage{amsfonts}
\usepackage{booktabs}
\usepackage{dsfont}
\usepackage{graphicx,psfrag,epsf}
\usepackage{enumitem}
\usepackage{url} % not crucial - just used below for the URL 
\usepackage[dvipsnames]{xcolor}
\usepackage{multirow,multicol}
\usepackage{tabularx}
\usepackage[title]{appendix}
\usepackage[section]{placeins}
\usepackage{stackengine}
\usepackage{threeparttable}
\usepackage{footmisc}
\usepackage[font={footnotesize}]{subcaption}
\usepackage{setspace}
\usepackage[textsize=footnotesize]{todonotes}
\usepackage[font=small, labelfont=bf]{caption}

\usepackage[
	breaklinks,
	colorlinks=true,
	pagebackref=false,
	pdfauthor={Timo Dimitriadis and Andrew J. Patton and Patrick W. Schmidt},%
	pdftitle={Testing Forecast Rationality for Measures of Central Tendency},% 
	citecolor=blue,
	linkcolor=blue,
	urlcolor=blue
	]{hyperref}      
\usepackage{cleveref}
\usepackage[authoryear,round]{natbib}  
\numberwithin{equation}{section}

\newtheorem{theorem}{Theorem}[section]
\newtheorem{lemma}[theorem]{Lemma}
\newtheorem{definition}[theorem]{Definition}
\newtheorem{proposition}[theorem]{Proposition}
\newtheorem{corollary}[theorem]{Corollary}

\newtheorem{remark}[theorem]{Remark}

\newtheorem{assumption}[theorem]{Assumption}



\allowdisplaybreaks
	
\newcommand{\defeq}{\mathrel{\mathop:}=}

\newcommand{\tod}{\overset{d}{\to}}
\newcommand{\toP}{\overset{P}{\to}}

\newcommand{\Var}{\operatorname{Var}}

\DeclareMathOperator*{\argmin}{arg\,min}

\DeclareMathOperator*{\mode}{Mode}

\newcommand{\cP}{\mathcal{P}}

\newcommand{\bh}{\mathbf{h}_t}

\newcommand{\WMean}{w_{\mathrm{Mean}}}
\newcommand{\WMeanHat}{\widehat{w}_{T,\mathrm{Mean}}}

\newcommand{\WMed}{w_{\mathrm{Med}}}
\newcommand{\WMedHat}{\widehat{w}_{T,\mathrm{Med}}}

\newcommand{\WMode}{w_{\mathrm{Mode}}}
\newcommand{\WModeHat}{\widehat{w}_{T,\mathrm{Mode}}}

\newcommand{\edit}[1]{{\color{black}#1}}

\begin{document}
%\the\textwidth\the\textheight

\title{Testing Forecast Rationality for Measures of Central Tendency\thanks{We thank Peter Boswijk, Jonas Brehmer, Tobias Fissler, Rafael Frongillo, Tilmann Gneiting, Arnaud Maurel, Adam Rosen, Melanie Schienle, J\"{o}rg Stoye, as well as seminar participants at Boston University, Duke University, KIT Karlsruhe, HITS Heidelberg, Monash University, UC Riverside, University of Amsterdam, Universität Hohenheim, Universität Heidelberg, University of Sydney, CSS Workshop in Zürich, HKMEtrics Workshop in Mannheim, ISI World Statistics Congress in Kuala Lumpur, the 2019 Statistische Woche in Trier, the 2019 $\text{EC}^2$ Conference in Oxford, and the 2020 Econometric Society World Congress. 
%Code to replicate the empirical work in this paper is available at \citet{DPS_fcrat_Package}.
Code to replicate the empirical work in this paper is available under \url{https://github.com/Schmidtpk/replication-central-tendency}, which draws on the \texttt{R} package \texttt{fcrat}, available at \url{https://github.com/Schmidtpk/fcrat}, that provides an implementation of the methods developed in this paper.}}

% Other title ideas - let's keep thinking about titles until we find the best one
% Testing Rationality for Central Tendency Forecasts
% Testing Forecast Rationality for Measures of Central Tendency
% Rationality Tests for Central Tendency Forecasts

\author{Timo Dimitriadis\thanks{Alfred-Weber-Institute for Economics, Heidelberg University and Heidelberg Institute for Theoretical Studies, Heidelberg, Germany, e-mail: \texttt{timo.dimitriadis@awi.uni-heidelberg.de} } 
\and Andrew J. Patton\thanks{Department of Economics, Duke University, Durham, USA, e-mail: \texttt{andrew.patton@duke.edu}} 
%\and Patrick Schmidt\thanks{Heidelberg Institute for Theoretical Studies, Heidelberg, Germany and Goethe University Frankfurt, Frankfurt, Germany, e-mail: \texttt{p.schmidt@econ.uni-frankfurt.de}}
\and Patrick W. Schmidt\thanks{University of Zurich, Zurich, Switzerland, e-mail: \texttt{patrickwolfgang.schmidt@uzh.ch}}
}

\date{}

\maketitle

{ \centering
%	\textbf{$\bullet \bullet \bullet$ \hspace{0.5cm}Preliminary draft. Please do not circulate. \hspace{0.1cm} $\bullet \bullet \bullet$ \vspace{1cm} }\\ 
\today
%	This Version: \today \\
%	First Version: February 15, 2019 
	\par
}

%\vspace*{1em}

\noindent\rule{\linewidth}{.4pt}
\textbf{Abstract}
\\[.3\baselineskip]
\noindent
\onehalfspacing
Rational respondents to economic surveys may report as a point forecast any measure of the central tendency of their (possibly latent) predictive distribution, for example the mean, median, mode, or any convex combination thereof. We propose tests of forecast rationality when the measure of central tendency used by the respondent is unknown. We overcome an identification problem that arises when the measures of central tendency are equal or in a local neighborhood of each other, as is the case for (exactly or nearly) symmetric distributions. As a building block, we also present novel tests for the rationality of mode forecasts. We apply our tests to income forecasts from the Federal Reserve Bank of New York's Survey of Consumer Expectations. We find these forecasts are rationalizable as mode forecasts, but not as mean or median forecasts. We also find heterogeneity in the measure of centrality used by respondents when stratifying the sample by past income, age, job stability, and survey experience. 
\\
[.3\baselineskip] 
\textbf{J.E.L. Codes:} C53, D84, E27 %E37
%\\
%\textit{JEL:} C13, C51, G31
% \\[-.5\baselineskip]
%\noindent\rule{\linewidth}{.4pt}

%\newpage
%\tableofcontents
% \pagebreak 

\section{Introduction} 
\doublespacing
% \onehalfspacing

%\color{red} Target length using this format (including tables, figures, references, paper appendix): 50-52 pages. Supplemental appendix has no limit. So need to cut about 15 pages. Should be doable. \color{black}

Economic surveys are a rich source of information about future economic conditions, yet most economic surveys are vague about the specific statistical quantity the respondent should report. For example, the New York Federal Reserve's labor market survey asks respondents ``What do you believe your annual earnings will be in four months?'' A reasonable response to this question is the respondent reporting her mathematical expectation of future earnings, or her median, or her mode; all common measures of the central tendency of a distribution. When these measures coincide, as they do for symmetric unimodal distributions, this ambiguity does not affect the information content of the forecast. When these measures differ, the specific measure adopted by the respondent can influence its use in other applications, and testing rationality of forecasts becomes difficult. Subjective forecast distributions have been found to be asymmetric for many important economic variables as GDP growth \citep{adrian2019vulnerable, bekaert2019link}, inflation rates \citep{garcia2007can}, firm earnings \citep{Foster1986, givoly2000changing, gu2003earnings} and consumer expenses \citep{Howard2022}, motivating the need for reliable forecast evaluation methods for the different measures of central tendency.

% \citep{salgado2019skewed}: skewness in employment, sales, and productivity (even time-varying left and right skewness)

While the assumption of a mean forecast\footnote{We use the phrase ``mean forecasts,'' or similar, as shorthand for the forecaster reporting the mean of her predictive distribution as her point forecast.} is common in the economic and statistical literature (see, e.g., \citealp{coibion2015information, bordalo2020overreaction}), it may be mistaken in some applications. For example, \cite{Knuppel2012} conclude that point forecasts of inflation published by central banks almost always correspond to the modes of the forecast densities, and \cite{ReifschneiderTulip2017} suggest that forecasts from the U.S. Board of Governors are best interpreted as mode forecasts. \cite{Howard2022} and \cite{zhao2022internal} find that some survey forecasts are more consistent with the respondents' modes than means or medians. In an early experimental study, \cite{peterson1964mode} found that respondents could accurately predict the mode and median if incentivized correctly, but had difficulty reporting accurate estimates of the mean, and a more recent study by \cite{KroegerThibaud2019} reports that participants asked to summarize their predictive distribution responded most frequently with their mode. We propose that other measures of central tendency deserve further consideration, especially given the growing evidence that point forecasts may reflect the mode.

Given the ambiguity around which specific measure of central tendency is used by survey respondents, we consider a \textit{class} of such measures defined by the set of convex combinations of the mean, median, and mode.\footnote{These are the three measures described in introductory statistics textbooks (\citealp{McClave2017}), in previous studies (\citealp{engelberg2009comparing}), in central bank publications (the \citeauthor{BoE2019}'s quarterly inflation reports), and in psychological work, \citep[p. 86,][]{kahneman1982judgment}. Nevertheless, our approach can easily be extended to consider a broader set of measures of centrality.} \edit{This approach, naturally, nests each of these measures separately, and also allows for survey respondents who use a mixture of measures of central tendency, e.g., a respondent who ``tilts'' her mean forecast in the direction of the most likely outcome (the mode) or towards a robust measure of location (the median). Such mixtures may arise through the formal minimization of a mixture of the respective loss functions for these measures, or informally as the respondent uses ``expert judgment'' to adjust a forecast of a specific measure of central tendency.}
	%Notably, our approach allows for forecasts that lie between the mean, median and mode, and remains informative for mixtures of the three forecasts. Tests of rationality that consider each functional separately would dismiss such forecasts as irrational or non-central. AJP REVISE, AND PROBABLY EXTEND, THIS PARAGRAPH. THIS IS R1'S FIRST COMMENT -- NEED A COMPELLING ARGUMENT HERE.}

 %\textcolor{red}{In line with early psychological work, we consider it natural that representative values of a probability distribution are close to the mean, median, or mode \citep[p. 86,][]{kahneman1982judgment}. Thus, we propose tests of forecast rationality that consider any convex combination of the mean, median, and mode. These are also the three measures of central tendency described in standard introductory statistics textbooks, e.g.\ \cite{McClave2017}\footnote{The \citeauthor{BoE2019}'s quarterly Inflation Report, for example, reports these three measures as distinct point forecasts for future GDP growth and inflation. Our approach can easily be extended to consider a broader set of measures of centrality.}.  %AJP in 2022: re-phrased as we cut the footnote mentioning trimmed mean, etc. 
%Further, \cite{engelberg2009comparing} consider the three measures in their initial work on the consistency between point forecasts and probabilistic predictions.} 

Similar to \citet{EKT2005}, we propose a testing framework that nests the mean as a special case, but unlike that paper we allow for alternative forecasts \textit{within} a class of measures of central tendency, rather than measures that represent other aspects of the predictive distribution (such as non-central quantiles or expectiles). Our approach faces an identification problem: for symmetric distributions the combination weight vector is unidentified, for ``mildly'' asymmetric distributions the weight vector is weakly identified, and even for strongly asymmetric distributions the weight vector may only be partially identified. Economic variables may fall into any of these cases, and a valid testing approach must accommodate these measures of central tendency being equal, unequal, or in a local neighborhood of each other. We use the work of \cite{StockWright2000} to obtain asymptotically valid confidence sets for the combination weights and to test forecast rationality.

%\textcolor{red}{Depending on the skewness of the respondent's predictive distribution, our approach allows the researcher to distinguish mean, median, and mode forecasts based on point forecasts and realizations alone. An alternative approach to answer the same question was proposed in \cite{engelberg2009comparing}, who combine point and density forecasts to determine whether the point forecasts are the mean, median, or mode of their predictive densities. \cite{engelberg2009comparing} find no conclusive evidence as most point forecasts are consistent with the bounds derived for all three functionals. Similar conclusions are drawn by \cite{clements2010explanations} on an extended sample. \cite{zhao2022internal} applies the methodology to inflation expectations of the Survey of Consumer Expectations and finds the mode of the probabilistic prediction to be most consistent with the point forecasts, but the mean and median were also valid for the majority of forecasts.}

Depending on the skewness of the respondent's predictive distribution, our approach may allow the researcher to distinguish mean, median, and mode forecasts based on point forecasts and realizations alone. An alternative approach was proposed in \cite{engelberg2009comparing}, who combine point and density forecasts to determine whether the point forecasts are the mean, median, or mode of the predictive density. In their study of respondents to the Federal Reserve Bank of Philadelphia's Survey of Professional Forecasters, \cite{engelberg2009comparing} find that most point forecasts are consistent with the bounds derived for all three functionals. \cite{zhao2022internal} similarly studies inflation forecasts from the Federal Reserve Bank of New York's Survey of Consumer Expectations and finds the mode of the density forecast to be most consistent with the point forecasts, but the mean and median are also valid for a majority of forecasts. Our testing approach accommodates the fact that these functionals may be hard, or impossible, to separately identify in data.

Before implementing the above test for rationality for a general forecast of central tendency, we must first overcome a lack of rationality tests for mode forecasts. Rationality tests for mean forecasts go back to at least \cite{MincerZarnowitz1969}, see \citet{ETbook2016} for a recent survey, while rationality tests for quantile forecasts (nesting median forecasts as a special case) are considered in \citet{christoff} and \citet{GaglianoneJBES2011}. A critical impediment to similar tests for mode forecasts is that the mode is not an ``elicitable functional'' \citep{Heinrich2014}, meaning that it cannot be obtained as the solution to an expected loss minimization problem.\footnote{\cite{Gneiting2011} provides an overview of elicitability and identifiability of statistical functionals and shows that several important functionals such as variance, Expected Shortfall, and mode are not elicitable. \cite{FisslerZiegel2016} introduce the concept of higher-order elicitability, which facilitates the elicitation of vector-valued (stacked) functionals such as the variance and Expected Shortfall, though not the mode \citep{Dearborn2019}.}$^,$\footnote{For categorical data, rationality tests for %mean, median and 
mode forecasts were established in \cite{das1999comparing} and extended in \cite{madeira2018testing}. Our focus is on continuously distributed target variables, and so we cannot use their results.} We obtain a test for mode forecast rationality by first proposing novel results on the \textit{asymptotic elicitability} of the mode. We define a functional to be asymptotically elicitable if there exists a sequence of elicitable functionals that converges to the target functional. We consider the (elicitable) ``generalized modal interval,'' defined in detail in Section \ref{sec:ModeFunctional}, and show that it converges to the mode for a general class of probability distributions. We combine these results with recent work on mode regression (\citealp{Kemp2012}; \citealp{Kemp2019}) and nonparametric kernel methods to obtain mode forecast rationality tests analogous to well-known tests for mean and median forecasts. In addition to size control, we show that the proposed test has non-trivial asymptotic power against both fixed and local alternative hypotheses. 

We evaluate the finite sample performance of the new mode rationality test and of the proposed method for obtaining confidence sets for measures of centrality through an extensive simulation study. We use cross-sectional and time-series data generating processes with a range of asymmetry levels. We find that our proposed mode forecast rationality test has satisfactory size properties, even in small samples, and exhibits strong power across different misspecification designs. Our simulation design allows us to consider the four identification cases that can arise in practice: strongly identified (skewed data), where the mean, median and mode differ; unidentified (symmetric unimodal data), where all centrality measures coincide; weakly identified (mildly skewed data), where the centrality measures differ but are close to each other; and partially identified (skewed location-scale data), where one centrality measure is a convex combination of the other two. We find that in the symmetric case, the resulting confidence sets contain, correctly, the entire set of convex combinations of mean, median and mode. In the asymmetric cases, our rationality test is able to identify the combination weights corresponding to the issued centrality forecast.

We apply the new tests to the income survey responses from the Survey of Consumer Expectations conducted by the Federal Reserve Bank of New York. In the full sample of respondents, we find the intriguing result that we can reject rationality with respect to the mean or median, however we cannot reject rationality when interpreting these as mode forecasts, suggesting that survey participants report the anticipated \emph{most likely} outcome rather than the average or median. When allowing for cross-respondent heterogeneity, we find that forecasts from low-income younger survey respondents cannot be rationalized using any measure of central tendency, while forecasts from high-income respondents, regardless of their age, are rationalizable for many different measures of centrality. We also find evidence of learning between survey rounds \citep{kim2022learning} for high-income respondents, but much less so for low-income respondents. %We also find that mean forecasts can only be rejected for the subgroup with relatively small income shocks in the first round, while for respondents who faced a relatively large income shock, subsequent reports are consistent with mean forecasts. This pattern indicates that more sophisticated measures like the mean are reported after repeated participation and after large shocks. The latter is consistent with ``atypical'' outcomes moving predictions from the mode to the mean \citep{Howard2022}. In the first round, and if no unexpected income change happened, predictions are consistent with the mode and entail no information on the uncertainty in the tails of the individual income distribution.
We compare our rationality test results with those obtained using the approach of \citet{EKT2005} (EKT), which allows for rational optimism or pessimism. We find no cases where allowing for optimism or pessimism ``overturns'' a rejection of rationality based on a measure of centrality. As a stark example, we find that forecasts from low-income younger respondents cannot be rationalized as any centrality measure, nor as a measure in the EKT framework.

Our paper is related to the large literature on forecasting under asymmetric loss, see \citet{Granger69}, \citet{christoffersen1997optimal}, \citet{EKT2005}, \citet{PattonTimmermann2007} and \citet{biases} amongst others. The work in these papers is motivated by the fact that forecasters may wish to use a loss function other than the omnipresent squared-error loss function. The use of asymmetric loss functions generally leads to point forecasts that differ from the mean (though this is not always true, see \citealp{Gneiting2011} and \citealp{patton2014comparing}), and generally these point forecasts are not interpretable as measures of central tendency. For example, \citet{christoffersen1997optimal} show that the linex loss function implies an optimal point forecast that is a weighted sum of the mean and variance, while \citet{biases} find that their sample of macroeconomic forecasters report an expectile with asymmetry parameter around 0.4. Instead of moving from the mean to a point forecast that is not a measure of location, our novel approach considers moving only \textit{within} a general class of central tendency measures. %\footnote{\cite{BauerRudebusch2016} suggest using the \textit{difference} between the mean and the mode of the predictive distribution for future interest rates as a measure of the likelihood that interest rates will hit the zero lower bound.}

The remainder of the paper is structured as follows. 
In Section \ref{sec:ElicitAndEvalModeFC} we propose new forecast rationality tests for the mode based on the concepts of asymptotic elicitability and identifiability. Section \ref{sec:FunctionalCentrality} presents forecast rationality tests for general measures of central tendency, allowing for weak and partial identification. Section \ref{sec:SimulationStudy} presents simulation results on the finite-sample properties of the proposed tests, and Section \ref{sec:ApplicationSCE} presents the empirical application using income survey forecasts. A supplemental appendix contains \edit{all proofs, further theoretical and simulation results}, and empirical analyses of the Federal Reserve Board's ``Greenbook'' forecasts of US GDP growth and random walk forecasts of exchange rates.

%\vspace{2cm}
%\color{red}
%Timo: Mention \cite{Howard2022}, who find that consumers' expense forecasts are closer to the mode than the mean. They argue that consumers forecast ``typical'' expenses, and forget infrequent but larger ones like car repairs, etc. This results in a skewed distribution.
%(They refer to this as ``underprediction'', but it is actually mode prediction...)
%
%In their ``General Discussion'' section, they have a ``Cognitive Accessibility and the Psychology of Prediction'' section that could be good for us. I think one referee wanted us to go a bit in this direction.
%
%The ``Typicality'' subsection also looks interesting!
%
%In the outlook: ``This raises the intriguing possibility that what people perceive to be typical or average in these contexts could be the mode, as our research suggests is the case for expenses. Future research should also examine the perceptual difference between modal and median outcomes and test the possibility that the latter also influences perceived typicality. ''
%
%Based on experimental variation in skewness, ``skewed distribution shifted predictions away from the mean and toward the mode''
%\color{black}

\section{Eliciting and Evaluating Mode Forecasts}
\label{sec:ElicitAndEvalModeFC}

%In this section, we consider testing the rationality of mode forecasts with asymptotically valid identification functions, as introduced in the previous section.
%The results can be applied to cross-sectional applications as a special case with time-independent observations.

\subsection{General Forecast Rationality Tests}
\label{sec:GeneralRatTests}

Let $Z_t = \big( Y_{t}, X_t, \widetilde{\mathbf{h}}_t \big)$ \edit{with $t \in \mathbb{N}$} be a stochastic process defined on a common probability space $\big( \Omega, \mathcal{F}, \mathbb{P} \big)$. $Y_{t+1}$ denotes the (scalar) variable of interest, $\widetilde{\mathbf{h}}_t$ denotes a vector of variables known to the forecaster at the time she issues her point forecast for $Y_{t+1}$, which is denoted $X_t$.\footnote{Given our focus on survey forecasts, where the underlying model used by the respondents, if any, is unknown, we take the forecasts as given, putting this paper in the general framework of \cite{giacomini2006tests}, as opposed to that of \cite{West1996}.} 
\edit{We define the information set $\mathcal{F}_t = \sigma \big\{ \widetilde{\mathbf{h}}_s; s \le t \big\}$ as the $\sigma$-field containing all information known to the forecaster at time $t$. 
	Throughout the paper, we assume that $X_t \in \mathcal{F}_t$.}
	We denote the distribution of $Y_{t+1}$ given $\mathcal{F}_t$ by $F_t$, and %Whenever they exists, we denote 
with corresponding density $f_t$. 
Neither the forecaster's information set, $\mathcal{F}_t$ nor her predictive distribution, $F_t$, are assumed known to the econometrician. In conducting the forecast rationality test, we assume that the econometrician uses an $\mathcal{F}_t$-measurable $(k \times 1$) vector $\bh$, which can be thought of as a subset of $\widetilde{\mathbf{h}}_t$.\footnote{The intepretability of the outcome of a forecast rationality test, including ours, critically depends on the ``test vector'' or ``instrument vector,'' $\bh$, being observable to the forecaster. In our empirical analysis we only use vectors that are guaranteed to satisfy this requirement.} Note that since $F_t$ is not observed, we do not know whether this distribution is strongly skewed, mildly skewed, or symmetric, and thus our inference method must be valid for all of these possibilities. Conditional expectations are denoted $\mathbb{E}_t[\cdot] =\mathbb{E}[\cdot|\mathcal{F}_t]$.
We use $\cal{P}$ to denote a class of distributions.

We start by considering rationality tests (also known as calibration tests; \citealp{Nolde2017}) for the mean, i.e.\ we assume that the forecasts $X_t$ are one-step ahead \textit{mean} forecasts for $Y_{t+1}$.
We are interested in testing if these forecasts are rational% given some instruments $\bh$
, which would imply the null hypothesis:
\begin{align} 
	\label{eq:H0MeanRat}
	\mathbb{H}_{0}: X_t = \mathbb{E} [ Y_{t+1} | \mathcal{F}_t ] ~ \forall~t \text{ a.s.} 
 \end{align}
We test this hypothesis using the ``identification function'' for the mean, which is simply the difference between the forecast and the realized value, i.e., the forecast error:\footnote{The identification function for a point forecast can be obtained as the first derivative of any loss function that elicits that forecast. The quadratic loss function elicits the mean, and so its identification function is simply the forecast error, up to scale and sign. In econometrics the forecast error is usually defined as $Y_{t+1}-X_t$, that is, as the negative of the identification function in equation (\ref{eq:MeanIF}). Given the important role that forecast identification functions play in this paper, we adopt the definition for $\varepsilon_{t}$ given in equation (\ref{eq:MeanIF}), and we refer to $\varepsilon_{t}$ as the forecast error.}
 \begin{align}
	V_{\mathrm{Mean}}  \big( X_t, Y_{t+1} \big) = X_t - Y_{t+1} =: \varepsilon_{t}. \label{eq:MeanIF}
 \end{align}

Specifically, the null hypothesis in \eqref{eq:H0MeanRat} implies that the identification function $V_{\mathrm{Mean}}  \big( X_t, Y_{t+1} \big)$ is uncorrelated with any $(k \times 1)$ instrument vector $\bh \in \mathcal{F}_t$, which provides a testable moment condition.
Under the above null hypothesis and subject to standard regularity conditions, it is straight forward to show that $T^{-1/2} \sum_{t=1}^T V_{\mathrm{Mean}}  \big(X_t, Y_{t+1} \big) \bh \tod \mathcal{N} \big( 0,\Omega_{\mathrm{Mean}} \big)$ as $T \rightarrow \infty$,
and that
\begin{align}
	\label{eqn:WaldTestStatMean}
	J_{T} = \frac{1}{T} \left( \sum_{t=1}^T V_{\mathrm{Mean}} \big( X_t, Y_{t+1}\big) \bh^\top  \right)  \widehat{\Omega}_{T,\mathrm{Mean}}^{-1} \left(  \sum_{t=1}^T V_{\mathrm{Mean}}  \big( X_t, Y_{t+1} \big)  \bh \right) \tod \chi^2_k
\end{align}
as $T \rightarrow \infty$, where $\widehat{\Omega}_{T,\mathrm{Mean}}$ is a consistent estimator of $\Omega_{\mathrm{Mean}}$.
This result facilitates testing whether given forecasts $X_t$ are rational mean forecasts for the realizations $Y_{t+1}$ by using the test statistic $J_T$ in equation (\ref{eqn:WaldTestStatMean}) to test for uncorrelatedness of  the identification function $V_{\mathrm{Mean}}  \big( X_t, Y_{t+1} \big)$ and the instrument vector $\bh$. 
% Under the null of forecast rationality, this correlation should be zero. 
As in most other tests in the literature, this is of course only a test of a necessary condition for forecast rationality, and the conclusion may be sensitive to the choice of instruments, $\bh$.\footnote{Like most of the forecast evaluation literature, we assume that the vector of instruments is of fixed and finite length. A \cite{Bierens82}-type test, where the length of the vector diverges with the sample size, is considered for forecast evaluation in, for example, \cite{CS02}.} %\color{red}Good instruments span the information set of the forecaster and are not strongly correlated with each other. Additional instruments generally improve power asymptotically, but can be harmful in finite samples. A valid and informative instrument, one used as far back as \cite{MincerZarnowitz1969}, is the forecast $X_t$ itself \citep{schmidt2021interpretation}. \color{black}
% @PWS: I moved these sentences to merge them with a similar discussion of instruments for the Stock-Watson test. 

%\td{PWS: I see now that we already had a discussion on the instruments (beginning of subsection{Forecast Rationality Tests for the Mode}) and decided to keep the issue mostly in the background. In this light, we might adapt our plan of how to respond to the referee comment and change nothing or very little in the text.}

%While testing whether the average forecast error equals zero is intuitive, this procedure 
The test statistic in equation (\ref{eqn:WaldTestStatMean}) is only informative about rationality if the forecasts are interpreted as being for the \textit{mean} of $Y_{t+1}$. The decision-theoretic framework of identification functions and consistent loss functions is fundamental for generalizations to other measures of central tendency, such as the median and the mode. For a general real-valued functional $\Gamma: \mathcal{P} \longrightarrow \mathbb{R}$, a strict identification function $V_\Gamma(x,Y)$ is defined by being zero in expectation if and only if $x$ equals the functional $\Gamma(F)$. Strict identification functions are generally obtained as the derivatives of strictly consistent loss functions, which are defined as having the functional $\Gamma(F)$ as their unique minimizer (in expectation). A functional is called \textit{identifiable} if a strict identification function exists, and is called \textit{elicitable} if a strictly consistent loss function exists. %We refer to Section \ref{sec:AsymptoticElicitability} in the online supplement for an extensive introduction to this theory. 
See \cite{Gneiting2011} for a general introduction to elicitability and identifiability.

The forecast error $X_t - Y_{t+1}$ is a strict identification function for the mean, and a strict identification function for the median is given by the step function:
\begin{align}
	\label{eq:MedIF}
	V_{\mathrm{Med}} \big( X_t, Y_{t+1} \big) = \mathds{1}_{\{Y_{t+1} < X_t\}} - \mathds{1}_{\{Y_{t+1} > X_t\}}.
\end{align}
We obtain a test of median forecast rationality by replacing $	V_{\mathrm{Mean}}$ and 
$\widehat{\Omega}_{T,\mathrm{Mean}}$ by $V_{\mathrm{Med}}$ and $\widehat{\Omega}_{T,\mathrm{Med}}$
in equation (\ref{eqn:WaldTestStatMean}).

\subsection{The Mode Functional}
\label{sec:ModeFunctional}

In contrast to the mean and the median, rationality tests for mode forecasts are more challenging to consider. 
The underlying reason is that there do not exist identification functions for the mode for random variables with continuous Lebesgue densities \citep{Heinrich2014, Dearborn2019}. In this section we simplify notation and refer to the target variable and forecast as $Y$ and $x$. We define the mode for random variables with continuous Lebesgue densities as the global maxima of the density function.\footnote{
More generally, the mode is often defined as the limit, as $\delta \to 0$, of the modal midpoint functional $\operatorname{MMP}_\delta$, given in equation (\ref{eqn:ModalMidPoint}) below \citep{Gneiting2011, Dearborn2019}.
This definition coincides with the global maxima of the density function for distributions with continuous Lebesgue density; and it coincides with the points of maximal probability for discrete distributions \citep{Heinrich2014}.
Note that our definition of the mode as the global maxima of a density function is only valid for distributions with \textit{continuous} Lebesgue densities as otherwise the density function can be modified on singletons (null sets in the distribution) without altering the underlying probability measure.
%Following \cite{Gneiting2011, Heinrich2014, Dearborn2019}, for a general random variable with distribution $P$, the mode functional can be
}
We make the following distinction in the notion of \textit{unimodality}. 
\begin{definition}
	\label{def:Unimodality}
	An absolutely continuous distribution is defined as 
	\edit{(a) \emph{weakly unimodal} if it has a unique mode, i.e., a unique point $m \in \mathbb{R}$ such that $f(x) < f(m)$ for all $x \neq m$, and (b) \emph{strongly unimodal} if
	there exists a unique point $m \in \mathbb{R}$ such that for its differentiable density function $f$, it holds that $f'(x) \ge 0$ if $x<m$ and $f'(x) \le 0$ if $x>m$.}
	%\edit{there exists a unique point $m \in \mathbb{R}$ such that the density function is strictly increasing for all $x < m$ and strictly decreasing for all $x > m$.}
	% it is weakly unimodal  and does not have further local modes. 
\end{definition}
\cite{HeinrichFissler2021} show that even for the class of strongly unimodal distributions, neither strictly consistent loss functions nor strict identification functions exist for the mode.
%While this does not imply that there do not exist any such functions for the class of strongly unimodal distributions, none have yet been found.
\cite{Gneiting2011} notes that it is sometimes stated informally that the mode is an optimal point forecast under the loss function $L_\delta(x,Y) = \mathds{1}_{\{ | x - Y | \le \delta \}}$ for some fixed $\delta > 0$.
In fact, this loss function elicits the midpoint of the modal interval (also known as the modal midpoint, or MMP) of length $2 \delta$. %(Modal intervals are often reported in Bayesian empirical work as ``credible sets'' for estimated parameters.)
The MMP of $Y \sim P$ is defined as the midpoint of the interval of length $2 \delta$ that contains the highest probability:
\begin{align}
	\label{eqn:ModalMidPoint}
	\operatorname{MMP}_\delta( P ) = \sup_{x \in \mathbb{R} } \mathbb{P} \big( Y \in [x-\delta, x+\delta] \big).
\end{align}  
More formally, it holds that for any $\delta > 0$ small enough, the modal midpoint is well-defined for all distributions with unique and well-defined mode, and it holds that $ \lim_{\delta \downarrow 0}\operatorname{MMP}_\delta( P ) =  \mode( P )$ \citep{Gneiting2011}.
In a similar manner, \cite{Eddy1980}, \cite{Kemp2012} and \cite{Kemp2019} propose estimation of the mode by estimating the modal interval with an asymptotically shrinking length. We formalize these ideas in the decision-theoretical framework in the following definition.

\begin{definition}
The functional $\Gamma: \mathcal{P} \longrightarrow \mathbb{R}$ is \textit{asymptotically elicitable (identifiable)} relative to $\cP$ if there exists a sequence of elicitable (identifiable) functionals $\Gamma_\delta: \mathcal{P} \longrightarrow \mathbb{R}$, such that $\Gamma_\delta(P) \to \Gamma(P)$ as $\delta \to 0$  for all $P \in \mathcal{P}$.
\end{definition}

%For any $\delta > 0$ small enough, the modal midpoint is well-defined for all distributions with unique and well-defined mode and it holds that $ \lim_{\delta \downarrow 0}\operatorname{MMP}_\delta( P ) =  \mode( P )$ \citep{Gneiting2011}.
As the modal midpoint converges to the mode, this establishes \textit{asymptotic elicitability} for the mode functional for the class of weakly unimodal probability distributions with continuous Lebesgue densities.
Unfortunately, this does not directly allow for asymptotic \textit{identifiability} of the mode as any pseudo-derivative of the loss function $L_\delta$ equals zero.
We establish asymptotic identifiability of the mode through the \textit{generalized modal midpoint}.

\begin{definition}
	\label{def:GeneralizedModalMidpoint}
Given a kernel function $K(\cdot)$ and bandwidth parameter $\delta$, the \textit{generalized modal midpoint}, $\Gamma_\delta^K(P)$, of $Y \sim P$ is defined as %the functional with parameter $\delta$ as the minimizer of the expected loss function
	\begin{align}
	\begin{aligned}
	\label{eqn:LossGenModalInt}
	\Gamma_\delta^K(P) = \argmin_{x \in \mathbb{R}} \mathbb{E} \left[ L^K_\delta(x,Y) \right], \qquad \text{ where } \qquad
	L^K_\delta(x,Y) = - \frac{1}{\delta} K\left( \frac{x-Y}{\delta} \right).
	\end{aligned}
	\end{align}
\end{definition}
The familiar modal midpoint is nested in this definition by using a rectangular kernel for $K(u) = \mathds{1}_{\{ | u | \le 1 \}}$, and this definition allows for smooth generalizations. As this definition involves the argmin of a function, we first establish that this is well-defined and that it converges to the mode functional. The following theorem also considers identifiability of the generalized modal midpoint.

%\begin{proposition}
%	\label{prop:GenModalIntWellDefined}
%	Let $\mathcal{P}$ be the class of distributions consisting of absolutely continuous and weakly unimodal distributions with bounded and Lipschitz-continuous density.
%	Furthermore, let $K$ be a strictly positive kernel function such that $\int K(u) \mathrm{d}u = 1$, $\int |u| K(u) \mathrm{d}u < \infty $ and $\log(K(u))$ is a concave function.
%	Then, the functional $\Gamma_\delta^K$ induced by the loss function (\ref{eqn:LossGenModalInt}) is well-defined for all $\delta > 0$ and for $\delta \to 0$, it holds that $\Gamma_\delta(P) \to \mode(P)$ for all $P \in \mathcal{P}$.
%\end{proposition} 

\begin{theorem}
	\label{thm:GenModalMidpointProperties}
	Let $K$ be a strictly positive kernel function on the real line that is log-concave, i.e.\ $\log(K(u))$ is a concave function, and additionally let $\int K(u) \mathrm{d}u = 1$ and $\int |u| K(u) \mathrm{d}u < \infty$.
	Let $\mathcal{P}$ be the class of absolutely continuous and weakly unimodal distributions with bounded and Lipschitz-continuous density and let $\tilde{\mathcal{P}} \subset \mathcal{P}$ be the subclass of strongly unimodal distributions.
%	Let $\tilde{\mathcal{P}}$ be a class of strongly unimodal and absolutely continuous distributions with a continuously differentiable and Lipschitz-continuous density $f$.
	\begin{enumerate}[label=(\alph*)]
		\item 
		\label{statement:GenModalMidpointWellDefined}
		The functional $\Gamma_\delta^K$ induced by the loss function (\ref{eqn:LossGenModalInt}) is well-defined for all $\delta > 0$ and $P \in \mathcal{P}$.

		\item 
		\label{statement:GenModalMidpointConvergence}
		It holds that $\Gamma^K_\delta(P) \to \mode(P)$ as $\delta \to 0$ for all $P \in \mathcal{P}$.
		
		\item 
		\label{statement:GenModalMidpointIdentifiable}
		If $K$ is differentiable, for all (fixed) $\delta>0$ and $P \in \tilde{\mathcal{P}}$, it holds that the function 
		\begin{align}
			\label{eq:ModeIF}
			V_\delta^K (x,Y) = \frac{\partial}{\partial x}  L^K_\delta(x,Y) = - \frac{1}{\delta^2} K' \left( \frac{x-Y}{\delta} \right)
		\end{align}
		is a strict identification function for $\Gamma_\delta^K$.
%		$E_{y \disas P}[ K' \left( \frac{x-y}{\delta} \right) ]$ has a unique root.\\
		In particular, the generalized modal midpoint is identifiable and the mode is asymptotically identifiable with respect to $\tilde{\mathcal{P}}$.
	\end{enumerate}	
\end{theorem} 
This theorem shows that for the classes of weakly (strongly) unimodal distributions, the generalized modal midpoint is elicitable (and identifiable), and consequently, the mode is asymptotically elicitable (and identifiable).
While $V_\delta^K (x,Y)$ being an identification function for the generalized modal midpoint is obvious from Definition \ref{def:GeneralizedModalMidpoint}, Theorem \ref{thm:GenModalMidpointProperties}\ref{statement:GenModalMidpointIdentifiable} establishes its strictness.

For a fixed $\delta > 0$, strict identifiability is doomed to fail when both the underlying distribution and the kernel function have bounded support as the expected identification function equals zero for values far outside both supports.
Theorem \ref{thm:GenModalMidpointProperties}\ref{statement:GenModalMidpointIdentifiable} shows that employing log-concave kernels with infinite support circumvents this problem.
While kernels with bounded support generally exhibit a superior performance in nonparametric statistics, this proposition motivates the use of kernels with infinite support like the Gaussian density function.
Furthermore, the Gaussian density function, among many others, satisfies the required log-concavity of the kernel function.\footnote{Theorem \ref{thm:GenModalMidpointProperties}\ref{statement:GenModalMidpointIdentifiable} also holds if log-concavity of the underlying density, instead of the kernel function, holds. This illustrates that kernels with bounded support can be employed at the cost of restricting the class of distributions.}

\edit{Beyond the rationality tests proposed in the following, the concept of asymptotic elicitability is of interest in its own right. Asymptotic elicitability may facilitate forecast evaluation and comparison for novel statistical functionals. For example, \cite{schmidt2021belief} propose an elicitation procedure for the maximum, a functional that generally is not elicitable \citep{bellini2015elicitable}.}

\subsection{Forecast Rationality Tests for the Mode}
\label{sec:ForecastRatioanlityTestMode}

Having established the asymptotic identifiability of the mode in the previous section, we now consider rationality testing of mode forecasts, i.e.\ testing the following null hypothesis,
\begin{align}
	\label{eqn:NullHypothesisMode}
	\mathbb{H}_0: \quad X_t = \mode( Y_{t+1} | \mathcal{F}_t )  \quad  \forall ~ t \text{ a.s.}
\end{align}
While classical, $\sqrt{T}$-consistent rationality tests based on strict identification functions are unavailable for the mode due to its non-identifiability, we next propose a rationality test for mode forecasts based on an asymptotically shrinking bandwidth parameter $\delta_T$. Consider the (asymptotically valid) identification function $V^K_{\delta_T}$ with bandwidth $\delta_T$, and multiplied by the instruments $\bh$,
\begin{align}
	\label{eqn:PsiMode}
	\psi(Y_{t+1},X_t,\bh, \delta_T) := V_{\delta_T}^K(X_t, Y_{t+1}) \bh =  - \frac{1}{\delta_T^2} K' \left( \frac{X_t-Y_{t+1}}{\delta_T} \right)  \bh.
\end{align}
 
%Asymptotically shrinking $\delta_T \to 0$ allows to construct an asymptotically valid mode rationality test.
%The identification functions are given by
%\begin{align}
%	\psi_{t,n} = \psi_{\sigma_n} (Y_{t+1}, X_t, \bh) 
%	= \bh \frac{X_t - Y_{t+1} }{\sigma_n^2} \phi \left( \frac{X_t-Y_{t+1}}{\sigma_n} \right)
%	= \bh \frac{1}{\sigma_n^2} \phi'\left( \frac{X_t-Y_{t+1}}{\sigma_n} \right)
%\end{align}
%I guess the variance stabilizes at the following transformation
%\begin{align}
%	\Var \left( (n \sigma_n^3)^{(1/2)} \frac{1}{n} \sum \psi_{t,n} \right)
%	\to M_1 \mathbb{E} \left[ f(0) \bh^\top \bh \right],
%\end{align}
%where $M_1 = \int \phi'(u)^2 \mathrm{d}u$.
%
%
%What is $f$? density of the FC error conditional on $\bh$?
%\subsection{One-Step Ahead Forecasts}
We make the following assumptions. For the remainder of the paper all limits are taken as $T \to \infty$, unless stated otherwise. 
\begin{assumption}
	\label{assu:ModeRationality}
	%	We make the following assumptions: 
	$ $ 
	\begin{enumerate}[label=(A\arabic*)]
		\item 
		\label{assu:StationaryErgodicSequence}
%		OLD: The sequence $\big( \varepsilon_{t}, \bh \big)$ is stationary and ergodic. \\
		The sequence $\big( \varepsilon_{t}, \bh \big)$ for $t \in \mathbb{N}$ is $\alpha$-mixing of size $-r/(r-1)$ for some $r > 1$. 
			
		\item 
		\label{assu:Moments}
		It holds that $\mathbb{E} \left[ ||\bh||^{2r+\delta} \right] < \infty$ for all $t \in \mathbb{N}$ for some $\delta > 0$.
		
		\item
		\label{assu:FullRank}
		The matrix $\mathbb{E} \left[\bh \bh^\top \right]$ has full rank for all $t \in \mathbb{N}$.
		
		\item 
		\label{assu:CovConvergence}
		The limit, $\Omega_{\mathrm{Mode}}$, of $\Omega_{T,\mathrm{Mode}} := \frac{1}{T} \sum_{t=1}^T \mathbb{E} \left[ \bh \bh^\top f_{t}(0) \right]  \int K'(u)^2 \mathrm{d}u$ is positive definite.
		
		%		\item 
		%		\label{assu:ConditionalDensity}
		%		The conditional distribution of $\varepsilon_t = X_t - Y_{t+1}$ given $\mathcal{F}_t$ is absolutely continuous with density $f_{t}(\cdot)$ which has the following properties:
		%		\color{red} (i) $c \ge f_{t}(0) > f_{t}(x)$ for all $x \in \mathbb{R}$ almost surely, \color{black}
		%		(ii) $f_{t}$ is three times continuously differentiable with bounded derivatives.
		
		\item 
		\label{assu:ConditionalDensity}
		For all $t \in \mathbb{N}$, the conditional distribution of $\varepsilon_t = X_t - Y_{t+1}$ given $\mathcal{F}_t$ is absolutely continuous with density $f_{t}(\cdot)$ which is three times continuously differentiable with bounded derivatives.
		
		\item 
		\label{assu:Kernel}
		$K: \mathbb{R} \to \mathbb{R}$, $u \mapsto K(u)$ is a non-negative and continuously differentiable kernel function such that:
		(i) $\int K(u) \mathrm{d}u = 1$,
		(ii) $\int u K(u) \mathrm{d}u = 0$,
		(iii) $\sup K(u) \le c <\infty$,
		(iv) $\sup K'(u) \le c <\infty$,
		(v) $\int u^2 K(u) \mathrm{d}u <\infty$,
		(vi) $\int \big| K'(u) \big| \mathrm{d}u  <\infty$,
		(vii) $\int u K' (u) \mathrm{d}u <  \infty$.
		
		\item 
		\label{assu:Bandwidth}
		$\delta_T$ is a strictly positive and deterministic sequence such that (i) $T \delta_T \to \infty $, and
		(ii) $T \delta_T^7 \to 0$.
		
	\end{enumerate}
\end{assumption}

The above assumptions are a  combination of standard assumptions from rationality testing and nonparametric statistics. \edit{Importantly, for our applications, these assumptions allow for heterogeneous predictive distributions, $F_t$.} Conditions \ref{assu:StationaryErgodicSequence} and \ref{assu:Moments} facilitate the use of a law of large numbers and of a central limit theorem for martingale difference arrays (MDA) (that generalize MD sequences to triangular arrays; see \cite{davidson1994stochastic}) and allows for both possibly non-stationary time-series and cross-sectional applications. Notice that the rate in the mixing condition \ref{assu:StationaryErgodicSequence}  is relatively weak as we only need it for a law of large numbers while we apply a central limit theorem for MDA in the proof of Theorem \ref{thm:ModeRationality} below.
%For non-stationary data, this can be replaced by the assumption of an $\alpha$-mixing process (see e.g.\ \cite{White2001}, Section 5.4)  by slightly strengthening the moment condition \ref{assu:Moments}.
In cross-sectional applications with independent observations, this assumption can be replaced (and weakened) by the classical Lindeberg condition (see e.g.\ \cite{White2001}, Section 5.2).
%Condition \ref{assu:Moments} is a standard moment assumption in time-series applications. 
Notice that as the kernel function $K'$ is bounded, we do not require existence of any moments of $Y_t$ or $X_t$, which makes this more flexible than rationality testing for mean forecasts.
The full rank condition \ref{assu:FullRank} prevents the instruments from being perfectly colinear which in turn prevents the asymptotic covariance matrix from being singular.
Condition \ref{assu:CovConvergence} guarantees that the asymptotic covariance matrix is well behaved for non-stationary data. Assumption \ref{assu:ConditionalDensity} assumes a relatively smooth behavior of the conditional density function which is required to apply a Taylor expansion common to the nonparametric literature. Conditions \ref{assu:Kernel} and \ref{assu:Bandwidth} are standard kernel and bandwidth conditions from the nonparametric literature. We discuss specific kernel and bandwidth choices in Supplemental Appendices \ref{sec:KernelChoice} and \ref{sec:BandwidthChoice} respectively.
		
\begin{theorem}
	\label{thm:ModeRationality}
	Under Assumption \ref{assu:ModeRationality} and $\mathbb{H}_0:  X_t = \mode( Y_{t+1} | \mathcal{F}_t ) ~ \forall~t$ a.s., it holds that
	\begin{align}
		\label{eqn:AsyNormModeRationality}
		\delta_T^{3/2}  T^{-1/2} \sum_{t=1}^T \psi(Y_{t+1},X_t,\bh, \delta_T) \tod \mathcal{N} \big( 0, \Omega_{\mathrm{Mode}} \big),
	\end{align}
%	\begin{align}
%		\label{eqn:AsyNormModeRationality}
%		\delta_T^{3/2}  T^{-1/2} \Omega^{-1/2}_{T, \mathrm{Mode}}  \sum_{t=1}^T \psi(Y_{t+1},X_t,\bh, \delta_T) \tod \mathcal{N} \big( 0, I_k \big),
%	\end{align}
 	 where $\Omega_{\mathrm{Mode}}$ is the limit of $\Omega_{T,\mathrm{Mode}} = \frac{1}{T} \sum_{t=1}^T \mathbb{E} \left[ \bh \bh^\top f_{t}(0) \right]  \int K'(u)^2 \mathrm{d}u$ as $T \to \infty$.
\end{theorem}
%The proof of Theorem \ref{thm:ModeRationality} is given in Appendix \ref{sec:Proofs}.

We obtain a test for the rationality of mode forecasts by drawing on the literature on nonparametric estimation. Unsurprisingly, therefore, the rate of convergence is slower than $\sqrt{T}$; Assumption \ref{assu:Bandwidth} requires $\delta_T \propto T^{-\kappa}$ with $\kappa \in (1/7, 1)$, which implies that the fastest convergence rate approaches $T^{2/7}$ and is obtained by setting $\delta_T \approx T^{-1/7}$.\footnote{Under additional assumptions, the speed of convergence of a nonparametric estimator may be increased via the use of higher-order kernel functions, see e.g.\ \cite{LiRacine2006}. However, as the generalized modal midpoint introduced in Definition \ref{def:GeneralizedModalMidpoint} requires a log-concave kernel to be well-defined and unique (see Theorem \ref{thm:GenModalMidpointProperties}), and as this assumption is automatically violated for higher-order kernels, we do not consider them here.} 

%\color{red} Our focus on one-step ahead forecasts allows for the application of a central limit theorem (CLT) for martingale difference sequences (MDAs), which only requires the existence of second moments of $\delta_T^{3/2} \psi(Y_{t+1},X_t,\bh, \delta_T)$; see the proof of Theorem \ref{thm:ModeRationality} for details. Multi-step ahead forecasts, on the other hand, are usually handled with CLTs for mixing sequences with a moment condition of order $r > 2$, which is not satisfied in our case as $\delta_T \to 0$. Accommodating multi-step ahead forecasts requires strengthening the dependence assumptions to either mixingales or $M$-dependent sequences \citep{White2001}, an extension that is left for future research. \color{black}

Theorem \ref{thm:GenModalMidpointProperties} guarantees the \textit{strictness} of the identification function $V_\delta(x,Y)$ only if $K$ has infinite support and is strictly increasing (decreasing) left (right) of its mode. These conditions are satisfied by the familiar Gaussian kernel, which we adopt for our analysis.\footnote{We also considered the relatively efficient quartic (or biweight) kernel but did not observe a change in power relative to the Gaussian kernel. See Supplemental Appendix \ref{sec:KernelChoice} for further discussion of the kernel choice and simulation results.}

Following \cite{Kemp2012} and \cite{Kemp2019}, we estimate the covariance matrix by its sample counterpart,
\begin{align}
	\widehat{\Omega}_{T,\mathrm{Mode}} = \frac{1}{T} \sum_{t=1}^T \delta_T^{-1} K' \left( \frac{X_t-Y_{t+1}}{\delta_T} \right)^2  \bh \bh^\top.
\end{align}
The following theorem shows consistency of the asymptotic covariance estimator without imposing our null hypothesis.
%  its proof is presented in Supplemental Appendix \ref{sec:TechnicalProofs}.
\begin{theorem}
	\label{thm:ConsistencyCovEstimation}
	Given that $X_t \in \mathcal{F}_t$, $\forall t \in \mathbb{N}$ and Assumption \ref{assu:ModeRationality}, %\edit{and $\mathbb{H}_0:  X_t = \mode( Y_{t+1} | \mathcal{F}_t ) ~ \forall~t$ a.s.,} 
	it holds that $\widehat{\Omega}_{T,\mathrm{Mode}} - \Omega_{T,\textrm{Mode}} \toP 0$. 
%	$\widehat{\Omega}_{T,\mathrm{Mode}} \toP \Omega_{\textrm{Mode}}$.
%
%	\begin{align}
%	\label{eqn:ConsistencyCovMatrix}
%	\widehat{\Omega}_{T,\mathrm{Mode}} \toP \Omega_{\mathrm{Mode}}.
%	\end{align}
\end{theorem}
% The proof of Theorem \ref{thm:ConsistencyCovEstimation} is given in Appendix \ref{sec:Proofs}.
We can now define the Wald test statistic:
\begin{align}
J_T =  	\left( \delta_T^{3/2}  T^{-1/2} \sum_{t=1}^T \psi \big( Y_{t+1},X_t,\bh, \delta_T \big) \right)^\top \widehat{\Omega}_{T,\mathrm{Mode}}^{-1} \left( \delta_T^{3/2}  T^{-1/2} \sum_{t=1}^T \psi \big( Y_{t+1},X_t,\bh, \delta_T \big) \right).
\end{align}
The following statement follows directly from Theorem \ref{thm:ModeRationality} and Theorem \ref{thm:ConsistencyCovEstimation}.
\begin{corollary}
	\label{cor:TestStatisticChiSquared}
	Under Assumption \ref{assu:ModeRationality} and the null hypothesis $\mathbb{H}_0:  X_t = \mode( Y_{t+1} | \mathcal{F}_t ) ~ \forall ~ t$ a.s., it holds that $J_T \tod \chi^2_k$.
%	\begin{align}
%	\label{eqn:TestStatisticChiSquared}
%	J_T \tod \chi^2_k.
%	\end{align}
\end{corollary}
This corollary justifies an asymptotic test at level $\alpha \in (0,1)$ which rejects $\mathbb{H}_0$ when \linebreak $J_T > Q_k(1-\alpha)$, where $Q_k(1-\alpha)$ denotes the $(1-\alpha)$ quantile of the $\chi^2_{k}$ distribution.\footnote{Our focus on one-step ahead forecasts allows for the application of a central limit theorem for MDAs, which only requires the existence of second moments. Multi-step ahead forecasts, on the other hand, are usually handled via CLTs for processes with more memory (e.g., mixing processes) at a cost of imposing stronger moment conditions. We leave this extension for future research.} Note that the bandwidth parameter, $\delta_T$, is introduced only to conduct the test of mode forecast rationality; the forecast itself, $X_t$, is, under the null, the true conditional mode of the target variable, not a (smoothed) modal midpoint.

%\td{Timo: The following red part is new, the blue part old. I will copy the red proof to the appendix in the end. I was in the mood to characterize the local behaviour of our test quite generally in the points (a)--(c) below. If we think this is too complicated and is bad for the reading flow, we can simplify the exposition in the main paper e.g., as follows: First, we could only discuss the rate of the local power and refer to my new theorem in the appendix. (Disadvantage: We would even take out a theorem on power, even though the referee wants us to put more emphasis on it.)
%	Second, we could focus on the point (a) in the main paper and only derive the \emph{real local} power in a neighborhood that shrinks with exactly the correct rate. (and maybe put the general theorem in the appendix...)}

We now turn to the behavior of our test statistic $J_T$ under a local alternative hypothesis,
%For this, we consider the %global alternative,
%alternative hypothesis:
% OLD ALTERNATIVE HYPOTHEISIS
%\begin{align}
%\label{eqn:GlobalAlternativeHypothesis}
%\mathbb{H}_A: | f'_t(0) | \ge \varepsilon\quad \text{a.s. for some } \varepsilon > 0 \quad \text{ and for all } t = 1, \dots, T.
%\end{align}
\begin{align}
	\label{eqn:LocalAlternativeHypothesis}
	\mathbb{H}_{A, \text{loc}}: \frac{1}{T} \sum_{t=1}^T f'_t(0) \bh  = c \cdot a_T + o_P(1),
\end{align}
%\td{This formulation of a local alternative in \eqref{eqn:LocalAlternativeHypothesis} is a little more general and allows for a stochastic than writing $\mathbb{E}[f'_t(0) \bh]  = c a_T$. I think it is also a little more realistic as it is unclear how $\mathbb{E}[f'_t(0) \bh]$ should depend on $T$ in our stationary setting... \\
%Next day thought: However,under stationarity, the distribution of $\frac{1}{T} \sum_{t=1}^T f'_t(0) \bh$ should also be independent of $T$, which makes the local power assumption always a little weird. But I went through the proofs and I think we could easily replace the stationarity/ergodicity by a mixing assumption (in Section 2, not sure about Section 3)... We just need to be a little more careful how we state results as Theorem 2.6., as $\Omega_\text{Mode}$ should then be $\Omega_{\text{Mode},T} =  \frac{1}{T} \sum_{t=1}^T \mathbb{E} \left[ \bh \bh^\top f_{t}(0) \right]  \int K'(u)^2 \mathrm{d}u$.} 
\edit{for some forecast $X_t \in \mathcal{F}_t$,} some constant $c \in \mathbb{R}^k$ and some possibly stochastic sequence $a_T$ to be specified in the following Theorem \ref{thm:ModeRationalityLocalAlternative}, which characterizes the behavior of our mode rationality test under the local alternative.
\begin{theorem}
	\label{thm:ModeRationalityLocalAlternative} 
%	If $T \delta_T^3 \to \infty$, then under Assumption \ref{assu:ModeRationality} and the alternative $\mathbb{H}_{A, \text{loc}}$ in \eqref{eqn:LocalAlternativeHypothesis}:
	\edit{Assume that $T \delta_T^3 \to \infty$, Assumption \ref{assu:ModeRationality} and the alternative $\mathbb{H}_{A, \text{loc}}$ in \eqref{eqn:LocalAlternativeHypothesis} hold for the forecasts $X_t \in \mathcal{F}_t$. Then:}
		
%	\begin{align}
%		J_T \tod \chi^2_k \big( c^\top \Omega_{\textrm{Mode}}^{-1} c \big), 
%	\end{align}		

	\begin{enumerate}[label=(\alph*)]
		\item 
		\label{enm:LocPower}
		If $a_T T^{1/2} \delta_T^{3/2} \toP 1$, then $J_T \tod \chi^2_k \big( c^\top \Omega_{\textrm{Mode}}^{-1} c \big)$, 	where $\chi^2_k \big( \tilde c \big)$ denotes a $\chi^2_k$ distribution with non-centrality parameter $\tilde c \in \mathbb{R}$.
		
		\item 
		\label{enm:LocPower_Null}
		If $a_T T^{1/2} \delta_T^{3/2} \toP 0$, then $J_T \tod \chi^2_k$.
		%which extends Corollary \ref{cor:TestStatisticChiSquared}.
		
		\item 
		\label{enm:LocPower_UnifPower}
		If $\big( a_T T^{1/2} \delta_T^{3/2} \big)^{-1} \toP 0$, then $\mathbb{P} \left(  J_T  \ge \bar c \right) \to 1$ for any  $\bar c > 0$, i.e., we have uniform power. 
	\end{enumerate}
\end{theorem}

This theorem shows that the power of our mode rationality test is essentially driven by the condition that $\frac{1}{T} \sum_{t=1}^T f'_t(0) \bh$ is non-zero, which implies that $X_t$ cannot be the mode of $F_t$.
It essentially means that  the instruments $\bh$ are correlated with the conditional density slope at zero, $f'_t(0)$. 
This is analogous to the case for standard mean and median rationality tests which have power against the alternative that $\mathbb{E} \big[V(X_t, Y_{t+1}) \bh] \not= 0$, where $V(X_t, Y_{t+1})$ is the identification function for the mean or median.\footnote{See e.g.\ Theorem 2 of \cite{giacomini2006tests}, where their Comment 6 and \cite{Nolde2017} point out that the theory can be directly adapted to rationality testing.}
When the instruments include a constant, a sufficient condition for a global alternative hypothesis, where $a_T$ is constant, is $f'_t(0) \ge c > 0$ ($f'_t(0) \le -c < 0$) for all $t \in \mathbb{N}$, i.e.\ when all forecasts are issued to the left (right) of the mode. This can be interpreted as uniform power against the class of strongly unimodal distributions. Our test will have low power to reject forecasts that are far in the tail, where the density is (almost) flat, however our test is intended for forecasts of a measure of central tendency, and such forecasts will lie broadly in the central region of the distribution where this concern does not arise.

Part \ref{enm:LocPower} of Theorem \ref{thm:ModeRationalityLocalAlternative} shows that if $\frac{1}{T} \sum_{t=1}^T f'_t(0) \bh$ converges at rate $T^{-1/2} \delta_T^{-3/2}$, then the asymptotic distribution of $J_T$ stabilizes as a non-central $\chi^2$-distribution, implying that for a fixed alternative, our test statistic diverges at rate $T^{1/2} \delta_T^{3/2}$, which is approximately  $T^{2/7}$ in practice.
In contrast, it is easy to show that rationality tests for identifiable functionals as the mean and median have local power against alternatives that converge with rate $T^{-1/2}$, which is of course faster than the rate $T^{-1/2} \delta_T^{-3/2}$ of the mode test . We demonstrate in the simulations and the applications that our mode test nevertheless has satisfactory power in practice in typically encountered sample sizes.

Part \ref{enm:LocPower_Null}  of Theorem \ref{thm:ModeRationalityLocalAlternative} further shows that  if $\frac{1}{T} \sum_{t=1}^T f'_t(0) \bh$ converges to zero \emph{faster} than $T^{-1/2} \delta_T^{-3/2}$, then our test behaves as under the null and has no power.
Finally, part \ref{enm:LocPower_UnifPower} implies that our test has unit power asymptotically when $\frac{1}{T} \sum_{t=1}^T f'_t(0) \bh$ converges to zero slower than $T^{-1/2} \delta_T^{-3/2}$, which nests the classical case of a fixed, global alternative.

\edit{
Following \cite{Kemp2012} and \cite{Kemp2019}, we choose $\delta_T$ proportional to $T^{-0.143}$, which is almost $T^{-1/7}$ as required by  Assumption \ref{assu:ModeRationality}. Specifically, we set \linebreak $\delta_T = k_1 \cdot k_2 \cdot T^{-0.143}$ where $k_1 = 2.4\widehat{\operatorname{Med}} \big( \big| \varepsilon_{1:T} - \widehat{\operatorname{Med}} [ \varepsilon_{1:T} ] \big| \big)$, $k_2 = \exp(-3 \left| \hat \gamma \right|)$, and $\hat \gamma =  3 \, \big( \frac{1}{T} \sum_t \varepsilon_t - \widehat{\operatorname{Med}} [\varepsilon_{1:T}] \big)/\widehat{\sigma}(\varepsilon_{1:T})$, where $\widehat{\sigma}$ and  $\widehat{\operatorname{Med}}$ denote the sample standard deviation and median. As in \cite{Kemp2012} and \cite{Kemp2019}, we choose $k_1$ proportional to the median absolute deviation of the forecast error, a robust measure of variation, and introduce a second constant, $k_2$, to adjust for the skewness of the forecast error, measured by the absolute value of Pearson's second skewness coefficient, $\hat \gamma$. We illustrate the good finite sample properties of our bandwidth choice through simulations in Supplemental Appendix \ref{sec:BandwidthChoice}.}

\section{Evaluating Forecasts of Measures of Central Tendency} 
\label{sec:FunctionalCentrality}

We define a class of measures of central tendency nesting the mean, median, and mode, and we propose tests of whether a forecast is rational with respect to \textit{any} element of the class.
Formally, we consider convex combinations of loss functions pertaining to the mean, median, and mode:  $L_{\mathrm{Mean}}(x,y) = (x-y)^2$, $L_{\mathrm{Med}}(x,y) = |x-y|$, $L_{\mathrm{Mode}, \delta}(x,y) =  -\delta^{1/2} K\left( (x-y)/\delta \right)$, for some kernel $K$.\footnote{Note that $L_{\mathrm{Mode}, \delta}$ differs from $L^K_{\delta}$ in equation (\ref{eqn:LossGenModalInt}) by a scaling factor of $\delta^{3/2}$. This arises from the convergence rate presented in Theorem \ref{thm:ModeRationality}.} 
Each  vector in the unit simplex, $\Theta := \{ \theta \in \mathbb{R}^3: ||\theta||_1 = 1, \theta \ge 0 \}$, \edit{together with a positive weight vector $\big(\WMean, \WMed, \WMode\big) \in \mathbb{R}_+^3$,} generates an optimal forecast 
% \footnote{\citet{fissler2022measurability} provide sufficient conditions for $X_t^*(\theta)$ being an $\mathcal{F}_t$-measurable random variable.} 
\begin{align}
	\label{eqn:CentralityAssumptionLoss}
	X_t^*(\theta) = \argmin_{X \in \mathcal{F}_t} ~ \mathbb{E}_t \big[  \edit{\WMean} L_{\mathrm{Mean}}( X, Y_{t+1}),   \edit{\WMed}  L_{\mathrm{Med}}( X, Y_{t+1}), \edit{\WMode}  L_{\mathrm{Mode}, \delta}( X, Y_{t+1})   \big]\cdot \theta \,
\end{align}
where we minimize over \edit{all $\mathcal{F}_t$-measurable potential forecasts $X$.}
%, where  $\mathcal{F}_t = \sigma \big\{ X_{s-1}, \widetilde{\mathbf{h}}_s; s \le t \big\}$} is the information set available to the forecaster at time $t$.
%\footnote{We minimize over all $\widetilde{\mathcal{F}}_t$-measurable instead of $\mathcal{F}_t$-measurable random variables in order 
%\footnote{The argmin over $\mathcal{F}_t$-measurable random variables $\tilde X_t$ is defined as in \citet[Theorem 1]{holzmann2014role}.}
%Hence, we analyze whether the forecast $X_t$ is (conditionally on $\mathcal{F}_t$) optimal with respect to the convex combination of loss functions given in equation \eqref{eqn:CentralityAssumptionLoss}. 
\edit{The scalar weights $\WMean$, $\WMed$, and $\WMode$ can be used to equalize the influence of the different loss functions, discussed further below.}
%We discuss the scalar weights $\WMean$, $\WMed$, and $\WMode$, and the factors $1/2$ for the mean and $\delta_T^{3/2}$ for mode below. 
At the vertices of $\Theta$, this nests the mean, median, and generalized modal midpoint, the latter tending to the mode as $\delta \to 0$; see Theorem \ref{thm:GenModalMidpointProperties}.

Assuming that a forecaster generates her forecasts $X_t = X_t^\ast(\theta_0)$ according to \eqref{eqn:CentralityAssumptionLoss}, we aim to determine the set of values for $\theta_0$ for which the given forecasts are optimal. Tests of forecast optimality are commonly conducted via an unconditional moment condition obtained by interacting an $\mathcal{F}_t$-measurable $(k \times 1)$ vector of instruments $\bh$ with the first order condition of the optimal forecast \citep{EKT2005, Nolde2017}. The latter also arise in the rationality tests introduced in equations  \eqref{eqn:WaldTestStatMean} and \eqref{eqn:PsiMode}. We continue in this tradition and test rationality across our class of central tendency measures by examining the variable $\phi_{t,T}(\theta)$, defined as:
\begin{equation}
	\phi_{t,T}(\theta) :=  
	\bh \; \theta^\top 
	\begin{pmatrix}
		\WMean \, V_{\mathrm{Mean}}( X_t, Y_{t+1}) \\
		\WMed \, V_{\mathrm{Med}}( X_t, Y_{t+1}) \\
		\WMode \, V_{\mathrm{Mode,\delta_T}}( X_t, Y_{t+1}) 
	\end{pmatrix}, \quad \text{for~~} \theta \in \Theta.
	\label{eqn:phi_t}
\end{equation}
The identification functions $V_{\mathrm{Mean}}$ and $V_{\mathrm{Med}}$ are given in equations \eqref{eq:MeanIF} and \eqref{eq:MedIF}, and $V_{\mathrm{Mode},\delta}( x,y) =  - \delta^{-1/2}  K' \left( (x-y)/\delta \right)$ is an appropriately scaled version of \eqref{eq:ModeIF} that guarantees asymptotic normality in Theorem \ref{thm:ModeRationality} and Theorem \ref{thm:GMMWeakId} below. \edit{Under forecast rationality, there must exist a $\theta_0$ (possibly set-valued) such that $\phi_{t,T}(\theta_0)$ is ``small'' on average. We formalize this statement in Assumption \ref{assu:GMMWeakIDRegCond} below.} Notice that $\phi_{t,T}$ is a triangular array, depending on $t$ and $T$, through its dependence on the bandwidth $\delta_T$.

\begin{remark}
	\label{rem:ConvexFunctionalCombinations}
	\edit{The forecast $X_t^*(\theta)$ from equation (\ref{eqn:CentralityAssumptionLoss}) is optimal with respect to a convex combination of loss functions.}
	Intuitively one might be inclined to consider a convex combination of the functional values rather than of the associated loss functions, however such functionals are generally neither elicitable nor identifiable (see Proposition \ref{prop:convexvalues} in the Supplemental Appendix), rendering infeasible an extension of the rationality tests introduced in Section \ref{sec:GeneralRatTests} to this case. However it can be shown (Proposition \ref{prop:BetaThetaMapping}) that any forecast following equation \eqref{eqn:CentralityAssumptionLoss} lies between the mean, median and generalized modal midpoint, though with combination weights for loss functions that generally differ from the combination weights for the functional values.
\end{remark} 
\begin{remark} 
	\edit{It is possible to consider a probabilistic mixture of mean, median, and mode type forecasts, discussed further in Supplemental Appendix \ref{sec:HeterogeneityForecasters}. Such forecasts are consistent with the forecast rationality moment condition we test if only a constant is used as instrument (Proposition \ref{prop:GammaThetaMapping}), but will generally be rejected for multivariate instruments.
	} 
\end{remark}

In \edit{equations \eqref{eqn:CentralityAssumptionLoss} and \eqref{eqn:phi_t}}, we allow for normalizations of each \edit{loss and} identification function using the scalar weights $\WMean$, $\WMed$ and $\WMode$. This allows one to adjust the importance of each loss function in order to construct tests that are robust to linear data transformations. To do so, we use the inverse of the standard deviations of the respective identification functions in our empirical analysis, though one could instead use equal weights, or some other choice. We define $\widehat \phi_{t,T}(\theta)$ as in equation \eqref{eqn:phi_t}, but using sample-dependent weights $\WMeanHat$, $\WMedHat$ and $\WModeHat$.

Consider the GMM objective function based on $\widehat \phi_{t,T}(\theta)$:
\begin{align}
\label{eqn:GMMWeakIdTestStatistic}
S_T(\theta) = \left[ T^{-1/2} \sum_{t=1}^T \widehat \phi_{t,T}(\theta) \right]^\top  \widehat{\Sigma}_T^{-1}( \theta) \left[  T^{-1/2} \sum_{t=1}^T \widehat \phi_{t,T}(\theta) \right],
\end{align}
where $\widehat{\Sigma}_T^{-1}(\theta)$ denotes an $O_P(1)$ positive definite weighting matrix, which may depend on the parameter $ \theta$. Unlike the problem in \cite{EKT2005}, the unknown parameter in our framework (the weight vector $\theta$) cannot be assumed to be well identified. For example, for symmetric distributions, the combination weights are completely unidentified. For distributions that exhibit only mild asymmetry a weak identification problem arises. For asymmetric distributions where one measure is a convex combination of the other two (a situation that arises naturally in location-scale processes) we have partial identification of the weight vector. %\footnote{Partial identification also arises when the three centrality measures can be strictly ordered, and the reported forecast is equal to the middle measure. This case arises in our simulation study in the next section.} 
The distribution of economic variables may or may not exhibit asymmetry, and so addressing this identification problem is a first-order concern.

The possibility that the true parameter $\theta_0$ is unidentified, partially identified, or weakly identified implies that the objective function $S_T(\theta)$ may be flat or almost flat in a neighborhood of $\theta_0$, ruling out consistent estimation of $\theta_0$. \cite{StockWright2000} show that, under regularity conditions, we can nevertheless construct asymptotically valid confidence bounds for $\theta_0$, by showing that the objective function $S_T$ evaluated at $\theta_0$ continues to exhibit an asymptotic $\chi^2$ distribution. This facilitates the construction of asymptotically valid confidence bounds even in a setting where the parameter vector may be strongly identified, weakly identified, or unidentified.\footnote{Alternative approaches to estimate the confidence sets under partial identification include \cite{Kleibergen2005}, \cite{CHT07}, \cite{BM08} and \cite{Chen2018} among others.}
We further impose the following regularity conditions on our process.
\begin{assumption}	
	\label{assu:GMMWeakIDRegCond}
	\textup{(A)} $\mathbb{E} \left[ |\varepsilon_t|^{2+\delta} \right] < \infty$ and 
	$\mathbb{E} \left[ ||\bh||^{2+\delta} |\varepsilon_t|^{2+\delta} \right] < \infty$,
	\textup{(B)}  $\WMeanHat \toP \WMean$, $\WMedHat \toP \WMed$, and $\WModeHat \toP \WMode$ for some positive weights $\WMean$, $ \WMed$ and $\WMode$;
%\end{assumption}
%
%\begin{assumption}
%	\label{assu:GMMWeakIDDecomposition}
	\textup{(C)} 
	The limit of $\Sigma_T(\theta_0)$ defined in \eqref{eq:SigmaT} is positive definite.
	\textup{(D)} \edit{For the forecasts $X_t \in \mathcal{F}_t$, $t\in \mathbb{N}$,} there exists a \edit{(not necessarily unique)} $\theta_0 \in \Theta$ and triangular arrays $\phi_{t,T}^\ast(\theta_0)$ and $u_{t,T}(\theta_0)$, such that 
	\begin{equation}
		\label{eqn:AsymptoticMDA} \phi_{t,T}(\theta_0) = \phi_{t,T}^\ast(\theta_0) +  u_{t,T}(\theta_0), \qquad \text{ where }
	\end{equation}
%	and
	\begin{enumerate}[label=(\alph*), leftmargin=1.0cm]
		\item 
		$\big\{ T^{-1/2} \phi_{t,T}^\ast(\theta_0), \mathcal{F}_{t+1} \big\}$ is a martingale difference array,
		
		\item 
		$T^{-1}  \sum_{t=1}^T || u_{t,T}(\theta_0) ||^2  \toP 0$, and 
		$\sum_{t=1}^T \mathbb{E} \left[ || T^{-1/2} u_{t,T}(\theta_0)||^{2+\delta} \right] \to 0$,
		
		\item 
		$T^{-1} \sum_{t=1}^T u_{t,T}(\theta_0)  \phi_{t,T} (\theta_0)^\top \toP 0$ and 	 
		$T^{-1} \sum_{t=1}^T \mathbb{E} \left[ u_{t,T}(\theta_0)  \phi_{t,T} (\theta_0)^\top \right] \to 0$.
	\end{enumerate}
\end{assumption}

%\todo[inline]{TD: What we do feels a little weird: The null hypothesis in (3.2) is actually nested in Assumption 3.2 (D) (loosely speaking, MDAs have zero mean). I wanted to mention this at some point, especially as we do not mention the null in Corollary 3.5 anymore but just (the stronger) Assumption 3.2. So I added the red part below. (Remember that some time ago, we decided to keep the null hypothesis for illustrative reasons even though it would be necessary mathematically.)}

% TD: SENTENCE BEFORE MAY 2023:
%While conditions (A), (B) and (C) are standard, a discussion of (D) is in order, which restricts the dependence structure of the sequence $\phi_{t,T} (\theta_0)$ such that a CLT can be applied.
\edit{Conditions (A), (B), and (C) are standard regularity conditions, and Assumption (D) represents the null hypothesis of forecast rationality that we test. Importantly, these conditions allow for heterogeneous predictive distributions, $F_t$. }%strengthens the null hypothesis in \eqref{eqn:WeakIDExpectationAssumption} by restricting the dependence structure of the array $\phi_{t,T} (\theta_0)$ such that a CLT can be applied.
\edit{The scaling with $T^{-1/2}$ in Assumption \ref{assu:GMMWeakIDRegCond} (D)(a) %and in the following discussion 
	matches the normalization in the CLT we employ (\citealp[Theorem 24.3]{davidson1994stochastic}).}
The decomposition in equation (\ref{eqn:AsymptoticMDA}) implies that the array $T^{-1/2}  \phi_{t,T}(\theta_0)$ is an \textit{approximate} MDA in the sense that $T^{-1/2} \phi_{t,T}(\theta_0)$ can be decomposed into a MDA $T^{-1/2} \phi^\ast_{t,T}(\theta_0)$ and some asymptotically vanishing array $T^{-1/2} u_{t,T}(\theta_0)$. This decomposition is required as the two standard assumptions---mixing and \textit{exact} MDA conditions---are too restrictive for our application.
First, the assumption that $\big\{ T^{-1/2}  \phi_t(\theta_0), \mathcal{F}_{t+1} \big\}$ is a MDA does not hold for the baseline case that $X_t$ is an optimal mode forecast (see the proof of Theorem \ref{thm:ModeRationality} for details). 
Second, imposing mixing conditions is too weak for our case as CLTs for mixing processes generally require finite moments of some order $r>2$, which is not fulfilled for the mode case as these moments diverge through the bandwidth parameter $\delta_T$.

%The assumption of an approximate MDA is weaker than the assumption that $\big\{ T^{-1/2}  \phi_t(\theta_0), \mathcal{F}_{t+1} \big\}$ is a MDA.
%Such an \textit{exact} MDA assumption 

%\begin{align}
%	\label{eqn:WeakIDMDAAssumption}
%	 \exists ~ \theta_0 \in \Theta \text{ s.t.}~ \big\{ T^{-1/2}  \phi_t(\theta_0), \mathcal{F}_{t+1} \big\} \text{ is a MDA,}
%\end{align}
%but stricter than the null hypothesis specified in \eqref{eqn:WeakIDExpectationAssumption} as it additionally imposes (asymptotic) conditions on its dependence structure.

%An exact MDA assumption as in (\ref{eqn:WeakIDMDAAssumption}) does not hold for the baseline case that $X_t$ is an optimal mode forecast, as in that case the MDA assumption only holds asymptotically (see the proof of Theorem \ref{thm:ModeRationality} for details). 
%Strengthening the unconditional (asymptotic) moment condition in (\ref{eqn:WeakIDExpectationAssumption}) by imposing mixing conditions is too weak for our case as a CLT for mixing processes requires that moments of order $r>2$ are finite, which is not fulfilled for the mode case as these moments diverge through the bandwidth parameter $\delta_T$.

%Assumption \ref{assu:GMMWeakIDRegCond} (C) 
% is an intermediate case of classically-imposed conditions  (\ref{eqn:WeakIDExpectationAssumption}) and (\ref{eqn:WeakIDMDAAssumption}) and 

The intermediate case of Assumption \ref{assu:GMMWeakIDRegCond} (D) allows to apply a CLT based on finite second moments only, and 
can easily be shown to hold at the three vertices, where the forecast is the mean, median or mode. Specifically, when $X_t$ is a mean or median forecast (i.e.\ $\theta_0 = (1,0,0)$ or $\theta_0 = (0,1,0)$), set $u_{t,T}(\theta_0) = 0$ and $\big\{ T^{-1/2} \phi_{t,T}(\theta_0), \mathcal{F}_{t+1} \big\}$ is obviously a MDA. When $X_t$ is the true conditional mode of $Y_{t+1}$, (i.e.\ $\theta_0 = (0,0,1)$), set
\begin{align}
	\label{eqn:u_Mode}
	u_{t,T}(\theta_0) 
	=  \mathbb{E}_t \left[ \phi_{t,T}(\theta_0) \right] 
	= - \omega_\text{Mode} \delta_T^{-1/2} \mathbb{E}_t \left[ K' \left( \frac{\varepsilon_t}{\delta_T} \right) \right] \bh.
\end{align}
Then, the conditions on $u_{t,T}(\theta_0)$ in Assumption \ref{assu:GMMWeakIDRegCond} (D) are fulfilled as shown in Lemma \ref{lem:Ass_TrueMode} in the Supplemental Appendix.
%by the arguments given in the proof of Theorem \ref{thm:ModeRationality} and Lemma \ref{lemma:ExpectationPsiTo0} -- Lemma \ref{lemma:MaxConvZero} in the Supplemental Appendix by setting $T^{-1/2} u_{t,T}(\theta_0) = g_{t,T}^e$, where the latter is defined in the proof of Theorem \ref{thm:ModeRationality}.

%\color{blue}
%OLD: 
%as in the proof of Theorem \ref{thm:ModeRationality}.
%The conditions on $u_{t,T}(\theta_0)$ in Assumption \ref{assu:GMMWeakIDRegCond} (C) are fulfilled by the arguments given in the proof of Theorem \ref{thm:ModeRationality} and Lemma \ref{lemma:ExpectationPsiTo0} - Lemma \ref{lemma:MaxConvZero} in the Supplemental Appendix. 
%\color{black}

When $X_t$ is a convex combination of a mean and median forecast, i.e., $\theta_0 = (\xi,1-\xi,0)$ for some $\xi \in [0,1]$, we set $u_{t,T}(\theta_0) = 0$ and $\big\{T^{-1/2} \phi_{t,T}(\theta_0), \mathcal{F}_{t+1} \big\}$ is again a MDA. When $X_t$ is a convex combination with non-zero weight on the mode, Assumption \ref{assu:GMMWeakIDRegCond} (D) is difficult to verify. 

Theorem \ref{thm:GMMWeakId} below presents the asymptotic distribution of the process $T^{-1/2}  \sum_{t=1}^T \widehat \phi_{t,T}(\theta_0)$ at the true parameter $\theta_0$. 
%The proof is given in Appendix \ref{sec:Proofs}.
\begin{theorem}
	\label{thm:GMMWeakId} 
	\edit{Under Assumptions \ref{assu:ModeRationality} and \ref{assu:GMMWeakIDRegCond},}
		% and the null hypothesis, $\mathbb{H}_0: \exists ~ \theta_0 \in \Theta$ such that $\lim_{T \to \infty}~ T^{-1/2} \mathbb{E} \left[ \phi_{t,T} (\theta_0)  \right] = 0$ for all $t \le T$, 
		it holds that
	\begin{align}
		T^{-1/2}  \sum_{t=1}^T \widehat \phi_{t,T}(\theta_0) \tod \mathcal{N} \big( 0, \Sigma(\theta_0) \big),
	\end{align}
	where $\Sigma(\theta_0)$ is the limit as $T \to \infty$ of
	\begin{align}
		\begin{aligned}
		\label{eq:SigmaT}
		\Sigma_T(\theta_0)
		&:= \frac{1}{T} \sum_{t=1}^T \mathbb{E} \left[ \theta_{10}^2 \WMean \bh \bh^\top \WMean \varepsilon_t^2 + \theta_{20}^2 \WMed \bh \bh^\top \WMed \big( \mathds{1}_{\{\varepsilon_t > 0\}} - \mathds{1}_{\{\varepsilon_t < 0 \}} \big)^2 \right. \\
		&\qquad \qquad \qquad + \theta_{30}^2 \WMode \bh \bh^\top \WMode f_t(0) \int K'(u)^2 \mathrm{d}u  \\
		&\qquad \qquad \qquad \left.  + 2 \theta_{10} \theta_{20} \WMean  \bh \bh^\top \WMed
		\varepsilon_t \big( \mathds{1}_{\{\varepsilon_t > 0\}} - \mathds{1}_{\{\varepsilon_t < 0 \}} \big) \right].
		\end{aligned}
	\end{align}
\end{theorem}

Under the null hypothesis, Assumption \ref{assu:GMMWeakIDRegCond} (D) implies that $T^{-1/2} \phi_{t,T}(\theta_0)$ (and hence also $T^{-1/2}  \widehat \phi_{t,T}(\theta_0) $) is an approximate MDA, i.e.\ this array is approximately  (as $T \to \infty$) uncorrelated. Consequently, we do not need to rely on HAC covariance estimation, and can instead estimate the asymptotic covariance matrix using the simple \edit{sample covariance matrix $\widehat{\Sigma}_T(\theta) = \frac{1}{T} \sum_{t=1}^T \widehat \phi_{t,T}(\theta) \widehat \phi_{t,T}(\theta)^\top$.}
%\vspace{-0.2cm}
%\begin{align}\label{eq:covest}
%\widehat{\Sigma}_T(\theta) = \frac{1}{T} \sum_{t=1}^T \widehat \phi_{t,T}(\theta) \widehat \phi_{t,T}(\theta)^\top.
%\end{align}
The next theorem shows consistency of the outer product covariance estimator.
%  its proof is presented in Supplemental Appendix \ref{sec:TechnicalProofs}.
\begin{theorem}
	\label{thm:GMMWeakIDCovarianceConsistency}
	Given Assumptions \ref{assu:ModeRationality} and \ref{assu:GMMWeakIDRegCond}, it holds that $\widehat{\Sigma}_T(\theta_0) - \Sigma_T(\theta_0) \toP 0$.
%	{\color{red} $\widehat{\Sigma}_T(\theta_0) \toP \Sigma(\theta_0)$.}
\end{theorem}

\begin{corollary}
	\label{corr:GMMWeakID}
	Given Assumptions \ref{assu:ModeRationality} and \ref{assu:GMMWeakIDRegCond}, it holds that $S_T(\theta_0) \tod \chi^2_k$.
\end{corollary}

Following \cite{StockWright2000}, this corollary allows one to construct asymptotically valid confidence regions for $\theta_0$ with coverage probability $(1-\alpha)\%$ by considering the set
\begin{align}
	\label{eqn:GMMWeakIdConfidenceRegion}
	\big\{ \theta \in \Theta:  S_T(\theta) \le Q_k(1-\alpha) \big\},
\end{align}
where $Q_k(1-\alpha)$ denotes the $(1-\alpha)$ quantile of the $\chi^2_{k}$ distribution.

Given the above results, we obtain a test for forecast rationality for a general measure of central tendency by evaluating the GMM objective function using a dense grid of convex combination parameters $\theta_j \in \Theta$ for $j = 1,\dots,J$. An asymptotically valid confidence set is given by the values of $\theta_j$ for which $S_T(\theta_j) \le Q_k(1-\alpha)$. These values represent the centrality measures  that ``rationalize'' the observed sequence of forecasts and realizations, in that rationality cannot be rejected for these measures of centrality. It is possible that the confidence set is empty, in which case we reject rationality at the $\alpha$ significance level for the \textit{entire class} of general centrality measures.

%
%However, the problem we encounter with this assumption is again very similar to the problem we already have for multi-step ahead forecasts. 
%That is, so far I cannot show asymptotic normality for the mode identification function without having an (approximate) martingale difference sequence.
%The reason is that then, a mixing CLT requires finite moments of order $2+\delta$ (or of order $r > 2$), which are not finite in our case.

%A second idea is to make the following stronger assumption.
%\begin{assumption}
%	\label{ass:GMMAssInfeasible2}
%	There exists some $\theta_0 \in \Theta$, such that $\left\{ \tilde \phi_t(\theta_0), \mathcal{F}_{t+1}  \right\}$ is a martingale difference sequence.
%\end{assumption}
%Assumption \ref{ass:GMMAssInfeasible2} clearly implies Assumption \ref{ass:GMMAssInfeasible1}.
%However, if we closely look at the proof of Theorem \ref{thm:ModeRationality}, we can see that we actually don`t have a MDA but rather an approximate MDA in the sense that we split our (asymptotic) identification function in a MDA and a sequence which converges to zero in probability and thus, this sequence does not alter the asymptotic normality result of the MDA.
%Consequently, true conditional mode forecasts do not satisfy Assumption \ref{ass:GMMAssInfeasible2}.
%
%While Assumption \ref{ass:GMMAssInfeasible1} seems to be too weak,  Assumption \ref{ass:GMMAssInfeasible2} seems to be too strong.
%Thus, we are hunting for something \textit{in between}.

The power of the rationality test depends on the instrument choice. This property is shared with many other tests in the literature \citep[e.g.,][]{EKT2005, PattonTimmermann2007,schmidt2021interpretation}. Good instruments span the information set of the forecaster and are not strongly correlated with each other. Additional instruments generally  improve power asymptotically (or shrink the identified set, in the partial identification case), but can deteriorate power in finite samples. An informative instrument, and one used as far back as \cite{MincerZarnowitz1969}, is the forecast $X_t$ itself. In our simulations and applications we found the forecast to be a powerful instrument. Additionally, it is guaranteed to be in the information set of the forecaster, and so is a valid instrument.

The forecast evaluation problem and approach considered here is related to, but distinct from, \cite{EKT2005}. These authors consider the case that a respondent's point forecast corresponds to some quantile (or expectile) of her predictive distribution. They employ a parametric loss function (``lin-lin'' for quantiles, ``quad-quad'' for expectiles), $L(X,Y;\tau)$ with a scalar unknown parameter ($\tau$) characterizing the asymmetry of the loss. \cite{EKT2005} use GMM to estimate the $\tau$ that best describes the sequence of forecasts and realizations, and test whether forecast rationality holds at the estimated value for $\tau$. Economically, our approach differs from \cite{EKT2005} in that we consider forecasts only as measures of centrality, allowing for a wide range of such measures, while that paper considers only a single centrality measure nested within a wide range of asymmetric forecasts. Statistically, our approach differs as we are forced to address the feature that our parameter may be partially, weakly, or un-identified, which precludes point estimation. 

We use convex combinations of identification functions to parametrically nest the mean, the median, and (asymptotically) the mode. Alternative approaches are possible, e.g., via the $L^p$ norm. Our convex combination approach has the advantage of separating the parameter of interest $\theta$ from the bandwidth parameter $\delta_T$, and further does not suffer from technical difficulties such as bandwidth parameters in the exponent, or identification functions with singularities at the mode.

\section{\edit{Simulation Study}}
\label{sec:SimulationStudy}

%This section studies the finite-sample performance of the methods proposed above.
%\Cref{sec:SimModeRat} presents simulations for the mode rationality test and  \Cref{sec:SimCentralityMeasures} analyzes the test for rationality for general measures of central tendency.

\subsection{Rationality tests for mode forecasts}
\label{sec:SimModeRat}

To evaluate the finite-sample properties of our mode rationality test, we simulate data from a cross-sectional and a time series AR(1)-GARCH(1,1) data generating process (DGP), given by
%which are nested in the following unified framework:
% \vspace{-0.2cm}
\begin{align}
\label{eqn:GeneralDGP}
Y_{t+1} =  Z_t^\top \zeta + \sigma_{t+1} \xi_{t+1}, \quad \text{ where }  \quad
\xi_{t+1} \stackrel{iid}{\sim} \; \mathcal{SN}(0,1 ,\gamma),
\end{align}
where $\mathcal{SN}(0,1,\gamma)$ is a skewed standard Normal distribution, $Z_t$ denotes a vector of covariates (possibly including lagged values of $Y_{t+1}$), $\zeta$ denotes a parameter vector and $\sigma_{t+1}$ represents a conditional variance process.
Using the general formulation in \eqref{eqn:GeneralDGP}, the cases we consider are:\footnote{\label{fn:DGPs}In the Supplemental Appendix, we also present data for a heteroskedastic cross-sectional DGP that is as in case (1), but with $\sigma_{t+1} = 0.5 + 1.5(t + 1)/T$, and a homoskedastic AR(1) process with $Z_t = Y_{t}$, $\zeta = 0.5$ and $\sigma_{t+1} = 1$.}
\begin{eqnarray}
	\text{Cross-sectional (iid):~~~} Z_t &\stackrel{iid}{\sim}& \mathcal{N} \big( ~ (1,1,-1,2), ~  \operatorname{diag}(0,1,1,0.1) ~  \big), \zeta = \boldsymbol{\iota}, \sigma_{t+1} = 1 \label{eq:DGPiid} \\ 
	\text{AR(1)-GARCH(1,1):~~}  Z_t &=& Y_{t}, \zeta = 0.5, \sigma_{t+1}^2 = 0.1 + 0.85 \sigma_{t}^2 + 0.1 \sigma_{t}^2 \xi_{t}^2, \label{eq:DGParGARCH}
\end{eqnarray}
where $\boldsymbol{\iota}$ is a vector of ones. These DGPs are based on a skewed Gaussian residual distribution with skewness parameter $\gamma$. This choice nests the case of a standard Gaussian distribution at  $\gamma = 0$, and in this case all measures of centrality coincide. As the skewness parameter grows in magnitude, the measures of centrality differ increasingly.

%\vspace{2cm}
%
%We simulate data from the following AR(1)-GARCH(1,1) DGP,
%\begin{align}
%	Y_{t+1} &= 0.5 Y_{t} + \sigma_{t+1} \varepsilon_{t+1}, \\
%	\sigma_{t+1}^2 &= 0.1 + 0.8 \sigma_{t}^2 + 0.1 \sigma_{t}^2 \varepsilon_{t}^2, \\
%	\varepsilon_t &\stackrel{iid}{\sim} \; \mathcal{SN}(0,1,\gamma),
%\end{align}
%where $\mathcal{SN}(0,1,\gamma)$ denotes the skewed normal distribution with mean zero, unit variance and skewness parameter $\gamma \in (-1,1)$.
%We choose to work with this DGP as an AR-GARCH model replicates the empirical properties of economic data realistically in many instances.
%We employ the skewed normal distribution for the residuals in order to generate a setting with asymmetric conditional distributions in order to distinguish between the mean, median and the mode.
For the DGP in (\ref{eqn:GeneralDGP}), optimal mode forecasts are given by
$X_t = {\mode}(Y_{t+1}| \mathcal{F}_t) = Z_{t}^\top \zeta + \sigma_{t+1} \mode ( \xi_t )$,
%\begin{align}
%\label{eqn:DGPOptimalModeForecasts}
%X_t = {\mode}(Y_{t+1}| \mathcal{F}_t) = Z_{t}^\top \zeta + \sigma_{t+1} \mode ( \xi_t ),
%\end{align}
where $\mode ( \xi_t )$ depends on the skewness parameter $\gamma$.
We use a range of values $\gamma \in \{ 0, 0.1, 0.25, 0.5\}$ and sample sizes $T \in \{ 100, 500, 2000, 5000 \}$.
To evaluate the size of our test in finite samples, we generate optimal mode forecasts and apply the mode forecast rationality test based on the instrument choices $\bh = 1$ and $\bh = (1,X_t)$.
% $\mathbf{h}_{t,1} = 1$ and $\mathbf{h}_{t,2} = (1,X_t)$.
All simulation results are obtained with $2,000$  simulation replications.

\begin{table}[tb]
	\centering
	\scriptsize
	\resizebox{\columnwidth}{!}{
		\begin{tabular}{l lrrrr lrrrr lrrrr lrrrr lrrrr}
			\toprule 
			& & \multicolumn{9}{c}{(1) iid DGP}  & & \multicolumn{9}{c}{(2) AR-GARCH DGP}  \\
			\cmidrule(lr){3-11} \cmidrule(lr){13-21} 
			%		\cmidrule(lr){2-6} \cmidrule(lr){7-11} \cmidrule(lr){12-16}
			&&  \multicolumn{4}{c}{Instrument $\bh =  1$} & & \multicolumn{4}{c}{Instrument $\bh =  (1,X_t)$}   & & \multicolumn{4}{c}{Instrument $\bh =  1$} & & \multicolumn{4}{c}{Instrument $\bh =  (1,X_t)$} \\
			\cmidrule(lr){3-6} \cmidrule(lr){8-11} \cmidrule(lr){13-16} \cmidrule(lr){18-21}
			%		\cmidrule(lr){2-16}
			Skewness & & 0 & 0.1 & 0.25 & 0.5 & & 0 & 0.1 & 0.25 & 0.5  && 0 & 0.1 & 0.25 & 0.5 && 0 & 0.1 & 0.25 & 0.5\\ 
			\midrule 
			$100$ &   & 3.9 & 5.4 & 6.6 & 8.2 &   & 4.9 & 5.4 & 6.2 & 7.4 &   &  4.7 &  6.2 &  6.6 & 10.3 &   & 5.9 & 5.2 & 6.6 & 8.6\\
			$500$ &   & 5.2 & 6.3 & 8.6 & 8.3 &   & 5.7 & 5.6 & 8.2 & 7.2 &   &  5.1 &  6.2 &  8.0 &  8.2 &   & 4.9 & 6.0 & 6.8 & 7.3\\
			$2000$ &   & 5.8 & 6.5 & 9.2 & 7.2 &   & 6.2 & 6.5 & 7.6 & 6.3 &   &  6.0 &  5.0 &  7.1 &  7.0 &   & 5.1 & 4.9 & 7.0 & 6.2\\
			$5000$ &   & 4.7 & 5.8 & 7.4 & 7.0 &   & 5.2 & 5.5 & 6.8 & 6.9 &   &  5.8 &  5.4 &  8.5 &  5.9 &   & 4.9 & 4.5 & 6.4 & 5.7 \\		
			\bottomrule 
		\end{tabular}
	}
	\caption{\textbf{Size of the mode rationality test.} This table presents the empirical rejection rates (in percent) of the mode rationality test using four sample sizes, four levels of skewness in the residual distribution, different choices of instruments, and two DGPs described in equation (\ref{eqn:GeneralDGP}). The nominal test level is $5\%$.}
	\label{tab:SizeGaussian5}
\end{table}

Table \ref{tab:SizeGaussian5} presents the finite-sample sizes of the test under the different DGPs, sample sizes, instrument choices and skewness parameters. In all cases we use a Gaussian kernel and set the nominal size to $5\%$.\footnote{In the Supplemental Appendix we present similar results for different kernel choices and significance levels.}
% \textcolor{red}{Results for nominal test sizes of $1\%$ and $10\%$ are given in \Cref{tab:SizeGaussian1} and \Cref{tab:SizeGaussian10} in the Supplemental Appendix.}
We find that our mode rationality test leads to finite-sample rejection rates that are generally close to the nominal test size, across all of the different choices of DGPs, instruments, sample sizes and skewness parameters. 
Table \ref{tab:SizeGaussian5} reveals that an increasing degree of skewness in the underlying conditional distribution negatively influences the test's performance. 
This can be explained by the increasing difference between the mode and the modal midpoint (which is elicited by the corresponding identification function in finite samples) for an increasing skewness in the data.
As a consequence,  we choose a smaller bandwidth following the rule of thumb described at the end of Section \ref{sec:ForecastRatioanlityTestMode}, resulting in less efficient estimates. 
Consequently, for highly skewed distributions, the mode rationality test requires larger sample sizes in order to converge to the nominal test size.

%Supplemental Appendix \ref{sec:Power} presents simulation results for two studies of the finite-sample power of this test, confirming that the test has non-trivial power in reasonable SCEnarios.

%To analyze the power of the mode forecast rationality test introduced in Section \ref{sec:ForecastRatioanlityTestMode} of the main paper we use the same DGPs as in Section \ref{sec:SimModeRat} and consider two forms of sub-optimal forecasts:

\begin{figure}[tb]
	\centering
	\includegraphics[width=\linewidth]{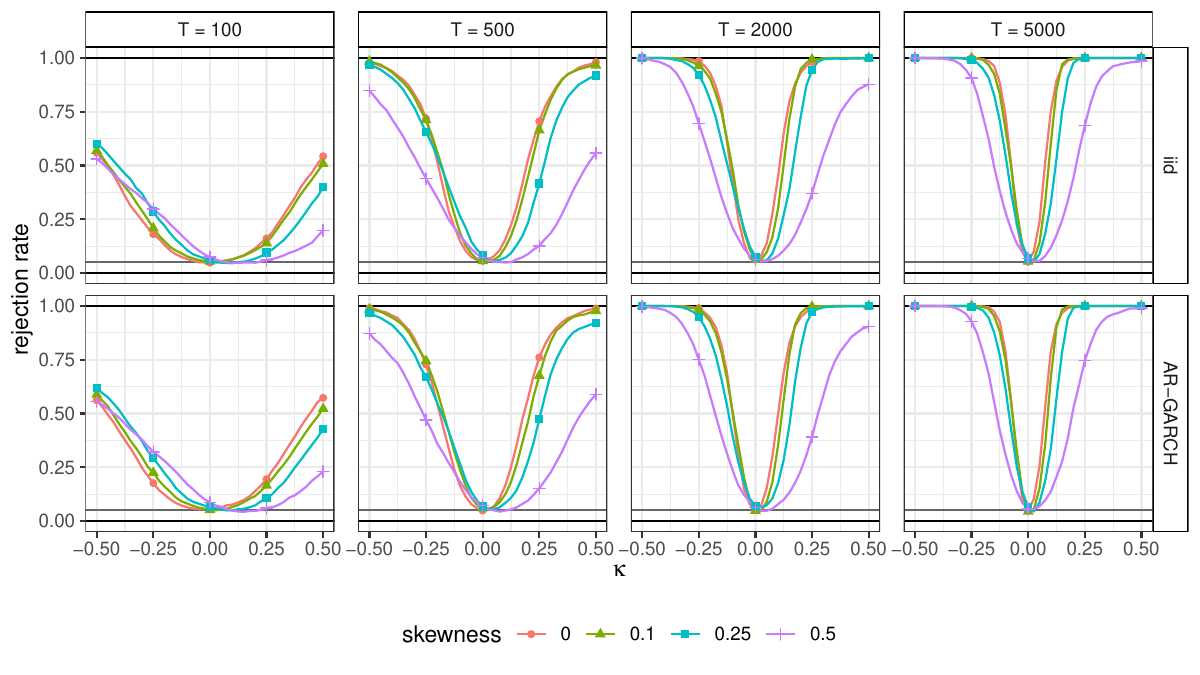}
	\caption{\textbf{Power of the mode rationality test.} This figure plots the empirical rejection frequencies against the degrees of misspecification $\kappa$ for different sample sizes in the vertical panels and for the two DGPs in the horizontal panels. The misspecification follows the bias design described in the main text and we use the instrument vector $(1,X_t)$ and a nominal significance level of $5\%$.}
	\label{fig:PowerBias}
\end{figure}

%\begin{figure}[tb]
%	\centering
%	\includegraphics[width=0.8\linewidth]{output/ModeRationality/power/sim_plots_noise_(1,x)_5.pdf}
%	\caption{\textbf{Power for the ``noise'' simulation.} This figure plots the empirical rejection frequencies against the degrees of misspecification $\kappa$ for different sample sizes in the vertical panels and for the four DGPs in the horizontal panels. The misspecification follows the design in equation (\ref{eq:noise}) and we utilize the instrument vector $(1,X_t)$ and a nominal significance level of $5\%$.}
%	\label{fig:PowerNoise}
%\end{figure}

To analyze the power of the mode forecast rationality test we use the DGPs from equations \eqref{eq:DGPiid}-\eqref{eq:DGParGARCH} and introduce a bias to the forecasts $\tilde X_t = X_t + \kappa \varsigma$,  where $\varsigma = \sqrt{\Var(\varepsilon_{1:T})}$ and $\kappa \in (-0.5, 0.5)$.\footnote{\label{fn:NoiseMisspec}We also consider a misspecification where we add noise as $\tilde X_t = X_t + \mathcal{N}(0, \kappa \varsigma^2)$ in Supplemental Appendix~\ref{sec:AdditionalPlotsTables}.}
% \vspace{-0.2cm}
%\begin{align}
%	\tilde X_t = X_t + \kappa \sigma_X, \text{~~~~~~~~where~~} \sigma_X = \sqrt{\Var(X_t)} \text{~and~} \kappa \in (-1,1) \label{eq:bias} 
%%	\textbf{Noise:~~~} \tilde X_t &=& X_t + \mathcal{N}(0, \kappa \sigma_X^2), \text{~~where~~} \sigma_X = \sqrt{\Var(X_t)} \text{~and~} \kappa \in (0,1) \label{eq:noise}
%\end{align}
This type of misspecification introduces a deterministic bias, where the degree of misspecification depends on the parameter $\kappa$.
%The second type of misspecification introduces independent noise, and the magnitude of the noise is regulated through the parameter $\kappa$: for $\kappa = 1$, the signal-to-noise ratio is one,
%% as the standard deviation of the signal equals the standard deviation of the independent noise, 
%and as $\kappa$ shrinks to zero the noise vanishes.

\Cref{fig:PowerBias} presents power plots for these ``biased'' forecasts for a range of sample sizes, skewness parameters and the two DGPs in equations (\ref{eq:DGPiid})-(\ref{eq:DGParGARCH}), where we plot the rejection rate against the degree of misspecification $\kappa$. 
For all plots, we use the instrument choice $\mathbf{h}_{t} = (1,X_t)$, a Gaussian kernel, and a nominal level of $5\%$. Notice that for $\kappa = 0$ the figures show the empirical test size. 

Figure \ref{fig:PowerBias} reveals that the proposed mode rationality test exhibits, as expected, increasing power for an increasing degree of misspecification. Also as expected, larger sample sizes lead to tests with greater power, although even the two smaller sample sizes exhibit reasonable power. The figure also reveals that an increasing degree of skewness yields to a slight loss of power (given a fixed degree of misspecification). This is driven through the bandwidth choice, where larger values of the (empirical) skewness result in a smaller bandwidth, and consequently a lower test power (analogous to the bias-variance trade-off in the nonparametric estimation literature). 
We illustrate the good behavior of our finite sample bandwidth choice through further simulations in the Supplemental Appendix \ref{sec:BandwidthChoice}.

Overall, our simulation results show that even though the mode rationality test converges at a slower-than-parametric rate (approximately $T^{2/7}$), it nevertheless has considerable power for the empirically relevant sample sizes of $T \ge 500$; as e.g., in our application in Section \ref{sec:ApplicationSCE}.
% @TD: I shortened this summary sentence, and addede the approximate rate, but kept the main idea. I agree this is a good thing to have here, given the referes' questions about power. -AJP

%Results corresponding to the ``biased'' forecast case when using a biweight kernel are presented in Figure \ref{fig:PowerBiasKernels} in Supplemental Appendix \ref{sec:AdditionalPlotsTables}. That figure reveals that the finite-sample size and power are very similar to those for the Gaussian kernel.

\subsection{Rationality tests for an unknown measure of central tendency}
\label{sec:SimCentralityMeasures}

We now examine the small sample behavior of the asymptotic confidence sets for the measures of central tendency, described in \Cref{sec:FunctionalCentrality}.
As in the previous section, we consider the two DGPs described in and after equation \eqref{eqn:GeneralDGP} and the same varying sample sizes $T$, skewness parameters $\gamma$, and instruments $\bh$.
%For this, we consider the AR-GARCH DGP,	
%\begin{align}
%	Y_{t+1} &= 0.5 Y_{t} + \sigma_{t+1} \varepsilon_{t+1}, \\
%	\sigma_{t+1}^2 &= 0.1 + 0.8 \sigma_{t}^2 + 0.1 \sigma_{t}^2 \varepsilon_{t}^2, \\
%	\varepsilon_t &\stackrel{iid}{\sim} \; \mathcal{SN}(0,1,\gamma),
%\end{align}
%
We generate optimal one-step ahead forecasts for the mean, median and mode as the true conditional mean, median and mode of $Y_{t+1}$ given $\mathcal{F}_t$:
%\begin{eqnarray}
%X_t^{\textrm{Mean}} &=& \zeta^\top Z_t + \sigma_{t+1} \operatorname{Mean}(\xi_{t+1}), \label{meanForecast} \\ 
%X_t^{\textrm{Med}} &=& \zeta^\top Z_t + \sigma_{t+1} \operatorname{Median}(\xi_{t+1}), \label{medianForecast} \\
%X_t^{\textrm{Mode}} &=& \zeta^\top Z_t + \sigma_{t+1} \operatorname{Mode}(\xi_{t+1}) \label{modeForecast}.
%\end{eqnarray}
\begin{align}
	\label{CentralityForecast}
	X_t^{\textsf{M}} = \zeta^\top Z_t + \sigma_{t+1}  \, \mathsf{M}(\xi_{t+1}), 
\end{align}
for $\mathsf{M} \in \{\textrm{Mean}, \textrm{Median}, \textrm{Mode}\}$, corresponding to the three vertices of $\Theta$.

For any other $\theta \in \Theta$, we generate the forecast $X_t^*(\theta)$ that arises as the optimal forecast under a convex combination of loss functions as in \eqref{eqn:CentralityAssumptionLoss} through a numerical optimization procedure.
For the weights in \eqref{eqn:CentralityAssumptionLoss}, we use the reciprocal of the estimated standard deviations of the respective loss functions over all $t=1,\dots,T$.
In detail, for a given $\theta \in \Theta$, and for each $t =1,\dots,T$, we draw 5000 replications from the residuals $\xi_t$ in \eqref{eqn:GeneralDGP} and obtain $X_t^*(\theta)$ in \eqref{eqn:CentralityAssumptionLoss} through a numerical optimization, where we approximate the expected value  in \eqref{eqn:CentralityAssumptionLoss} by the empirical average over the 5000 replications.
We repeat this process three times to stabilize the effect of the estimated weights through the inverted standard deviations, akin to the \emph{iterated} generalized method of moments.

%\begin{table}
%	\footnotesize
%	\centering
%	\begin{tabular}{lc}
%		\toprule 
%		Name & $\theta$ \\
%		\midrule 
%		True Mean Forecasts & $(1,0,0)$ \\
%		True Median Forecasts & $(0,1,0)$ \\
%		True Mode Forecasts  & $(0,0,1)$ \\
%		Loss-combined Mean-Mode Forecasts  & $(1/2,0,1/2)$ \\
%		Loss-combined Median-Mode Forecasts  & $(0,1/2,1/2)$ \\
%		Loss-combined Mean-Median-Mode Forecasts  & $(1/3,1/3,1/3)$ \\
%		\bottomrule
%	\end{tabular}
%	\caption{True loss-combination names and weight vectors.}
%	\label{tab:LossCombinations}
%\end{table}

Table \ref{tab:SizeCentrality} reports the coverage rates of the correct parameter $\theta_0$ in the $90\%$-confidence sets from \eqref{eqn:GMMWeakIdConfidenceRegion}, where we simulate the forecasts following \eqref{eqn:CentralityAssumptionLoss} and the previous paragraph, where the true $\theta_0$ is given in the second table column.
%E.g., the values in the last row of Table \ref{tab:SizeCentrality} report the frequency how often the true value of $\theta_0 =  (1/3, 1/3, 1/3)$ is included in the estimated confidence set from 	\eqref{eqn:GMMWeakIdConfidenceRegion}, when the forecasts $X_t^*(\theta_0)$ are simulated according to the previous paragraph.
We find accurate inclusion rates for both DGPs and both, symmetric as well as skewed data.
Only the true mode forecasts (compare to Section \ref{sec:SimModeRat}) exhibit some undercoverage for skewed data for the smallest considered sample size.

\begin{table}[tb]
	\centering
	\scriptsize
	\resizebox{\columnwidth}{!}{
	\begin{tabular}{lll lrrr lrrr lrrr lrrr}
		\midrule 
		& & & \multicolumn{7}{c}{(1) iid DGP}  & & \multicolumn{7}{c}{(2) AR-GARCH DGP}  \\
		\cmidrule(lr){4-10} \cmidrule(lr){12-18} 
		%		\cmidrule(lr){2-6} \cmidrule(lr){7-11} \cmidrule(lr){12-16}
		& &&  \multicolumn{3}{c}{symmetric: $\gamma =  0$} & & \multicolumn{3}{c}{skewed: $\gamma = 0.5$}   & &   \multicolumn{3}{c}{symmetric: $\gamma =  0$} & & \multicolumn{3}{c}{skewed: $\gamma = 0.5$} \\
		\cmidrule(lr){4-6} \cmidrule(lr){8-10} \cmidrule(lr){12-14}  \cmidrule(lr){16-18} 
		%		\cmidrule(lr){2-16}
		Forecast Name & \multicolumn{2}{l}{$\theta_0$ \qquad $T=$} & 500 & 2000 & 5000 & & 500 & 2000 & 5000  && 500 & 2000 & 5000 && 500 & 2000 & 5000 \\ 
		\midrule 
		Mean   & $(1,0,0)$  &   & 89.0 & 90.2 & 89.5 &   & 90.0 & 89.5 & 89.5 &   & 89.7 & 91.8 & 90.6 &   & 89.6 & 88.2 & 89.8\\
		Median  & $(0,1,0)$   &   & 89.4 & 89.4 & 91.2 &   & 89.8 & 90.7 & 88.4 &   & 90.1 & 88.8 & 91.0 &   & 89.6 & 89.9 & 89.8\\
		Mode &  $(0,0,1)$  &   & 89.4 & 89.9 & 89.9 &   & 85.2 & 86.3 & 89.1 &   & 88.6 & 90.2 & 90.0 &   & 86.2 & 88.2 & 89.2\\
		Mean-Mode  & $(\tfrac{1}{2},0,\tfrac{1}{2})$  &   & 90.3 & 90.8 & 90.6 &   & 90.3 & 90.3 & 90.8 &   & 91.6 & 89.8 & 91.3 &   & 89.0 & 89.2 & 89.5\\
		Mean-Median  &  $(\tfrac{1}{2},\tfrac{1}{2},0)$ &   & 90.0 & 88.8 & 89.3 &   & 88.8 & 89.9 & 89.9 &   & 89.5 & 89.2 & 89.8 &   & 91.0 & 91.5 & 90.1\\
		Median-Mode &   $(0, \tfrac{1}{2},\tfrac{1}{2})$  &   & 89.7 & 89.6 & 89.6 &   & 90.5 & 89.3 & 89.1 &   & 90.0 & 89.5 & 89.6 &   & 91.2 & 90.0 & 89.3\\
		Mean-Median-Mode &   $(\tfrac{1}{3},\tfrac{1}{3},\tfrac{1}{3})$  &   & 90.1 & 90.5 & 88.8 &   & 90.3 & 89.9 & 90.0 &   & 88.5 & 88.4 & 90.0 &   & 89.2 & 89.8 & 91.1\\
		\bottomrule 
	\end{tabular}
	}
	\caption{\textbf{Confidence set coverage rates under the null hypothesis.} This table presents the empirical coverage rates (in percent) of the confidence sets for the forecasts of central tendency with a nominal coverage rate of $90\%$ for the Homoskedastic cross-sectional and the AR-GARCH DGPs in \eqref{eqn:GeneralDGP}, two skewness choices, a sample size of $T=5000$ and the instruments  $\bh =  (1,X_t)$.}
	\label{tab:SizeCentrality}
\end{table}

\begin{figure}[tp]
	\centering
	\includegraphics[width=\linewidth]{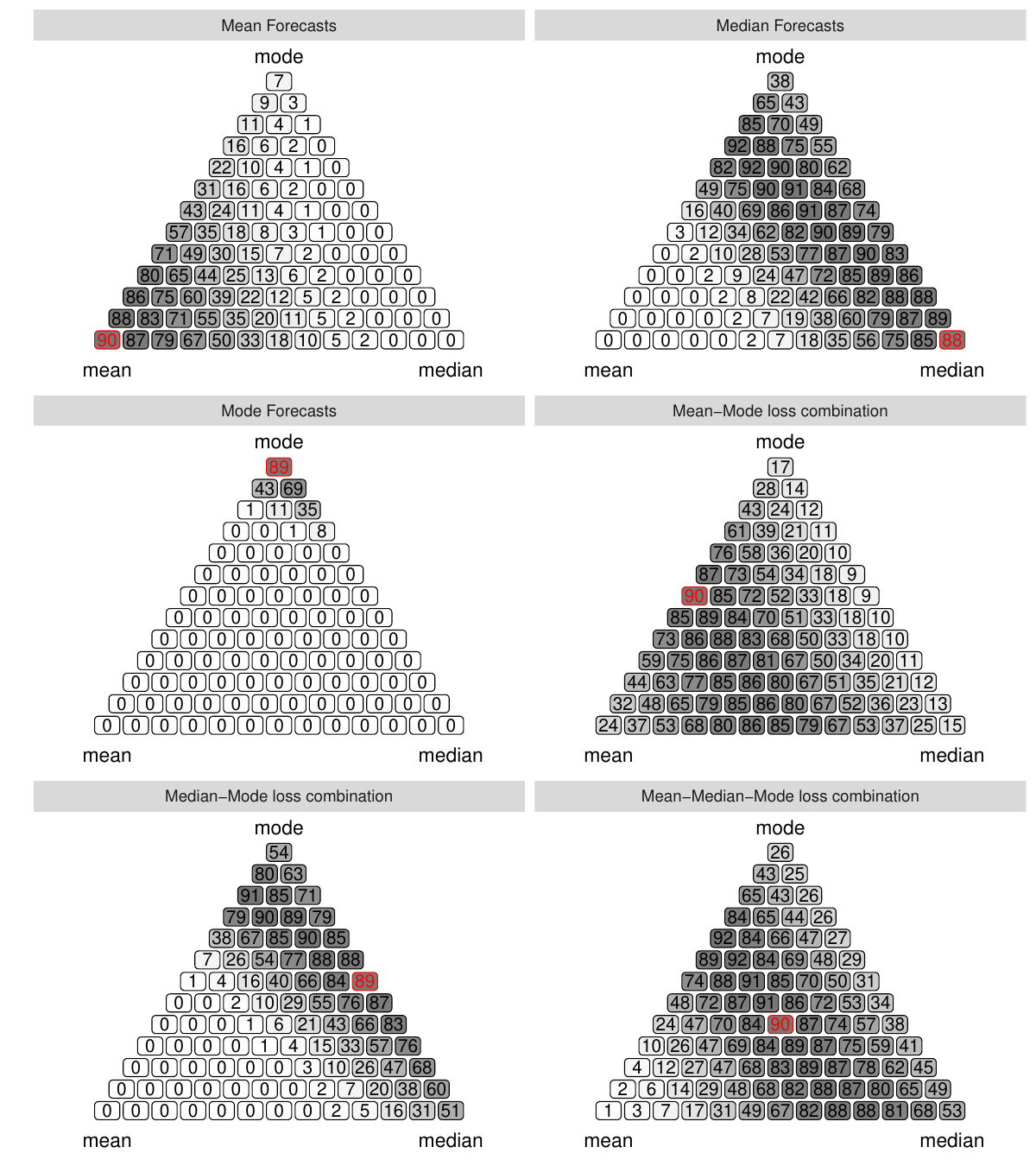}
	\caption{\textbf{Confidence set coverage rates under the alternative.} This figure shows the percentage values in how many of the simulation runs the respective confidence set includes a given point in the triangle.
		The six subfigures contain different true forecasts (whose weights are highlighted in red color) satisfying the null hypothesis in Assumption \ref{assu:GMMWeakIDRegCond} (D) for some $\theta_0 \in \Theta$, respectively.
		We use the  instruments $\mathbf{h}_{t} = (1,X_t)$, the iid DGP, $n=5000$ and $\gamma = 0.5$ throughout this figure.}
	\label{fig:sim_convex_thetacomb}
\end{figure}

Figure \ref{fig:sim_convex_thetacomb} shows the coverage rates of the confidence sets at an equally spaced grid of values in $\Theta$.
The panels in the figure correspond to six true combination weights, which are in each case highlighted in red and are explained in the first two columns of Table \ref{tab:SizeCentrality}.
Each point in the triangles corresponds to a tested centrality measure, i.e.\ to one value of $\theta \in \Theta$, and for each point we display as a percentage number and a color code how often it is contained in the $90\%$ confidence set. 
We use the instruments $\mathbf{h}_{t} = (1,X_t)$, and the iid DGP in equation (\ref{eq:DGPiid}).
%  and report equivalent results for the other DGPs in Supplemental Appendix \ref{sec:AdditionalPlotsTables}.

In each of the six subfigures the values highlighted in red confirm the correct coverage rates of the respective true forecasts from Table \ref{tab:SizeCentrality}.
We further find clearly decreasing coverage rates around the truth illustrating the power our test exhibits. 
Notably, there exist central tendency forecasts, e.g., mean-mode loss combinations in our simulation, that would typically be rejected by standard tests of mean, median, or mode rationality.
The confidence sets are stretching in the upper left direction, which illustrates the partial identification in the simulated process:
As the median is between the mean and mode for the skewed normal distribution, it can also be identified by a convex combination of the mean and mode as can be seen in the upper right panel.\footnote{When the DGP is conditionally location-scale, one of the three centrality measures can always be expressed as a (constant) convex combination of the other two. In this DGP, this is the median, but in other applications it may be any of the three measures. When the conditional distribution exhibits variation in higher-order moments or other ``shape'' parameters, this restriction will generally not hold, and the variation may or may not be sufficient to separately identify the three centrality measures.}

In Supplemental Appendix \ref{sec:AdditionalPlotsTables}, we show plots similar to Figure \ref{fig:sim_convex_thetacomb}, but in each case we exchange one feature:
First, when using symmetric data, all points are naturally identified in all panels.
Second, the results are qualitatively unchanged when using the AR-GARCH DGP in equation (\ref{eq:DGParGARCH}).
Third, the power is lower for the smaller sample size of $T=500$. Finally, Figure \ref{fig:sim_convex_gammacomb} illustrates that our approach identifies probabilistic mixtures of mean, median, and mode forecasts for instruments $\bh = 1$, but generally rejects rationality when using $\bh =  (1,X_t)$.

\section{Evaluating Survey Forecasts of Individual Income}
\label{sec:ApplicationSCE}

We apply our proposed tests to survey responses to the Federal Reserve Bank of New York's Survey of Consumer Expectations.\footnote{Source:  Survey of Consumer Expectations, © 2013-2019 Federal Reserve Bank of New York (FRBNY). The SCE data are available without charge at \url{http://www.newyorkfed.org/microeconomics/SCE} and may be used subject to license terms posted there. FRBNY disclaims any responsibility or legal liability for this analysis and interpretation of Survey of Consumer Expectations data.}$^{,}$\footnote{In Supplemental Appendix \ref{sec:SuppApplications} we consider two other empirical applications. The first is to the ``Greenbook'' forecasts of US GDP growth produced by economists at the Board of Governors of the Federal Reserve. In that application we find that the Greenbook forecasts are rational mean forecasts, but not mode or median forecasts. Our second application is to random walk forecasts of exchange rates, in the spirit of \citet{MeeseRogoff1983}. We find random walk forecasts to be rational mean forecasts for three different exchange rates, while the mode and median are rejected for two of the three exchange rates.} We focus on the Labor Market Survey component, which is conducted each March, July, and November, and which asks participants a variety of questions, including about their current earnings and their beliefs about their earnings in four months (i.e., the date of the next survey). Using adjacent surveys over the period March 2015 to March 2018, we obtain a sample of 3,916 pairs of forecasts ($X_t$) and realizations ($Y_{t+1}$).\footnote{We drop observations that include forecasts or realizations of annualized income below \$1,000 or above \$1 million, which represent less than 1\% of the initial sample. We also drop observations where the ratio of the realization to the forecast, or its inverse, is between 9 and 13, to avoid our results being affected by misplaced decimal points or by the failure to report annualized income (leading to proportional errors of around 10 to 12 respectively).}$^,$\footnote{Most respondents in our sample appear just once: the 3,916 forecast-realization pairs come from 2,628 unique individuals. Our econometric approach does not require repeated samples and so this causes no difficulty, and our results are qualitatively unchanged if we use only the first forecast-realization pair from each respondent.} In testing the rationality of these forecasts we initially assume that all participants report the same, unknown, measure of centrality as their forecast. In the next subsection, we explore potential heterogeneity in the measure of centrality used by different respondents.

\begin{table}[t]
	\centering
	\scriptsize
			\begin{tabular}{llll}
				\toprule 
				\addlinespace
				&  \multicolumn{3}{c}{Centrality measure}  \\ %\vspace{0.1in} \\
				\addlinespace
				\cline{2-4} 
				\addlinespace
				vector $\bh$ & Mean & Median & Mode   \\
				\midrule   
				\addlinespace
				$1$  & 0.008 & 0.401 & 0.737 \\ 
				$1, X$ & 0.014 & 0.000 & 0.768 \\ 
				$1, X,$ lagged income \hspace{0.5cm} & 0.031 & 0.000 & 0.970 \\  
				$1, X,$ government sector & 0.064 & 0.000 & 0.949 \\ 
				$1, X,$ private sector & 0.047 & 0.000 & 0.898 \\
				$1, X,$ job offers  & 0.006 & 0.000 & 0.913 \\ 
				\addlinespace
				\bottomrule 
				\addlinespace
				%	\vspace{0.5cm}
		\end{tabular}
		\caption{\textbf{Evaluating income survey forecasts.} This table presents $p$-values from tests of rationality of individual income forecasts from the New York Federal Reserve's Survey of Consumer Expectations. The columns present test results when interpreting the point forecasts as forecasts of the mean, median or mode. The rows present results for different choices of instrument vectors $\bh$ used in the test: $1$ is the constant, $X$ is the forecast, ``lagged income'' is the respondent's income at the time of the forecast, ``government sector'' and ``private sector'' are indicators for the self-reported industry in which the respondent works, ``job offers'' is an indicator for whether the respondent received any job offers in the previous four months.}
		\label{tab:ResultsSCEIncome}
\end{table}

\Cref{tab:ResultsSCEIncome} presents the results of rationality tests for three measures of central tendency, and for a variety of instrument sets. The first instrument set includes just a constant, and the rationality test simply tests whether the forecast errors have unconditional mean, median or mode, respectively, of zero. The other instrument sets additionally include the forecast ($X_t$) itself, and other information about the respondent collected in the survey. We consider the respondent's income at the time of making the forecast, indicators for the respondent's type of employer,\footnote{The survey includes the categories government, private (for-profit), non-profit, family business, and ``other.'' The first two categories dominate the responses and so we only consider indicators for those.} and whether the respondent received any job offers in the past four months. 

\begin{figure}[t]
	\centering
	\includegraphics[width=.6\linewidth]{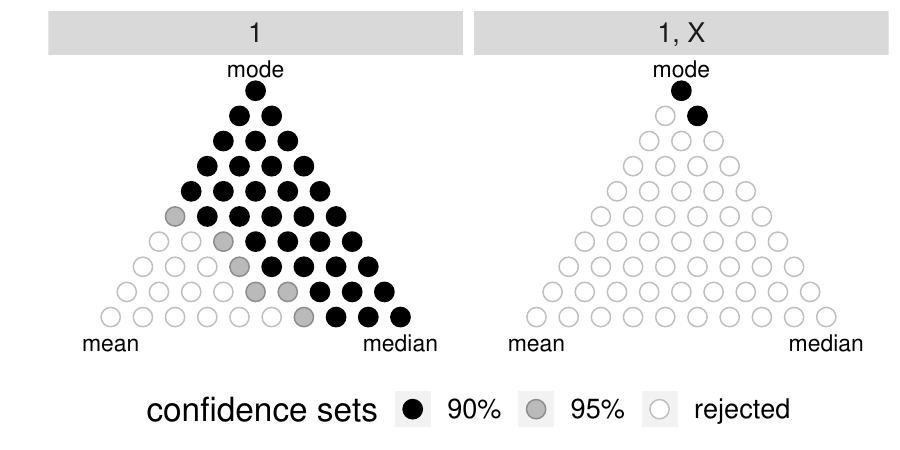}
	\caption{\textbf{Confidence sets for income survey forecasts.} This figure shows the measures of centrality that ``rationalize'' the New York Federal Reserve income survey forecasts. The circles that comprise each triangle correspond to specific convex combinations of the vertices, which are the mean, median and mode functionals. Black dots indicate that the measure is inside the Stock-Wright 90\% confidence set, grey dots indicate that the measure is inside the 95\% confidence set, and white dots indicate that rationality for that measure of centrality can be rejected at the 5\% level. The left panel uses only a constant as the instrument; the right panel uses a constant and the forecast.}
	\label{fig:ResultsSCEIncome}
\end{figure}

The first row of \Cref{tab:ResultsSCEIncome} shows that when only a constant is used, rationality of the survey forecasts can be rejected for the mean, but cannot be rejected for the other two measures of central tendency. When we additionally include the forecast as an instrument, we can reject rationality as mean or median forecasts, but we cannot reject rational mode forecasts. We are similarly able to reject rationality as mean and median forecasts when we include additional covariates, but find no evidence against rationality when these forecasts are interpreted as mode forecasts.

\Cref{fig:ResultsSCEIncome} shows the convex combinations of mean, median and mode forecasts that lie in the confidence set constructed using the methods for weakly-identified GMM estimation in \cite{StockWright2000}.
%\footnote{The interpretation of this figure is slightly different to that of \Cref{fig:sim_ConvexComb_ARGARCH}: in that figure the shade of each dot indicated the proportion of times, across simulations, that point was included in the confidence set, allowing us to study the finite-sample coverage rates of our procedure. In \Cref{fig:ResultsSCEIncome} the shade of each dot indicates whether, for this sample, that point is included in the 90\% confidence set, the 95\% confidence set, or is outside the 95\% confidence set, the latter indicating a rejection of rationality at the 5\% significance level.} 
In the left panel we see that when using only a constant as the instrument, we are able to reject rationality for the mean, and for measures of centrality ``close'' to the mean, but we are unable to reject rationality for the median or mode and measures in a neighborhood of these. This is consistent, naturally, with the entries in the first row of \Cref{tab:ResultsSCEIncome}, which correspond directly to the three vertices in \Cref{fig:ResultsSCEIncome}. Using the results in Section \ref{sec:FunctionalCentrality} and Supplemental Appendix \ref{sec:HeterogeneityForecasters}, the left panel of \Cref{fig:ResultsSCEIncome} can be interpreted assuming either that survey respondents all use the same mixture of loss functions, or that respondents each use a single measure of centrality and are possibly heterogeneous in which measure they use. In the former case the region with black circles identifies the set of loss function weights that are consistent with rationality; in the latter case this region indicates the set of proportions of forecaster \textit{types} in the population.

In the right panel of \Cref{fig:ResultsSCEIncome}, when the instrument set includes a constant and the forecast, we see that only the mode and centrality measures very close to the mode are included in the confidence set; all other forecasts can be rejected at the 5\% level. Overall, we conclude that the responses to the New York Fed's income survey, taken on aggregate, are consistent with rationality when interpreted as mode forecasts, but not when interpreted as forecasts of the mean, median, or convex combinations of these measures.\footnote{Given the panel structure of our data, the target variable may be subject to common unpredictable shocks, leading to forecast errors correlated across individuals within the same wave. In Supplemental Appendix Section \ref{sec:clusterSCE}, we repeat all of our analyses using a covariance estimator clustered at the time level and find very similar results.}

\edit{Holding power issues aside (discussed in the next paragraph), these test results indicate that at least some survey respondents have asymmetric predictive distributions, as the unconditional test rejects mean, but not mode and median rationality, and the conditional test rejects rationality for all centrality measures except those near the mode. Inferring (a)symmetry of the respondents' predictive distributions is not straightforward, as only point forecasts and realizations are observed. As shown by \cite{bai2005tests}, symmetry of the (unconditional) distribution of the observed forecast errors is not sufficient for symmetry of the predictive distribution.}

When interpreting the results in \Cref{fig:ResultsSCEIncome}, and similar figures below, it is worth keeping in mind that the power to detect sub-optimal forecasts is not uniform across values of $\theta$: sampling variation in the mean and median vanishes at rate $T^{-1/2}$, while for the mode it vanishes only at rate approximately $T^{-2/7}$ (see Theorem \ref{thm:ModeRationalityLocalAlternative}). 
This implies that for comparably sub-optimal forecasts, power will be lower at the mode vertex than at the mean or median vertices. This unavoidable variation in power means that the information conveyed by inclusion in the confidence set differs across values of $\theta$.\footnote{\citet{StockWright2000} suggest caution when interpreting a small but nonempty confidence set, as such an outcome is consistent with both a correctly-specified model (a rational forecast, in our case) estimated precisely and also with a misspecified model (irrational forecast) facing low power. These two interpretations clearly have very different economic implications, but cannot be disentangled empirically. Given the lower power at the mode vertex, the latter explanation may be relevant here.}

The analysis of individual income survey forecasts above used $3,916$ pairs of forecasts and realizations from a total of $2,628$ unique survey respondents. This naturally raises the question of whether there is heterogeneity in the measure of centrality used by different respondents.
%\edit{If forecasts are a probabilistic mixture of mean, median, and mode forecasts, the unconditional test has no power, but the conditional test generally does (see Section \ref{sec:HeterogeneityForecasters} for further discussion). With this in mind, Figure \ref{fig:ResultsSCEIncome} may indicate a mixture of mean and mode forecasts. The unconditional test does not reject median and mean-mode forecasts, while the conditional tests do so. Thus, the results would also be consistent with a mixture of mean and mode forecasts that leans so heavily to the mode that the mode cannot be refuted on the data at hand.}
Given that our survey respondents generally only appear once or twice in our sample, allowing for arbitrary heterogeneity is not empirically feasible. In our analysis below we instead stratify survey respondents by observable characteristics and test forecast rationality separately for each subsample. Stratifying the sample may reduce power, due to the smaller sample sizes available, or it may reveal irrationality that is undetected in a joint analysis, for example if one subsample deviates strongly from rationality while the remaining subsamples are rational, or if one subsample exhibits more asymmetric predictive distributions. %, allowing different measures of central tendency to be identified.}

\subsection{Forecast rationality and income level}

\begin{figure}[tb]
	\centering
	\includegraphics[width=.9\linewidth]{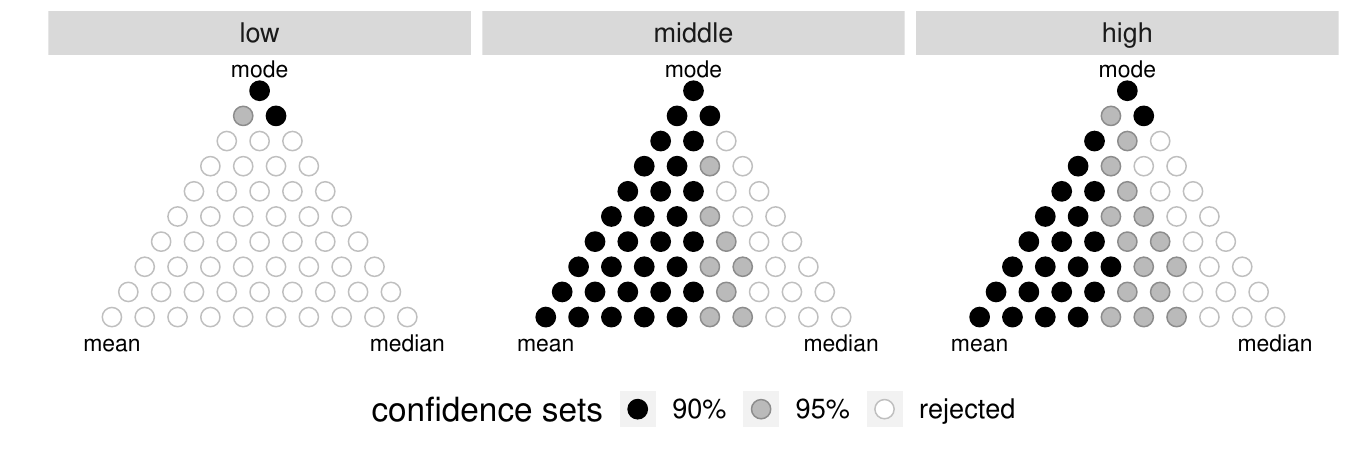}		\caption{\textbf{Confidence sets for income survey forecasts, stratified by income.} This figure shows the measures of centrality that ``rationalize'' the New York Federal Reserve income survey forecasts, for low-, middle- and high-income respondents. Groups are formed using terciles of lagged reported income. Black dots indicate that the measure is inside the Stock-Wright 90\% confidence set, grey dots indicate that the measure is inside the 95\% confidence set, and white dots indicate that rationality for that measure of centrality can be rejected at the 5\% level. All panels use a constant and the forecast as test instruments.}
	\label{fig:ResultsSCEIncomeTercile}
\end{figure}

Firstly, we consider stratifying our sample by income. This is motivated by the possibility that, in addition to a different \textit{level} of future income, low-income respondents face a different \textit{shape} of future income, compared with high-income respondents. This analysis may also reveal that respondents at different income levels use different centrality measures to summarize their predictive distribution. \Cref{fig:ResultsSCEIncomeTercile} presents confidence sets for forecast rationality of measures of centrality for low-, middle-, and high-income respondents based on terciles of the distribution of reported income.\footnote{Qualitatively similar results are found if we stratify the sample into just two groups based on median income.} We see that for low-income respondents, only the mode and measures very close to the mode are contained in the confidence set. For middle- and high-income respondents, the mode, mean, and centrality measures ``close'' to the mean and mode are included in the confidence set. This finding is consistent with all respondents using the mode, and only the distribution for low-income respondents allowing for separate identification of the mode and the mean.  Specifically, if the conditional distribution of future earnings is more skewed for low-income respondents than for high-income respondents, then the measures of centrality are better separated for the former, which allows for better identification of the functional used by survey respondents.  %It is also consistent with low-income respondents reporting the mode as their forecast, while other respondents report the mean or a convex combination of the mean and mode. 
As income is commonly correlated with numeracy and education levels, rejecting mean- and median-rationality for lower-income respondents is consistent with work in the experimental literature: \cite{KroegerThibaud2019QuestionFormats} observe high-numeracy individuals to be better at mean and median forecasting.  

\subsection{Forecast rationality and age}

\begin{figure}[t]
	\centering
	\includegraphics[width=.7\linewidth]{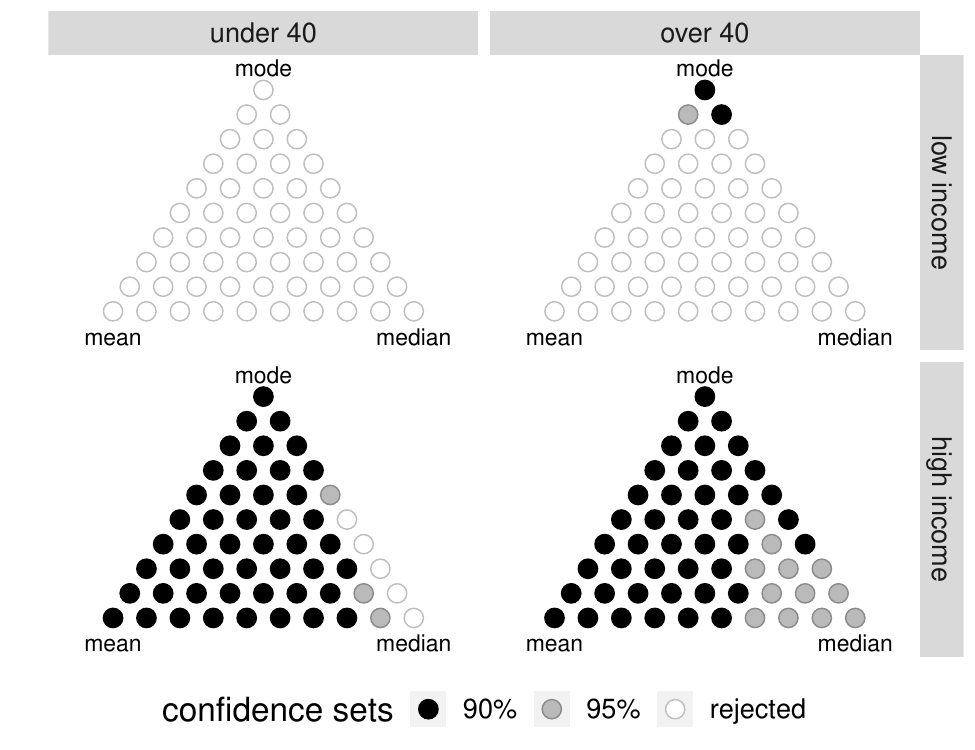}
	\caption{\textbf{Confidence sets for income survey forecasts, stratified by age and income.} See the notes to Figure \ref{fig:ResultsSCEIncomeTercile} for additional details.}
	\label{fig:ResultsSCEIncomeAge}
\end{figure}

We next stratify the sample both by income and by age, motivated by the possibility that younger respondents have less experience in the workforce and may be less able to predict their future earnings. The latest versions of the FRBNY survey contain a field for whether the respondent is above or below 40 years of age, while earlier versions asked for the respondent's specific age. To maximize coverage, we adopt the age 40 split contained in the later versions of the survey. We have 2,457 (1,332) respondents over (under) that age. We use the median income within each age category to define ``low'' and ``high'' income groups. 

\Cref{fig:ResultsSCEIncomeAge} presents the striking result that forecasts from younger low-income workers cannot be rationalized using \textit{any} measure of centrality; all points lie outside the 95\% confidence set. Forecasts from older low-income respondents can be rationalized as only mode forecasts. In contrast, forecasts from both young and old high-income workers can be rationalized for a variety of measures of centrality.\footnote{Note that a large confidence set for $\theta$ does not imply that the respondents are ``more rational'' than if the confident set were small; a respondent can be fully rational at just a single value of $\theta$, and if the predictive distribution is sufficiently skewed then the confidence set will, in the limit, contain just that value of $\theta$. A large confidence set may reflect a predictive distribution that does not allow for the identification of a single value of $\theta$, or low finite-sample power.}
This figure suggests that younger low-income workers make systematic errors when attempting to predict their income in the coming four months. Asymmetric loss functions, such as those used in \citet{EKT2005} (EKT), have the potential to explain these forecasts, and we consider this in  Section \ref{sec:SCEasym} below. Foreshadowing those results, we also reject rationality of forecasts from younger low-income respondents in the EKT framework.

% @TD and @PWS: I deleted the following summary paragraph, as it's too early to summarize results in our new longer analysis of the FRBNY survey.
%Overall, the results in this section indicate some important heterogeneity in both the rationality of point forecasts and the measure of central tendency employed. We find that forecasts from younger, low-income survey respondents cannot be rationalized using any measure of central tendency, and forecasts from low-income survey respondents who are likely to change jobs in the coming period can only be rationalized as mode forecasts. %(Given the  known lower power at the mode vertex, this result may be an indication of a failure of forecast rationality.) 
%In contrast, forecasts from respondents with income above the median, regardless of their age or likelihood of changing jobs, can be rationalized using many different centrality measures. 

\subsection{Forecast rationality and job stability}
 
An important component of income uncertainty is job search and new job offers, see \citet{mueller2022expectations} for a recent review of survey expectations in the job search literature. Respondents naturally have different prospects of changing jobs, and these changes impact their predictive distributions of income. To proxy for the likelihood of changing jobs, we stratify respondents based on whether they reported receiving at least one job offer over the previous four months.%Such respondents are more likely to change jobs in the coming period, and therefore face a more uncertain distribution of future income.
\footnote{We find very similar results when we stratify respondents using their estimated probability of receiving a job offer in the next four months, or by their estimated probability of staying in the same job for the next four months.} 

\edit{\Cref{fig:ResultsSCEIncomeOffer} reveals that respondents who did not receive a job offer, and high-income respondents who did, provide forecasts that can be rationalized as mean, mode, and combinations thereof. For high-income respondents without job offer the confidence set contains also (near-)median forecasts, indicating that the predictive distributions for those respondents are more symmetric. Forecasts from low-income respondents who reported receiving a job offer can only be rationalized as mode forecasts; rationality for all other centrality measures is rejected. This is consistent with the predictive distributions for these respondents being more asymmetric, allowing for discrimination between centrality measures. Overall, these results reveal that low-income respondents with more job uncertainty differ meaningfully from the other three subgroups.}

\begin{figure}[t]
	\centering
	\includegraphics[width=0.7\linewidth]{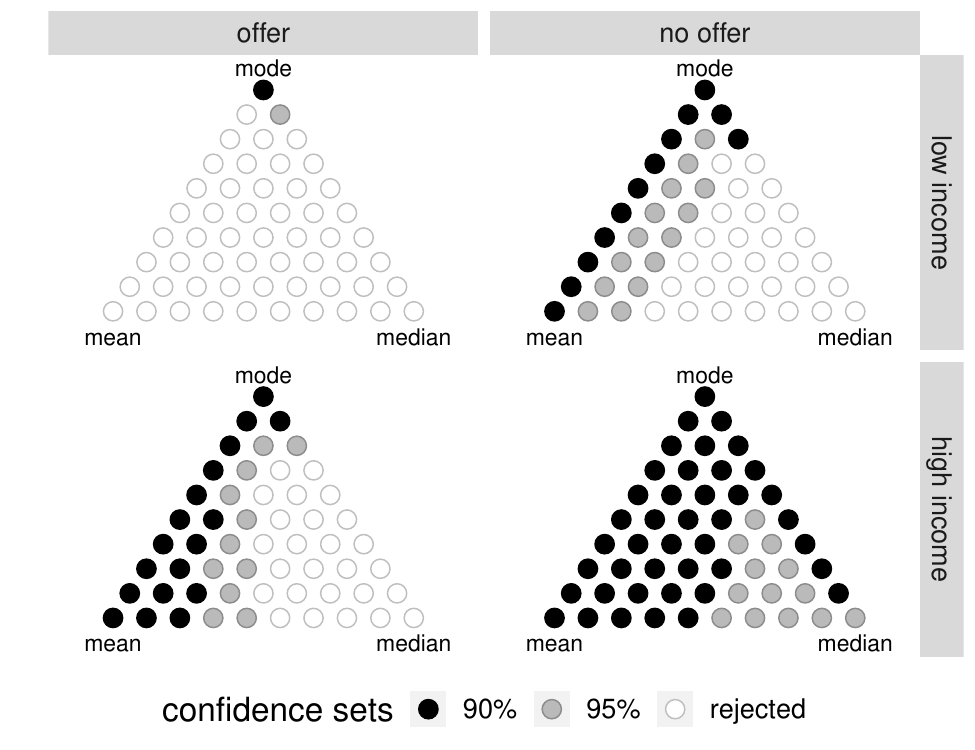}
	\caption{\textbf{Confidence sets for income survey forecasts, stratified by job offer status and income.} See the notes to Figure \ref{fig:ResultsSCEIncomeTercile} for additional details.}
	\label{fig:ResultsSCEIncomeOffer}
\end{figure}

\subsection{Forecast rationality and survey experience}

\begin{figure}[t]
	\centering
	\includegraphics[width=0.7\linewidth]{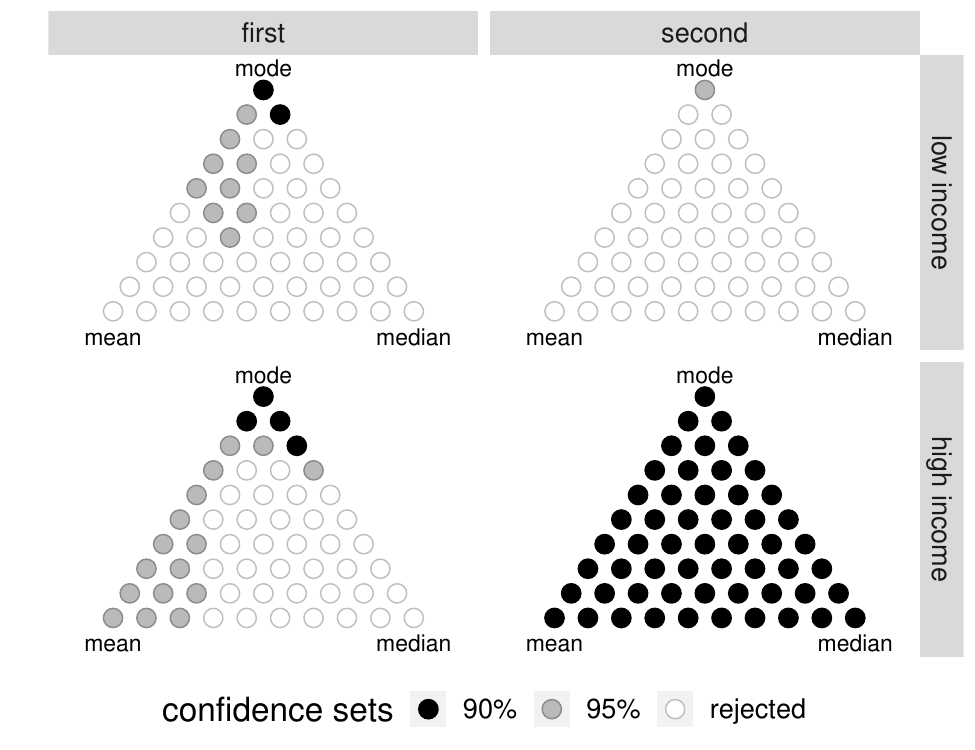}
	\caption{\textbf{Confidence sets for income survey forecasts, stratified by elicitation round and income.} See the notes to Figure \ref{fig:ResultsSCEIncomeTercile} for additional details.}
	\label{fig:ResultsSCEIncomeRound}
\end{figure}

As economic agents gather experience, their actions and expectations change, see \citet{vanNVeldkamp2006learning} and \citet{malmendier2011depression} for example. In recent work, \citet{kim2022learning} show that this effect is true in economic surveys as well: forecasts of inflation show less uncertainty as respondents garner experience responding to the survey. We next investigate whether having survey experience leads to more rational forecasts. The FRBNY survey includes participants for a maximum of twelve months, and within that period they are asked to predict future income and report current income three times, providing us with up to two matched pairs of forecasts and realizations for each respondent. We have a total of 1,288 respondents with two matched pairs, and we study the rationality of the first and second of these forecasts separately.

% AJP: old version below. 
%\color{red}Figure \ref{fig:ResultsSCEIncomeRound} reveals a stark difference in the impact of the round of elicitation between low- and high-income respondents. Forecasts from low-income respondents can be rationalized as mode forecasts, and the results change only slightly from the first to the second round: in the first round some ``near mode'' functionals are also rationalizable, and in the second round the mode functional is only rationalizable at the 95\% significance level. For high-income respondents, in contrast, only the mode and mean functionals (and convex combinations thereof) are rationalizable in the first round, while in the second round, four months later, all functionals are rationalizable. This suggests that high-income respondents' expectation formation changes such that respondents have different, potentially less skewed, beliefs after the initial round.  Our results for high-income respondents are consistent with the findings of \citet{kim2022learning} for inflation forecasts, while the similarity across survey rounds for low-income respondents differs from that study. \color{black}

% new interpretation (2may23)
Figure \ref{fig:ResultsSCEIncomeRound} reveals a stark difference in the impact of the round of elicitation between low- and high-income respondents. Forecasts from low-income respondents can be rationalized as mode forecasts, and the results change only slightly from the first to the second round: in the first round some ``near mode'' functionals are also rationalizable, and in the second round the mode functional is only rationalizable at the 95\% significance level. For high-income respondents, in contrast, only the mode and mean functionals (and convex combinations thereof) are rationalizable in the first round, while in the second round, four months later, all functionals are rationalizable. Assuming that the \textit{true} conditional distributions of future income are unaffected by survey participation, and from the second round results we can infer these are (approximately) symmetric, this suggests that high-income respondents' forecasts are more accurate in the second round, and as a result can be (correctly) rationalized via more functionals.\footnote{Recall that a large confidence set does not imply ``more rational,'' merely that the forecast is rationalizable by a wider variety of measures of centrality.} Our results for high-income respondents are consistent with the findings of \citet{kim2022learning} for inflation forecasts, while the similarity across survey rounds for low-income respondents differs from that study. 

In Supplemental Appendix \ref{sec:SupplPastFCAccuracy}, we present an additional analysis, stratifying respondents in the second round by the size of their forecast error in the first round. We find that respondents with high accuracy in the first round can be rationalized using almost any measure of centrality, regardless of income, as can high-income low-accuracy respondents. In contrast, lower-income respondents with low accuracy in the first round cannot be rationalized using any measure of centrality at the 90\% confidence level.

\subsection{Irrational or optimistic?}\label{sec:SCEasym}

The premise of the rationality test proposed in this paper is that survey respondents use a measure of central tendency when providing a point forecast but the specific centrality measure used is unknown. An alternative framework is to model the respondent as reporting either a centrality measure \textit{or} some other functional of their predictive distribution. For example, rather than using the median they may use some other quantile, capturing optimism or pessimism. \cite{EKT2005} (EKT) consider such a case, and model respondents as using either ``lin-lin'' loss, which elicits a quantile of the predictive distribution, or ``quad-quad'' loss, which elicits an expectile of the predictive distribution. In both cases, the loss function contains an asymmetry parameter, and the special case of symmetry corresponds to the respondent using the median or the mean.

\begin{table}[tb]
	\centering
	\scriptsize
			\begin{tabular}{lcccccc} %{lrlllll}
				\toprule 
				\addlinespace
				 & n & Mean & Median & Mode & Quantiles & Expectiles \\ 
%				\midrule   
				\cline{2-7} 
				\addlinespace
				Full sample & 3916 & 0.01 & 0.00 & 0.77 & 0.00 [0.61, 0.64] & 0.01 [0.49, 0.56] \\ \midrule \addlinespace
				
				Low income & 1263 & 0.03 & 0.00 & 0.89 & 0.00 [0.56, 0.61] & 0.01 [0.45, 0.57] \\ 
				Middle income & 1263 & 0.13 & 0.00 & 0.92 & 0.00 [0.62, 0.66] & 0.06 [0.40, 0.54] \\ 
				High income & 1263 & 0.15 & 0.02 & 0.46 & 0.05 [0.62, 0.67] & 0.06 [0.47, 0.57] \\ \midrule \addlinespace 
				
				Under 40, Low income & 665 & 0.00 & 0.00 & 0.00 & 0.00 [0.56, 0.62] & 0.00 [0.63, 0.80] \\ 
				Over 40,~~~Low income & 1236 & 0.05 & 0.00 & 0.61 & 0.00 [0.59, 0.64] & 0.02 [0.47, 0.59] \\ 
				Under 40, High income & 667 & 0.35 & 0.03 & 0.83 & 0.01 [0.61, 0.68] & 0.15 [0.42, 0.57] \\ 
				Over 40,~~~High income & 1221 & 0.15 & 0.09 & 0.98 & 0.26 [0.62, 0.67] & 0.05 [0.44, 0.56] \\ 
				\bottomrule 
				\addlinespace
			\end{tabular}
		\caption{\textbf{Summary of $p$-values for rationality tests in different samples.} This table presents the sample size and $p$-values from tests of rationality when interpreting the point forecasts as forecasts of the mean, median, or mode. The last two panels present $p$-values from tests of rationality when interpreting the point forecasts as quantiles or expectiles following \cite{EKT2005}, as well as 90\% confidence intervals for the asymmetry parameter, given in square brackets. }
		\label{tab:sce_samples}
\end{table}

%Compared to our method, \cite{EKT2005} allow for asymmetric deviations from the mean and median. The initial paper motivates the test with pice-wise linear and squared loss functions, but the tests can be interpreted as testing for quantile and expectile forecasts \citep{schmidt2021interpretation}. Table \ref{tab:sce_samples} shows $p$-values for rationality tests for the mean, median, and mode and for quantiles and expectiles for the main stratifications considered above (for the remainder of stratifications see Supplementary Section \ref{sec:SuppSCErest}). The quantile and expectile tests based on \cite{EKT2005} were implemented with the continuous-updating GMM estimator \citep{YaronCUE} with an instrumental vector consisting of a constant and the forecast, consistent with our main analysis. The table indicates the $p$-value for the rationality test of a quantile or expectile forecast and $90\%$ confidence sets for the asymmetry parameter. An asymmetry parameter of $0.5$ corresponds to the median for quantiles and mean for expectiles.

Table \ref{tab:sce_samples} presents tests of rationality in the EKT framework, as well as the results for the three measures of centrality considered in this paper. These are presented for the full sample, the three income subsamples, and the $2\times2$ stratifications based on age and income. Results for the other stratifications considered above are presented in Supplementary Appendix \ref{sec:SupplEKTSubsamples}. As in our previous analysis, we use a constant and the forecast as instruments for the EKT tests.

For the full sample of respondents, in the top line of Table \ref{tab:sce_samples}, we see that rationality is rejected in the EKT framework. When interpreted as a quantile forecast, the respondents' estimated quantile 90\% confidence interval is 0.61 to 0.64, indicating optimism, but the $p$-value for the test of rationality is less than 0.005, rejecting rationality.\footnote{The estimated shape parameter is interpretable as the one that makes the observed forecasts as close as possible to rational, and the corresponding $p$-value determines whether these forecasts are consistent or inconsistent with rationality.} When interpreted as an expectile forecast, the shape parameter straddles 0.5 (the case of symmetry, where the expectile equals the mean) but the test again rejects rationality. The only functional for which rationality is not rejected is the mode, using the new mode forecast rationality test proposed in this paper.

The results of the rationality tests stratified by income, presented in the second panel of Table \ref{tab:sce_samples}, show that middle- and high-income respondents can be rationalized as expectile forecasters, and in both cases the confidence interval for the estimated shape parameter includes the point of symmetry, consistent when the mean functional not being rejected for these subsamples. For low-income respondents, we see that allowing for optimism or pessimism through the EKT framework does not rationalize these forecasts: the $p$-values of the EKT rationality tests for these cases are both less than or equal to 0.01. 

Finally, we stratify survey respondents by both age and income, and present the results in the lower panel of Table \ref{tab:sce_samples}. Recall that above we found that the younger, low-income subsample could not be rationalized as \textit{any} measure of centrality. When using the EKT framework, we find that the estimated shape parameters for the lin-lin and quad-quad loss functions are both greater than 0.5, the point of symmetry, indicating optimistic forecasts. However the $p$-values in both cases are less than 0.005, strongly rejecting rationality. Thus, forecasts from younger low-income respondents cannot be rationalized as centrality measures, nor as quantiles or expectiles.

\section{Conclusion}
\label{sec:Conclusion}

Reasonable people can interpret a request for their prediction of a random variable in a variety of ways. Some, including perhaps most economists, will report their expectation of the value of the variable (i.e., the mean of their predictive distribution), others might report the value such that the observed outcome is equally likely to be above or below it (i.e., the median), and others may report the value most likely to be observed (i.e, the mode). Still others might solve a loss minimization problem and report a forecast that is not a measure of central tendency. Economic surveys generally request a \textit{point} forecast, despite calls for surveys to solicit distributional forecasts, see \citet{manski2004measuring} for example, and the specific type of point forecast (mean, median, etc.) to be reported is generally not made explicit in the survey. 

This paper proposes new methods to test the rationality of forecasts of some unknown measure of central tendency. Similar to \citet{EKT2005}, we propose a testing framework that nests the mean forecast as a special case, but unlike that paper we allow for alternative forecasts within a general class of measures of central tendency, rather than measures that represent other aspects of the predictive distribution (such as non-central quantiles or expectiles). We consider a class %of central tendency measures 
that nests convex combinations of the mean, median and mode, and we overcome an inherent identification problem that arises using the methods of \citet{StockWright2000}.

As a building block for the above test, we also present new tests for the rationality of mode forecasts. Some recent work (e.g., \citealp{Knuppel2012}; \citealp{ReifschneiderTulip2017}) suggests that point forecasts from central banks should be interpreted as mode rather than mean forecasts, and experimental studies (e.g., \citealp{peterson1964mode}; \citealp{KroegerThibaud2019}) report that participants are more likely to use the mode to summarize their predictive distribution than other measures of central tendency. While mode regression has received some attention in the recent literature (\citealp{Kemp2012} and \citealp{Kemp2019}), tests for mode forecast rationality similar to those available for the mean and median (e.g., \citealp{MincerZarnowitz1969} and \citealp{GaglianoneJBES2011}) are lacking. Direct analogs of existing tests are infeasible because the mode is not elicitable \citep{Heinrich2014}. We introduce the concept of \textit{asymptotic elicitability} and show it applies to the mode by considering a generalized modal midpoint with asymptotically vanishing length. We then present results that allow for tests similar to the famous Mincer-Zarnowitz regression for mean forecasts. 

We apply our tests to individual 
income expectation survey data collected by the Federal Reserve Bank of New York. We reject forecast rationality when interpreting responses as mean or median forecasts, however we cannot reject rationality when interpreting them as mode forecasts. We also find evidence of heterogeneity across respondents: for example, forecasts from younger, low-income respondents are not rationalizable using any measure of centrality, while forecasts from high-income respondents, regardless of their age, are rationalizable for many, though not all, measures of centrality. Further, we find that the behavior of high-income survey participants changes as they gain experience in the survey, consistent with the learning suggested in \citet{kim2022learning}.

%forecasts in the second panel round and after large income shocks are more likely to be consistent with mean forecasts. %In two additional examples in the Supplemental Appendix concerned with professional forecasters and market prices, we find rational mean forecasts and evidence against the median and the mode.

Previous analyses of economic surveys have typically assumed that responses are mean forecasts, and deviations from that benchmark have motivated learning and attention models which can account for such irrationality \citep[e.g.,][]{bordalo2020overreaction,coibion2015information,kohlhas2021asymmetric}. This paper proposes an intriguingly simple alternative to mean rationality, and one which is supported by our empirical work: point forecasts may not reflect the statistical average, but simply the most likely value, the mode. %As a consequence, the information contained in survey expectations is limited. In particular, quite substantive changes in the tail of a distribution (reflecting uncertainty) need not change the mode. This aspect could also explain the missing link between disagreement and uncertainty \citep[see][and references therein]{robert2010relationship}. Our empirical results suggest a heterogeneity in expectation formation that deserves further investigation with unexperienced survey participants reporting the mode, and experienced survey participants and professional forecasters reporting the mean.

% Next we study the Federal Reserve's ``Greenbook'' forecasts of U.S.\ GDP, and we find that we cannot reject rationality with respect to the mean, however we can reject with respect to the median and mode. This is consistent with the Greenbook forecasts being made using econometric models, which almost invariably focus on the mean as the measure of central tendency. Finally, we revisit the famous result of \cite{MeeseRogoff1983} that random walk forecasts for exchange rates are rational mean forecasts. For the USD/EUR exchange rate, we find we cannot reject rationality with respect to \textit{any} of convex combination of the mean, median and mode, indicating the random walk forecast is rational under any of these measures. For the JPY/EUR and AUD/EUR exchange rates, however, we find the random walk forecast is only rational for centrality measures ``close'' to the mean; rationality with respect to the median and mode is rejected. 

%%\pagebreak
%\doublespacing
%%\begin{appendices}
%\appendix

%\pagebreak
\singlespacing
\small
\setlength{\bibsep}{4pt}
%\addcontentsline{toc}{section}{References}
%\setstretch{0.8} 
\bibliographystyle{apalike}	
\bibliography{Bib_Mode}

%\end{appendices}

\newpage
\doublespacing
\setcounter{page}{1}
\setcounter{footnote}{0}
\setcounter{section}{0}
\setcounter{table}{0}
\setcounter{figure}{0}

\begin{center}
	SUPPLEMENTARY MATERIAL FOR   \vspace{10pt} \\
{\Large\bf {Testing Forecast Rationality for Measures of Central Tendency} \vspace{10pt} }\\
Timo Dimitriadis and Andrew J. Patton and Patrick W. Schmidt \\
	This Version: \today \\ 
\end{center}

%\begin{appendices}

\renewcommand{\thesection}{S.\arabic{section}}   
\renewcommand{\thepage}{S.\arabic{page}}  
\renewcommand{\thetable}{S.\arabic{table}}   
\renewcommand{\thefigure}{S.\arabic{figure}}   

%\renewcommand{\thetheorem}{\arabic{section}.\arabic{theorem}}

% The following commands allow to also alter page, table and figure numbering

%\renewcommand{\thefigure}{S\arabic{figure}}	

% @TD and @PWS: I re-ordered the sections here to match their order of reference in the main text. I also added a paragraph to describe all of the content we have in the supp app -- it's a lot!

This appendix contains additional details, discussions, and results. It is structured as follows: 
\edit{Section \ref{sec:Proofs} contains proofs for the theorems in the main paper and  \ref{sec:TechnicalProofs} contains some additional technical lemmas and their proofs.}
\edit{Sections \ref{sec:BandwidthChoice} and \ref{sec:KernelChoice}} discuss the impact of the choice of kernel and bandwidth on the results for mode forecast rationality testing in Section \ref{sec:ForecastRatioanlityTestMode} of the main paper. Section \ref{sec:convexfunctionals} presents a result showing that a convex combination of functionals is generally not elicitable, as stated in Remark \ref{rem:ConvexFunctionalCombinations} of the main paper.
\edit{Section \ref{sec:HeterogeneityForecasters} relates our tested null hypothesis to a probabilistic combination of mean, median and mode forecasters.}
Section \ref{sec:AdditionalPlotsTables} presents simulation results supplementing those presented in Section \ref{sec:SimulationStudy} of the main paper. 
Section \ref{sec:SuppSCErest} presents additional results for the application in Section \ref{sec:ApplicationSCE} of the main paper. In particular, we show tests of rationality using a cluster covariance estimator, results for an additional stratification with respect to past forecast accuracy and the results of tests of rationality in the framework of \citet{EKT2005} for the full set of stratifications.
Section \ref{sec:SuppApplications} presents two additional empirical analyses, the first to the ``Greenbook'' forecasts produced by the Board of Governors of the Federal Reserve, and the second to random walk forecasts of exchange rates.

\section{Proofs}
\label{sec:Proofs}

\begin{proof}[Proof of Theorem \ref{thm:GenModalMidpointProperties}]
	For the proof of statement \ref{statement:GenModalMidpointWellDefined}, let $Y \sim P \in \mathcal{P}$ with density $f$, define $\tilde K_\delta(e) = \frac{1}{\delta} K\left( \frac{e}{\delta} \right)$ and introduce the notation $\bar L^K_\delta(x,P) = \mathbb{E}_{Y\sim P} \left[ L^K_\delta(x,Y) \right]$.
	Then, it holds that
	\begin{align*}
		\bar L^K_\delta(x,P) 
		= - \int  \frac{1}{\delta} K\left( \frac{x-y}{\delta} \right) f(y) \, \mathrm{d}y
		= - \int  \tilde K_\delta \left( x-y \right) f(y) \, \mathrm{d}y
		= - (f \ast \tilde K_\delta) (x),
	\end{align*}
	where $f \ast \tilde K_\delta$ denotes the convolution of the functions $f$ and $ \tilde K_\delta$.
	\cite{Ibragimov1956} shows that for any log-concave density, its convolution with any other unimodal distribution function is again unimodal.\footnote{\cite{Ibragimov1956} calls densities satisfying this property \textit{strongly unimodal}. It is important to note that his notion of strong unimodality is different from ours introduced in \Cref{def:Unimodality}.}
	Thus, $\bar L^K_\delta(x,P)$ exhibits a unique minima which shows that $\Gamma_\delta^K$ is well-defined by (\ref{eqn:LossGenModalInt}).
	
	We continue with statement \ref{statement:GenModalMidpointConvergence} and show that $\Gamma_\delta(P) \to \mode(P)$ for all $P \in \mathcal{P}$.
	Notice that
	\begin{align}
		\bar L^K_\delta(x,P) 
		= - \int  \frac{1}{\delta} K\left( \frac{x-y}{\delta} \right) f(y) \, \mathrm{d}y
		= - \int f(x + u \delta) K\left(u \right)  \, \mathrm{d}u
		=  - \int f(x + u \delta) \, \mathrm{d} \mathcal{K}(u),
	\end{align}
	by applying integration by substitution and by interpreting the kernel $K(\cdot)$ as the density of the probability measure $\mathcal{K}$.
	%	Notice that the functional $\Gamma_\delta^K(P)$ induced by the loss function $L^K_\delta$ is given as $\argmin_x 	L^K_\delta(x,P)$.
	%	As the density $f$ is bounded from above by assumption, we can apply dominated convergence in order to conclude that
	%	\begin{align}
		%	\lim_{\delta \to 0} L^K_\delta(x,P)
		%	= - \int \lim_{\delta \to 0} f(x + u \delta) \, \mathrm{d} \mathcal{K} (u)
		%	= - \int f(x) \, \mathrm{d} \mathcal{K} (u)
		%	= - f(x)
		%	\end{align}
	%	for all $x \in \mathbb{R}$ as $\int \mathrm{d} \mathcal{K} (u) = 1$.
	It holds that
	\begin{align}
		&\sup_{x \in \mathbb{R}} \left| \bar L^K_\delta(x,P) - \big( - f(x) \big) \right|
		= \sup_{x \in \mathbb{R}} \left| f(x) -   \int f(x + u \delta) \mathrm{d} \mathcal{K}(u) \right| \\
		\le \, &\sup_{x \in \mathbb{R}} \int \big| f(x) -   f(x + u \delta) \big| \mathrm{d} \mathcal{K}(u) 
		\le \sup_{x \in \mathbb{R}} \int | c \delta u | \, \mathrm{d} \mathcal{K}(u) 
		= c \delta \int |u| {K}(u) \mathrm{d}u \to 0
	\end{align}
	as $\delta \to 0$ 
	%	 as $\big| f(x) -   f(x + u \delta) \big| \le c  | u \delta| $ for some $c > 0$ 
	as $f$ is Lipschitz continuous with constant $c\ge0$ and $\int  |u| {K}(u) \mathrm{d}u < \infty$.
	Hence, $\bar L^K_\delta(x,P)$ converges uniformly (for all $x \in \mathbb{R}$) to $-f(x)$ as $\delta \to 0$ and it  holds that $\argmin_x  \bar L^K_\delta(x,P) \to \argmin_x  (- f(x) )$ as $\delta \to 0$.
	Consequently,
	\begin{align*}
		\lim_{\delta \to 0} \Gamma_\delta(P)
		= \lim_{\delta \to 0}  \left( \argmin_{x \in \mathbb{R}}  \bar L^K_\delta(x,P) \right) 
		= \argmin_{x \in \mathbb{R}} \left( \lim_{\delta \to 0}  \bar L^K_\delta(x,P) \right)
		= \argmin_{x \in \mathbb{R}} \left( - f(x) \right),
	\end{align*}
	which equals the mode for distributions with continuous Lebesgue density by definition.
	%	as there exists some $x \in \supp(P)$ such that $-f(x) < -f(\tilde x)$ for all $\tilde x \in \mathbb{R}$ and consequently, $ \Gamma(P) = \argmin_x \left( - f(x) \right)$ is the mode of the distribution $P$ by definition.
	%	By assumption, it holds that there exists some $x \in \supp(P)$ such that $-f(x) < -f(\tilde x)$ for all $\tilde x \in \mathbb{R}$, where $x = \Gamma(P)$ is the mode of the distribution $P$ by definition.
	%	
	%	As 
	%	
	%	
	%	As  $L^K_\delta(x,P)$ is a continuous function in $x$ and $\delta$, we get that 
	%	\begin{align}
		%	\Gamma(P)
		%	= \lim_{\delta \to 0} \Gamma_\delta(P)
		%	= \lim_{\delta \to 0}  \left( \argmin_x  L^K_\delta(x,P) \right) 
		%	= \argmin_x \left( \lim_{\delta \to 0}  L^K_\delta(x,P) \right).
		%	\end{align}
	%The loss function $L^K_\delta(x,Y)$ is asymptotically strictly consistent for the mode for $\delta \to 0$.
	
	For the proof of statement \ref{statement:GenModalMidpointIdentifiable}, \edit{let $P \in \tilde{\mathcal{P}}$ with density $f$.
	We} consider a fixed $\delta > 0$ and define $\bar V^K_\delta(x, P) \defeq - \mathbb{E}_{Y\sim P} [ V_\delta(x,Y) ] = \frac{1}{\delta^2}\int K' \left( \frac{x-y}{\delta} \right) f(y) \mathrm{d}y$ and $\tilde K_\delta'(e) = \frac{1}{\delta^2} K'\left( \frac{e}{\delta} \right)$. 
	Then,
	\begin{align}
		\bar V^K_\delta(x, P)	&=  \int \tilde K_\delta'(x-y) f(y) dy 
		=  \big( \tilde K_\delta' \ast f \big) (x)
		=   \big(  \tilde K_\delta \ast f' \big) (x)
		=  \int \tilde K_\delta(x-y) f'(y) dy \\
		&= \int \tilde K_\delta(x+\mode(P)-y) f'(y-\mode(P)) dy\\
		&= \int \tilde K_\delta(x+\mode(P)-y) g'(y) dy
	\end{align}
	for some shifted density $g$ with mode at zero.
	As the kernel $\tilde K_\delta$ is log-concave it has a monotone likelihood ratio (Proposition 2.3 (b) in \citealp{Saumard2014}), i.e.\ for any $a \le b$ and $y \ge 0$, it holds that
	$\frac{\tilde K_\delta(a-y)}{\tilde K_\delta(a)} \le \frac{\tilde K_\delta(b-y)}{\tilde K_\delta(b)}$ and 	as $g'(y) \le 0$ for $y\ge 0$ \edit{(as $P \in \tilde{\mathcal{P}}$ is strongly unimodal)}, this implies that
	\begin{align}
		%	 \frac{\tilde K(a-y)}{\tilde K(a)} &\le \frac{\tilde K(b-y)}{\tilde K(b)}\\
		%	 \frac{\tilde K(a-y)}{\tilde K(b-y)} &\le \frac{\tilde K(a)}{\tilde K(b)}\\
		\frac{\tilde K_\delta(a-y)}{\tilde K_\delta(b-y)} g'(y) &\ge \frac{\tilde K_\delta(a)}{\tilde K_\delta(b)} g'(y).
	\end{align}
	Analogously, the same inequality follows for any $a \le b$ and $y \le 0$, where the monotone likelihood ratio above holds with the reverse inequality but at the same time $g'(y) \ge 0$ \edit{as $P \in \tilde{\mathcal{P}}$.}
	Thus, for any $x>z$,
	\begin{align}
		\bar V^K_\delta(x, P) &=  \int  \tilde K_\delta (x+\mode(P)-y) g'(y) dy \\
		&=  \int  \frac{\tilde K_\delta (x+\mode(P)-y)}{\tilde K_\delta (z+\mode(P)-y)} \tilde K_\delta (z+\mode(P)-y) g'(y) dy \\
		&\ge  \frac{\tilde K_\delta (x+\mode(P))}{\tilde K_\delta (z+\mode(P))} \int \tilde K_\delta (z+\mode(P)-y) g'(y) dy \\
		&=  \frac{\tilde K_\delta (x+\mode(P))}{\tilde K_\delta (z+\mode(P))} \; \bar V^K_\delta(z, P).
	\end{align}
	For any root $x^\ast$ such that $\bar V^K_\delta(x^\ast, P)= 0$, it follows that $\bar V^K_\delta(x_u, P) \ge 0$ for all $x_u > x^\ast$ and $\bar V^K_\delta (x_l, P) \le 0$ for all $x_l < x^\ast$ as the kernel $\tilde K_\delta$ has infinite support.
	If there exist two distinct roots $x_l^\ast < x_u^\ast$ of $\bar V^K_\delta (\cdot, P)$, the above implies that $\bar V^K_\delta(x, P)=0$ for all $x \in [x_l^\ast, x_u^\ast]$, which implies that the roots of $\bar V^K_\delta (\cdot, P)$ constitute a closed interval. 
	As $\bar L_\delta^K(x,P)$ has a unique maximum for all $P \in \tilde{\mathcal{P}} \subset \mathcal{P}$ by part \ref{statement:GenModalMidpointWellDefined} and $\bar V^K_\delta(x, P) = - \frac{\partial}{\partial x} \bar L_\delta^K(x,P)$, it follows that $\bar V^K_\delta(x, P) $ must have a unique root, which implies that $V_\delta^K(x,Y)$ is a strict identification function for the generalized modal midpoint.
\end{proof}

\begin{proof}[Proof of Theorem \ref{thm:ModeRationality}]
	We define
	\begin{align}
		g_{t,T} &:= \delta_T^{3/2}  T^{-1/2} \psi(Y_{t+1},X_t,\bh, \delta_T)  
		=  - (T\delta_T)^{-1/2} K'\left( \frac{X_t-Y_{t+1}}{\delta_T} \right) \bh, \\
		g_{t,T}^e &:=  \mathbb{E}_t[g_{t,T} ], \qquad \text{ and } \qquad
		g_{t,T}^\ast :=g_{t,T} - g_{t,T}^e,
	\end{align}
	such that  
	\begin{align}
		\delta_T^{3/2} T^{-1/2}  \sum_{t=1}^T \psi(Y_{t+1},X_t,\bh, \delta_T) 
		= \sum_{t=1}^T g_{t,T}^e + \sum_{t=1}^T g_{t,T}^\ast. 
	\end{align}
	%	Lemma \ref{lemma:ExpectationPsiTo0} shows that $\sum_{t=1}^T g_{t,T}^e \toP 0$.
	%	Thus, it remains to shows that $\sum_{t=1}^T g_{t,T}^\ast \tod \mathcal{N} \big( 0, \Sigma \big)$.
	%	For this, we define $z_{t,T} = \lambda^\top g_{t,T}^\ast$ for some arbitrary $\lambda \in \mathbb{R}^k, \lambda \not= 0$ and show that a univariate CLT for MDA holds for $\frac{1}{\sqrt{T}} \sum_{t=1}^T z_{t,T}$.
	%	For this, we apply Corollary 5.26 of \cite{White2001}  where we have to show the following three conditions hold:
	%	\begin{enumerate}[label=(\roman*)]
		%		\item 
		%		\label{cond:WhiteCLTMoments}
		%		$\mathbb{E} | z_{t,T} |^{2+\delta} < \Delta < \infty$ for some $\delta > 0$ and for all $t, T \in \mathbb{N}$.
		%		
		%		\item 
		%		\label{cond:VarianceBoundedFromZero}
		%		there exists a $\delta' > 0$ such that for all $T$ sufficiently large enough,
		%		$\bar \delta_T^2 = \frac{1}{T} \sum_{t=1}^T \mathbb{E}  \left[ z_{t,T}^2 \right] > \delta' > 0$,
		%		
		%		\item
		%		\label{cond:WhiteCLTzSquaredLLN}
		%		$\frac{1}{T} \sum_{t=1}^T z_{t,T}^2 - \bar \delta_T^2 \toP 0$,
		%	\end{enumerate}
	%
	%	From Lemma \ref{lemma:ConvergenceAverageVariance}, we get that $\bar \omega_T^2 = \frac{1}{T} \sum_{t=1}^T \mathbb{E}  \left[ z_{t,T}^2 \right] \to \bar \omega^2 $.
	%	Thus, we can conclude that there exists a $\delta' > 0$ such that for all $T$ sufficiently large enough, $\bar \delta_T^2 > \delta'$, which shows condition (ii) of Corollary 5.26 of \cite{White2001}.
	%	Condition (iii) follows from Lemma \ref{lemma:LLNzSquared}.
	%	
	Lemma \ref{lemma:ExpectationPsiTo0} shows that $\sum_{t=1}^T g_{t,T}^e \toP 0$.
	Thus, it remains to shows that $\sum_{t=1}^T g_{t,T}^\ast \tod \mathcal{N} \big( 0, \Omega_{\mathrm{Mode}} \big)$.
	For some arbitrary, but fixed $\lambda \in \mathbb{R}^k, || \lambda ||_2 = 1$, we define 
	\begin{align}
		z_{t,T} := \lambda^\top g_{t,T}^\ast, \qquad
		\bar \omega_T^2 := \sum_{t=1}^T \Var(z_{t,T}),	 \qquad
		h_{t,T} := \frac{z_{t,T}}{\bar \omega_T},
		\qquad \text{ and } \qquad
		\omega^2 := \lambda^\top \Omega_\mathrm{Mode} \lambda,
	\end{align}
	and show that a univariate CLT for martingale difference arrays (MDA) holds for $\sum_{t=1}^T h_{t,T}$.
	%	We define $\mathcal{F}_{t,T} = \sigma \{ \mathcal{F}_t, \delta_T \}$.
	It obviously holds that $\big( g_{t,T}^\ast, \mathcal{F}_{t+1} \big)$ is a MDA as $g_{s,T}^\ast \in \mathcal{F}_{t+1}$ for all $s \le t$ and 
	$\mathbb{E}_t \big[  g_{t,T}^\ast \big] = 0$ a.s.\ by definition.
	Thus, $\big( z_{t,T}, \mathcal{F}_{t+1} \big)$ and $\big( h_{t,T}, \mathcal{F}_{t+1} \big)$ are also MDAs.
	
	In the following, we verify the following three conditions of Theorem 24.3 of \cite{davidson1994stochastic}:
	(a) $\sum_{t=1}^T \Var( h_{t,T}) = 1$, (b) $\sum_{t=1}^T h_{t,T}^2 \toP 1$,  and (c) $\operatorname{max}_{1 \le t \le T} |h_{t,T}| \toP 0$.
	%	\begin{enumerate}[label=(\alph*), leftmargin=2cm]
		%		\item 
		%		$\sum_{t=1}^T \Var( h_{t,T}) = 1$,
		%		
		%		\item
		%		$\sum_{t=1}^T h_{t,T}^2 \toP 1$, \quad and
		%		
		%		\item 
		%		$\operatorname{max}_{1 \le t \le T} |h_{t,T}| \toP 0$.
		%	\end{enumerate}
	% TD @ AJP: Here, I define the term $h_{t,T}$ which is used often in the proofs of the Lemmas in the supplement. This notation might be a little to close to the new notation of the instruments $h_t$. Do you think we should change the notation of $h_{t,T}$ here?
	% AJP @ TD: Yes, I think we should change the notation of $h_{t,T}$ here. It's OK for this draft, but when you have time please do change this. 
	Lemma \ref{lemma:ConvergenceAverageVariance} shows that $ \bar \omega_T^2 = \sum_{t=1}^T \Var[ z_{t,T} ] \to  \lambda^\top \Omega_\mathrm{Mode} \lambda = \omega^2$.
	Thus, as $\Omega_\text{Mode}$ is assumed to be positive definite, $\bar \omega_T^2$ is strictly positive for all $T$ sufficiently large and hence, $h_{t,T}$ is well-defined and $\sum_{t=1}^T \Var( h_{t,T}) = 1$, which shows condition (a).
	Lemma \ref{lemma:LLNzSquared} shows that $\sum_{t=1}^T z_{t,T}^2  \toP  \omega^2$ as $T \to \infty$ which implies condition (b), i.e. $\sum_{t=1}^T h_{t,T}^2 = \sum_{t=1}^T \frac{z_{t,T}^2}{\bar \omega_T^2}  \toP 1$.
	Eventually, Lemma \ref{lemma:MaxConvZero} shows condition (c) and we can apply Theorem 24.3 of \cite{davidson1994stochastic} in order to conclude that for all  $\lambda \in \mathbb{R}^k, || \lambda||_2 = 1$, it holds that $\sum_{t=1}^T h_{t,T} \tod \mathcal{N}(0,1)$.
	As $\bar \omega_T^2 \to \omega^2$, Slutsky's theorem implies that $\sum_{t=1}^T z_{t,T} =\sum_{t=1}^T \lambda^\top g_{t,T}^\ast \tod \mathcal{N} \big( 0, \omega^2 \big)$, and  as this holds for all $\lambda \in \mathbb{R}^k, || \lambda||_2 = 1$, we apply the Cram\'er-Wold theorem and get that $\sum_{t=1}^T g_{t,T}^\ast \tod \mathcal{N} \big( 0, \Omega_{\mathrm{Mode}} \big)$, which concludes the proof of this theorem.
\end{proof}

\begin{proof}[Proof of Theorem \ref{thm:ConsistencyCovEstimation}]
	Let $\lambda \in \mathbb{R}^k$, $||\lambda||_2 = 1$ be a fixed and deterministic vector.
	Then,
	\begin{align}
		\begin{aligned}
			\label{eqn:CovConstSplit}
			&\lambda^\top \widehat{\Omega}_{T,\mathrm{Mode}} \lambda - \lambda^\top \Omega_{T,\mathrm{Mode}}  \lambda \\
			%		= \, &\frac{1}{T} \sum_{t=1}^T \delta_T^{-1} K' \left( \frac{X_t-Y_{t+1}}{\delta_T} \right)^2  (\lambda^\top \bh)^2
			%		- \mathbb{E} \left[ (\lambda^\top \bh)^2 f_{t}(0) \int K'(u)^2 \mathrm{d}u \right] \\
			= \, &\frac{1}{T} \sum_{t=1}^T \delta_T^{-1} K' \left( \frac{X_t-Y_{t+1}}{\delta_T} \right)^2  (\lambda^\top \bh)^2
			- \frac{1}{T} \sum_{t=1}^T \mathbb{E}_t \left[  \delta_T^{-1} K' \left( \frac{X_t-Y_{t+1}}{\delta_T} \right)^2  (\lambda^\top \bh)^2 \right] \\
			+ \, &\frac{1}{T} \sum_{t=1}^T \mathbb{E}_t \left[  \delta_T^{-1} K' \left( \frac{X_t-Y_{t+1}}{\delta_T} \right)^2  (\lambda^\top \bh)^2 \right]
			- \frac{1}{T} \sum_{t=1}^T \mathbb{E} \left[ (\lambda^\top \bh)^2 f_{t}(0) \int K'(u)^2 \mathrm{d}u \right].
		\end{aligned}
	\end{align}
	We start by showing that the last line in (\ref{eqn:CovConstSplit}) is $o_P(1)$.
	%	As
	%	\begin{align}
		%	&\frac{1}{T} \sum_{t=1}^T \mathbb{E}_t \left[  \delta_T^{-1} K' \left( \frac{X_t-Y_{t+1}}{\delta_T} \right)^2  (\lambda^\top \bh)^2 \right]
		%	= \frac{1}{T} \sum_{t=1}^T  (\lambda^\top \bh)^2  \delta_T^{-1}  \int K' \left( \frac{e}{\delta_T} \right)^2 f_t(e) \mathrm{d}e \\
		%	= \, &\frac{1}{T} \sum_{t=1}^T (\lambda^\top \bh)^2 \int K' \left( u \right)^2 f_t(\delta_T u) \mathrm{d}u,
		%	\end{align}
	%	as $ f_t(\delta_T u) \to  f_t(0) \le c$, and by further applying a weak law of large numbers for stationary and ergodic data as $\mathbb{E} \left[ ||\bh||^{2+\delta} \right] < \infty$.
	It holds that
	\begin{align}
		&\frac{1}{T} \sum_{t=1}^T \mathbb{E}_t \left[  \delta_T^{-1} K' \left( \frac{X_t-Y_{t+1}}{\delta_T} \right)^2  (\lambda^\top \bh)^2 \right]
		- \frac{1}{T} \sum_{t=1}^T \mathbb{E} \left[ (\lambda^\top \bh)^2 f_{t}(0) \int K'(u)^2 \mathrm{d}u \right] \\
		=\;&\frac{1}{T} \sum_{t=1}^T (\lambda^\top \bh)^2 \int K' \left( u \right)^2 f_t(\delta_T u) \mathrm{d}u -\frac{1}{T} \sum_{t=1}^T \mathbb{E} \left[ (\lambda^\top \bh)^2 \int K'(u)^2  f_{t}(0) \mathrm{d}u \right]
		\toP 0
	\end{align}
	as $ f_t(\delta_T u) \to  f_t(0) \le c$ a.s., and by applying a law of large numbers for mixing data as $\mathbb{E} \left[ ||\bh||^{2r+\delta} \right] < \infty$.

	We further show that the penultimate line in (\ref{eqn:CovConstSplit}) converges to zero in $L_p$ ($p$-th mean) for some $p > 1$ small enough.
	By applying the \cite{BahrEsseen1965} inequality for MDA (notice that the summands of the sum in the first line of  \eqref{eqn:BahrEsseen} naturally form an MDA without imposing our null hypothesis), we get 
	\begin{align}
		\begin{aligned}
		\label{eqn:BahrEsseen}
		&\mathbb{E} \left[ \left| \frac{1}{T} \sum_{t=1}^T \delta_T^{-1} K' \left( \frac{\varepsilon_t}{\delta_T} \right)^2  (\lambda^\top \bh)^2
		- \frac{1}{T} \sum_{t=1}^T \mathbb{E}_t \left[  \delta_T^{-1} K' \left( \frac{\varepsilon_t}{\delta_T} \right)^2  (\lambda^\top \bh)^2 \right] \right|^p \right] \\
		\le \, & 2 T^{-p} \sum_{t=1}^T \mathbb{E} \left[  \left| \delta_T^{-1} K' \left( \frac{\varepsilon_t}{\delta_T} \right)^2  (\lambda^\top \bh)^2 \right|^p  \right]
		+ 2 T^{-p} \sum_{t=1}^T \mathbb{E}  \left[  \left|  \mathbb{E}_t \left[  \delta_T^{-1} K' \left( \frac{\varepsilon_t}{\delta_T} \right)^2  (\lambda^\top \bh)^2 \right] \right|^p \right].
		\end{aligned}
	\end{align}
	For the first term in \eqref{eqn:BahrEsseen},  we get that
	\begin{align*}
		&T^{-p} \sum_{t=1}^T \mathbb{E} \left[  \left| \delta_T^{-1} K' \left( \frac{\varepsilon_t}{\delta_T} \right)^2  (\lambda^\top \bh)^2 \right|^p  \right]
		= (T \delta_T)^{-p} \sum_{t=1}^T \mathbb{E} \left[  \left| \lambda^\top \bh \right|^{2p} \int \left| K' \left( \frac{e}{\delta_T} \right) \right|^{2p} f_t(e) \mathrm{d}e \right] \\
		= \, &(T \delta_T)^{1-p} \frac{1}{T} \sum_{t=1}^T  \mathbb{E} \left[ \left| \lambda^\top \bh \right|^{2p} \int \left| K' \left( u \right) \right|^{2} f_t(\delta_T u) \mathrm{d}u \right]
		\to 0,
	\end{align*}
	as $(T \delta_T)^{1-p} \to 0$ for any $p > 1$, $\mathbb{E} \left[ ||\bh||^{2p} \right] < \infty$ for $p > 1$ small enough, the density $f_t$ is bounded from above, and $\int |K'(u)|^2 \mathrm{d}u < \infty$ by assumption. 
	The second term in  \eqref{eqn:BahrEsseen} converges by a similar argument as further detailed in (\ref{eqn:DetailsSecondTermBahrEsseenIneq}) in the proof of Lemma \ref{lemma:LLNzSquared}.
	As $L_p$ convergence  for any $p > 1$ implies convergence in probability, the result of the theorem follows.
\end{proof}

\begin{proof}[Proof of Theorem \ref{thm:ModeRationalityLocalAlternative}]
	We let
	\begin{align*}
		g_{t,T} &:= \delta_T^{3/2}  T^{-1/2} \psi(Y_{t+1},X_t,\bh, \delta_T)  
		=  - (T\delta_T)^{-1/2} K'\left( \frac{X_t-Y_{t+1}}{\delta_T} \right) \bh, \\
		g_{t,T}^e &:=  \mathbb{E}_t[g_{t,T} ], \qquad \text{ and } \qquad
		g_{t,T}^\ast :=g_{t,T} - g_{t,T}^e,
	\end{align*}
	such that $\sum_{t=1}^T g_{t,T} = \sum_{t=1}^T g_{t,T}^e + \sum_{t=1}^T g_{t,T}^\ast$
	as in the proof of Theorem \ref{thm:ModeRationality}.
	Then, it holds that $\sum_{t=1}^T g_{t,T}^\ast  \tod \mathcal{N}(0,\Omega_{\mathrm{Mode}})$ as in the proof of Theorem \ref{thm:ModeRationality}.
	Notice for this that the employed Lemmas \ref{lemma:ConvergenceAverageVariance}--\ref{lemma:MaxConvZero} do not require the null hypothesis in \eqref{eqn:NullHypothesisMode}.
	
	We continue to analyze the term $\sum_{t=1}^T g_{t,T}^e$ as in the proof of Lemma \ref{lemma:ExpectationPsiTo0}. 
	However, the Taylor expansion in \eqref{eqn:ProofLemmaTaylorExpansion} now entails another term as we now impose $\mathbb{H}_{A, \text{loc}}$ in \eqref{eqn:LocalAlternativeHypothesis} instead of using $f_t'(0) = 0$ that resulted from $\mathbb{H}_{0}$ in \eqref{eqn:NullHypothesisMode}.
	Hence, using that $\int K( u ) \mathrm{d}u = 1$, we get
	\begin{align}
		\sum_{t=1}^T g_{t,T}^e 
		= T^{-1/2} \delta_T^{3/2} \sum_{t=1}^T  f'_{t} (0) \bh  + o_P(1)
		%		= T^{1/2} \delta_T^{3/2} \mathbb{E} \big[ f'_{t} (0) \bh \big] + o_P(1) \\
		&= c \cdot \left( a_T T^{1/2} \delta_T^{3/2} + o_P(1) \right) + o_P(1)
	\end{align}
	
	For statement \ref{enm:LocPower}, $a_T T^{1/2} \delta_T^{3/2} \toP 1$, and thus, $\sum_{t=1}^T g_{t,T}^e = c + o_P(1)$.
	As Theorem \ref{thm:ConsistencyCovEstimation} shows that $\widehat{\Omega}_{T,\mathrm{Mode}} - \Omega_{T, \textrm{Mode}} \toP 0$ (without imposing our null hypothesis in \eqref{eqn:NullHypothesisMode}), applying Slutzky's theorem implies
	$\widehat{\Omega}_{T,\mathrm{Mode}}^{-1/2} \sum_{t=1}^T g_{t,T} \tod \mathcal{N} \big( \Omega_{\mathrm{Mode}}^{-1/2} c, I_k \big)$.
	Hence, $J_T$ converges to a non-central $\chi^2$-distribution,
	\begin{align*}
		J_T 
		=  \left( \widehat{\Omega}_{T,\mathrm{Mode}}^{-1/2} \sum_{t=1}^T g_{t,T} \right)^\top  \left( \widehat{\Omega}_{T,\mathrm{Mode}}^{-1/2} \sum_{t=1}^T g_{t,T} \right)
		\tod \chi^2_k \big( c^\top \Omega_{\textrm{Mode}}^{-1} c \big).
	\end{align*}
	
	For statement \ref{enm:LocPower_Null}, $a_T T^{1/2} \delta_T^{3/2} = o_P(1)$ and hence, $\sum_{t=1}^T g_{t,T}^e  = o_P(1)$ putting us exactly in the situation of Theorem \ref{thm:ModeRationality}.
	
	For statement \ref{enm:LocPower_UnifPower}, it holds that $\left( a_T T^{1/2} \delta_T^{3/2} \right)^{-1} \toP 0$, which implies that for any $ \bar c \in \mathbb{R}$,
	$\mathbb{P} \left( \left| \sum_{t=1}^T g_{t,T}^e \right| \ge \bar c \right) \to 1$,
	and consequently also
	$\mathbb{P} \left( \left| \sum_{t=1}^T g_{t,T} \right| \ge \bar c \right) \to 1$.
	As furthermore  $J_T =  \left( \sum_{t=1}^T g_{t,T} \right)^\top  \widehat{\Omega}_{T,\mathrm{Mode}}^{-1} \left( \sum_{t=1}^T g_{t,T} \right)$ and $\widehat{\Omega}_{T,\mathrm{Mode}} - \Omega_{T,\mathrm{Mode}} \toP 0$ by Theorem \ref{thm:ConsistencyCovEstimation}, where $\Omega_{T,\mathrm{Mode}}$ is uniformly positive definite for $T$ large enough, the conditions of Theorem 8.13 of \cite{White1994} are satisfied and we can conclude that for any $\bar c \in \mathbb{R}$, $\mathbb{P} \left( \left| J_T \right| \ge \bar c \right) \to 1$, which concludes the proof of this theorem.
\end{proof}

\begin{proof}[Proof of Theorem \ref{thm:GMMWeakId}]
	For all fixed $\lambda \in \mathbb{R}^k$ such that $||\lambda||_2 = 1$, we define
	\begin{align}
		\begin{aligned}
			\label{eqn:Defsigma}
			\sigma_T^{2} :=  \lambda^\top \Sigma_T(\theta_0) \lambda 
			&= \frac{1}{T} \sum_{t=1}^T \mathbb{E} \left[ \theta_{10}^2 \left( \bh^\top  \WMean \lambda  \right)^2 \varepsilon_t^2 + \theta_{20}^2 \left( \bh^\top  \WMed  \lambda \right)^2 \big( \mathds{1}_{\{\varepsilon_t > 0\}} - \mathds{1}_{\{\varepsilon_t < 0 \}} \big)^2 \right. \\
			&\qquad \qquad \quad+ \theta_{30}^2 \left( \bh^\top \WMode  \lambda \right)^2 f_t(0) \int K'(u)^2 \mathrm{d}u \\
			&\qquad \qquad \quad \left. + 2 \theta_{10} \theta_{20}\left( \bh^\top \WMean \lambda  \right) \left( \bh^\top  \WMed \lambda \right)
			\varepsilon_t \big( \mathds{1}_{\{\varepsilon_t > 0\}} - \mathds{1}_{\{\varepsilon_t < 0 \}} \big) \right].
		\end{aligned}
	\end{align}
	%	and  $\sigma_T^2 = T^{-1}  \sum_{t=1}^T \Var \big( \phi_{t,T}^\ast (\theta_0) \lambda \big)$ for the MDA $ T^{-1} \phi_{t,T}^\ast (\theta_0)$.
	Lemma \ref{lemma:WeakIdGMM_SumVariance} shows that $\sum_{t=1}^T \Var \big(T^{-1/2}  \phi_{t,T}^\ast (\theta_0) \lambda \big) - \sigma_T^2  \to 0$.
	As $\sigma^2 := \lambda^\top \Sigma(\theta_0) \lambda$, which is strictly positive by assumption, is the limit of $\sigma_T^2$, the latter must be strictly positive for $T$ large enough and consequently, $\sigma_T^{-1}  T^{-1/2} \sum_{t=1}^T \phi_{t,T}^\ast (\theta_0)  \lambda$ is well-defined.
	In the following, we first show that $\sigma_T^{-1}  T^{-1/2} \sum_{t=1}^T \phi_{t,T}^\ast (\theta_0)  \lambda  \tod \mathcal{N} \big( 0, 1 \big)$ by applying Theorem 24.3 in \cite{davidson1994stochastic} and by verifying that the respective conditions hold (with $X_{t,T} = \sigma_T^{-1}  T^{-1/2} \phi_{t,T}^\ast (\theta_0) \lambda$).
	
	Lemma \ref{lemma:WeakIdGMM_phito1} shows that $T^{-1} \sum_{t=1}^T \big( \phi_{t,T}^\ast (\theta_0) \lambda \big)^2  - \sigma_T^{2} \toP 0$, which implies condition (a) of Theorem 24.3 of \cite{davidson1994stochastic}, i.e.\ $\sigma_T^{-2} T^{-1} \sum_{t=1}^T \big( \phi_{t,T}^\ast (\theta_0) \lambda \big)^2  \toP 1$.
	Lemma \ref{lemma:WeakIdGMM_maxto0} shows condition (b), i.e.\ $\max_{t=1,\dots,T} \big| \sigma_T^{-1} T^{-1/2} \phi_{t,T}^\ast (\theta_0)   \lambda \big|  \toP 0$.
	Thus, we can apply Theorem 24.3 of \cite{davidson1994stochastic} and conclude that $\sigma_T^{-1} T^{-1/2} \sum_{t=1}^T \phi_{t,T}^\ast (\theta_0)  \lambda \tod \mathcal{N} \big( 0, 1 \big)$.	
	
	As $\sigma_T^2$ has the limit $\sigma^2$, Slutsky's theorem implies that $T^{-1/2} \sum_{t=1}^T \phi_{t,T}^\ast (\theta_0)  \lambda \tod \mathcal{N}(0, \sigma^2)$ and as this holds for all $\lambda \in \mathbb{R}^k$ such that $||\lambda||_2 = 1$, we can conclude that $T^{-1/2} \sum_{t=1}^T \phi_{t,T}^\ast (\theta_0) \tod \mathcal{N} \big( 0, \Sigma(\theta_0) \big)$.
	Furthermore, $ || T^{-1/2} \sum_{t=1}^T (\phi_{t,T}^\ast (\theta_0)  -  \phi_{t,T} (\theta_0)) || = || T^{-1/2} \sum_{t=1}^T u_{t,T}(\theta_0) || \toP 0$ by Assumption \ref{assu:GMMWeakIDRegCond} (D) and 
	% \le T^{-1/2} \sum_{t=1}^T || u_{t,T}(\theta_0) || 
	% \ref{assu:GMMWeakIDDecomposition} and 
	\begin{align*}
		T^{-1/2} \sum_{t=1}^T \left( \widehat \phi_{t,T} (\theta_0)  -  \phi_{t,T} (\theta_0) \right)
		= 
		T^{-1/2} \sum_{t=1}^T   \theta_0 \cdot
		\begin{pmatrix}
			\bh^\top \left( \WMeanHat  - \WMean \right) \varepsilon_t \\
			\bh^\top  \left( \WMedHat  - \WMed  \right) \big( \mathds{1}_{\{\varepsilon_t > 0\}} - \mathds{1}_{\{\varepsilon_t < 0 \}} \big)   \\
			\bh^\top  \left(   \WModeHat  - \WMode \right) \delta_T^{-1/2}  K' \left( \frac{- \varepsilon_t}{\delta_T} \right)
		\end{pmatrix} 
		\toP 0,
	\end{align*}
	as it holds that $\WMeanHat \toP \WMean$, $\WMedHat \toP \WMed$ and $\WModeHat \toP \WMode$ by assumption.
	Hence we can conclude that $T^{-1/2}  \sum_{t=1}^T \widehat \phi_{t,T} (\theta_0) \tod \mathcal{N} \big( 0, \Sigma(\theta_0) \big)$, which concludes the proof of this theorem.
\end{proof}

\begin{proof}[Proof of Theorem \ref{thm:GMMWeakIDCovarianceConsistency}]
	For notational simplicity, we show consistency of the covariance estimator by considering the bilinear forms $\lambda^\top \left( \frac{1}{T} \sum_{t=1}^T \widehat \phi_{t,T} (\theta_0) \widehat \phi_{t,T} (\theta_0)^\top \right) \lambda$ and $\sigma_T^{2} = \lambda^\top \Sigma_T(\theta_0) \lambda$, given in (\ref{eqn:Defsigma}), for some arbitrary but fixed $\lambda \in \mathbb{R}^k$ such that $||\lambda||_2 = 1$.
	%	For this, we obtain that
	%	\begin{align}
		%		\sigma^{2} 
		%		&= \mathbb{E} \left[ \theta_{10}^2 \left( \bh^\top  \WMean \lambda  \right)^2 \varepsilon_t^2 + \theta_{20}^2 \left( \bh^\top  \WMed  \lambda \right)^2 \big( \mathds{1}_{\{\varepsilon_t > 0\}} - \mathds{1}_{\{\varepsilon_t < 0 \}} \big)^2 \right. \\
		%		&\qquad \quad+ \theta_{30}^2 \left( \bh^\top  \WMode \lambda \right)^2 f_t(0) \int K'(u)^2 \mathrm{d}u \\
		%		&\qquad \quad \left. + 2 \theta_{10} \theta_{20}\left( \bh^\top \WMean\lambda  \right) \left( \bh^\top  \WMed \lambda \right)
		%		\varepsilon_t \big( \mathds{1}_{\{\varepsilon_t > 0\}} - \mathds{1}_{\{\varepsilon_t < 0 \}} \big)  \right],
		%	\end{align}
	Then, we get that
	\begin{align}
		&\lambda^\top \left( \frac{1}{T} \sum_{t=1}^T \widehat \phi_{t,T} (\theta_0) \widehat \phi_{t,T} (\theta_0)^\top \right) \lambda 
		= \frac{1}{T} \sum_{t=1}^T \left\{ \theta_{10}^2 \left( \bh^\top \WMeanHat \lambda  \right)^2 \varepsilon_t^2 \right. \\
		&\qquad+ \theta_{20}^2 \left( \bh^\top \WMedHat  \lambda \right)^2 \big( \mathds{1}_{\{\varepsilon_t > 0\}} - \mathds{1}_{\{\varepsilon_t < 0 \}} \big)^2  
		+ \theta_{30}^2 \left( \bh^\top \WModeHat \lambda \right)^2 \delta_T^{-1}  K' \left( \frac{-\varepsilon_t}{\delta_T} \right)^2 \\
		&\qquad+ 2 \theta_{10} \theta_{20}\left( \bh^\top \WMeanHat \lambda  \right) \left( \bh^\top  \WMedHat \lambda \right) \varepsilon_t \big( \mathds{1}_{\{\varepsilon_t > 0\}} - \mathds{1}_{\{\varepsilon_t < 0 \}} \big) \\
		&\qquad+ 2 \theta_{10} \theta_{30} \left( \bh^\top \WMeanHat \lambda  \right) \left( \bh^\top  \WModeHat \lambda \right) \varepsilon_t  \ \delta_T^{-1/2}  K' \left( \frac{-\varepsilon_t}{\delta_T} \right) \\		
		&\qquad \left.+ 2 \theta_{20} \theta_{30} \left( \bh^\top \WMedHat \lambda  \right) \left( \bh^\top  \WModeHat \lambda \right)  \big( \mathds{1}_{\{\varepsilon_t > 0\}} - \mathds{1}_{\{\varepsilon_t < 0 \}} \big)   \ \delta_T^{-1/2}  K' \left( \frac{-\varepsilon_t}{\delta_T} \right) \right\}.
	\end{align}
	We show convergence in probability for the individual matrix components for the first term,
	%	\begin{align}
		%	&\frac{1}{T} \sum_{t=1}^T \theta_{10}^2 \left( \bh^\top \WMeanHat \lambda  \right)^2 \varepsilon_t^2 
		%	%		= \, &\frac{1}{T} \sum_{t=1}^T \theta_{10}^2  \sum_{i,j,\iota,l}\boldsymbol{h}_{t,i} \widehat{\boldsymbol{W}}_{T,\mathrm{Mean},ij} \lambda_j \widehat{\boldsymbol{W}}_{T,\mathrm{Mean},ij}  \boldsymbol{h}_{t,\iota} \widehat{\boldsymbol{W}}_{T,\mathrm{Mean},\iota l}\lambda_l  \, \varepsilon_t^2  \\
		%	=  \sum_{i,j,\iota,l} \widehat{\boldsymbol{W}}_{T,\mathrm{Mean},ij} \widehat{\boldsymbol{W}}_{T,\mathrm{Mean},\iota l}\, \frac{1}{T} \sum_{t=1}^T \theta_{10}^2 \boldsymbol{h}_{t,i}   \lambda_j \boldsymbol{h}_{t,\iota}   \lambda_l  \varepsilon_t^2 \\
		%	\toP \, & \sum_{i,j,\iota,l} {\boldsymbol{W}}_{\mathrm{Mean},ij}  {\boldsymbol{W}}_{\mathrm{Mean},\iota l} \, \mathbb{E} \left[ \theta_{10}^2 \boldsymbol{h}_{t,i}  \lambda_j \boldsymbol{h}_{t,\iota}   \lambda_l  \varepsilon_t^2 \right] 
		%	=  \mathbb{E} \left[ \theta_{10}^2 \left( \bh^\top \WMean \lambda  \right)^2 \varepsilon_t^2 \right].
		%	\end{align}
	\begin{align}
		&\frac{1}{T} \sum_{t=1}^T \theta_{10}^2 \left( \bh^\top \WMeanHat \lambda  \right)^2 \varepsilon_t^2 
		- \frac{1}{T} \sum_{t=1}^T  \mathbb{E} \left[ \theta_{10}^2 \left( \bh^\top \WMean \lambda  \right)^2 \varepsilon_t^2 \right] 
		\toP 0
	\end{align}
	Convergence of the remaining terms follows analogously by considering the terms component-wisely and by applying similar arguments as in Lemma \ref{lemma:WeakIdGMM_phito1}.
	%	and by using the decomposition in (\ref{eqn:SplitphisquaredGMMWeakID}).	
	%		From Lemma \ref{lemma:WeakIdGMM_phito1}, we obtain that $\sum_{t=1}^T \big( \phi_{t,T}^\ast (\theta_0) \lambda \big)^2  \toP \sigma^{2}$.
	%	
	%	Further notice that
	%	\begin{align}
		%	\phi_{t,T} (\theta_0) = \big(  \phi_{t,T} (\theta_0) -  \tilde \phi_{t,T} (\theta_0) \big) +  u_{t,T} + \phi_{t,T}^\ast (\theta_0).
		%	\end{align}
\end{proof}

\section{Technical Lemmas and Their Proofs} \label{sec:TechnicalProofs}

\begin{lemma}
	\label{lemma:ExpectationPsiTo0}
	Given Assumption \ref{assu:ModeRationality} and the null hypothesis in (\ref{eqn:NullHypothesisMode}), it holds that $\sum_{t=1}^T g_{t,T}^e \toP 0$.
\end{lemma}

\begin{proof}
	Applying integration by parts yields that
	\begin{align}
		g_{t,T}^e &= - \mathbb{E}_t \left[ (T \delta_T)^{-1/2}  K' \left( \frac{ \varepsilon_t }{\delta_T} \right) \bh \right]
		= - (T \delta_T)^{-1/2}  \bh \int   K' \left( \frac{ e }{\delta_T} \right)  f_{t} (e) \, \mathrm{d}e \\
		&= T^{-1/2}  \delta_T^{1/2} \bh \int   K \left( \frac{ e }{\delta_T} \right) f'_{t} (e) \, \mathrm{d}e 
		- T^{-1/2}  \delta_T^{1/2} \bh  \left[ K \left( \frac{ e }{\delta_T} \right)  f_{t} (e) \right]_{e=-\infty}^{e=\infty}.
	\end{align}
	%	By assumption \ref{assu:ConditionalDensity}, it holds that
	%	\begin{align}
		%	\left| \left| T^{-1/2}  \delta_T^{1/2} \bh K \left( \frac{ e }{\delta_T} \right)  f_{t} (e) \right| \right|
		%	\le c T^{-1/2}  \delta_T^{1/2} \left| \left| \bh K \left( \frac{ e }{\delta_T} \right) \right|  \right|.
		%	\end{align}
	As $\lim_{e \to \pm \infty} K(e) = 0$ and $f_t$ is bounded from above, the latter term is zero a.s.\ for all $t \le T$.
	%	\begin{align}
		%		T^{-1/2}  \delta_T^{1/2} \bh  \left[ K \left( \frac{ e }{\delta_T} \right)  f_{t} (e) \right]_{e=-\infty}^{e=\infty} = 0
		%	\end{align}
	By transformation of variables, it further holds that
	\begin{align}
		\label{eqn:PsiTransformationOfVariables}
		g_{t,T}^e 
		=  T^{-1/2}  \delta_T^{1/2} \bh \int   K \left( \frac{ e }{\delta_T} \right) f'_{t} (e) \, \mathrm{d}e 
		=  T^{-1/2}  \delta_T^{3/2} \bh \int   K \left( u \right) f'_{t} (\delta_T u) \, \mathrm{d}u.
	\end{align}
	A Taylor expansion of $f'_t (\delta_T u)$ around zero is given by
	\begin{align}
		\label{eqn:ProofLemmaTaylorExpansion}
		f'_{t} (\delta_T u) 
		= f'_{t} (0) 
		+(\delta_T u)  f''_{t} (0) 
		+\frac{(\delta_T u)^2}{2}  f'''_{t} (\zeta \delta_T u), 
	\end{align}
	for some $\zeta \in [0,1]$ and $f'_{t} (0) = 0$ holds under the null hypothesis specified in (\ref{eqn:NullHypothesisMode}).
	Consequently,
	\begin{align}
		\sum_{t=1}^T g_{t,T}^e 
		= \, & T^{-1/2}  \delta_T^{5/2} \sum_{t=1}^T  f''_{t} (0) \bh \int  u K \left( u \right)  \, \mathrm{d}u \\
		&+ \frac{1}{2} T^{-1/2}  \delta_T^{7/2} \sum_{t=1}^T \bh \int  u^2 K \left( u \right) f'''_{t} (\zeta \delta_T u) \, \mathrm{d}u.
	\end{align}
	As $\int  u K \left( u \right)  \, \mathrm{d}u = 0 $ by assumption \ref{assu:Kernel}, the first term is zero for all $T \in \mathbb{N}$.
	%	\begin{align}
		%		- T^{-1/2}  \delta_T^{5/2} \sum_{t=1}^T \bh f''_{t} (0) \int  u K \left( u \right)  \, \mathrm{d}u = 0.
		%	\end{align}
	As $\sup_x f'''_{t}(x) \le c$ by Assumption \ref{assu:ConditionalDensity} and $\int  u^2 K \left( u \right)  \mathrm{d}u \le c < \infty$ by Assumption \ref{assu:Kernel}, we obtain
	%	\begin{align}
		%	&- 0.5 T^{-1/2}  \delta_T^{7/2} \sum_{t=1}^T \bh \int  u^2 K \left( u \right) f'''_{t} (\zeta \delta_T u) \, \mathrm{d}u  \\
		%	\le \, &- 0.5 \big( T  \delta_T^{7} \big)^{1/2} \frac{1}{T} \sum_{t=1}^T  \bh \sup_x f'''_{t} (x) \int  u^2 K \left( u \right) \, \mathrm{d}u \\
		%	\le \, &- 0.5 c^2 \big( T  \delta_T^{7} \big)^{1/2} \frac{1}{T} \sum_{t=1}^T  \bh  \toP 0,
		%	\end{align}
	\begin{align}
		\frac{1}{2} T^{-1/2}  \delta_T^{7/2} \sum_{t=1}^T \bh \int  u^2 K \left( u \right) f'''_{t} (\zeta \delta_T u) \, \mathrm{d}u  
		\le \frac{1}{2} c^2 \big( T  \delta_T^{7} \big)^{1/2} \frac{1}{T} \sum_{t=1}^T  \bh  \toP 0,
	\end{align}
	as $T \delta_T^7 \to 0$ for $T \to \infty$ by Assumption \ref{assu:Bandwidth} and $ \frac{1}{T} \sum_{t=1}^T  \bh \toP \mathbb{E}[\bh]$ by a law of large numbers for $\alpha$-mixing sequences \citep[Corollary 3.48]{White2001}.
	The result of the lemma follows.
\end{proof}

\begin{lemma}
	\label{lemma:ConvergenceAverageVariance}
	Given Assumption \ref{assu:ModeRationality}, % and under the null hypothesis in (\ref{eqn:NullHypothesisMode}), 
	it holds that  
	$\sum_{t=1}^T \Var \left( z_{t,T} \right) \to \omega^2 = \lambda^\top \Omega_{\mathrm{Mode}} \lambda$.
	%	\begin{align}
		%	\sum_{t=1}^T \Var \left( z_{t,T} \right) \to \bar \omega^2, \qquad 	\text{as} \qquad T \to \infty.
		%	\end{align}
\end{lemma}

\begin{proof}
	We first observe that $	\Var \left(  z_{t,T}  \right) = \mathbb{E} \left[  \big( \lambda^\top \big( g_{t,T} - g_{t,T}^e \big) \big)^2  \right]$ as $\mathbb{E} \left[  \lambda^\top \big( g_{t,T} - g_{t,T}^e \big)  \right] = 0$.
	%	\begin{align}
		%	\Var \left[  z_{t,T}  \right]
		%	= \Var \left[  \lambda^\top \big( g_{t,T} - g_{t,T}^e \big)  \right]
		%	= \mathbb{E} \left[  \big( \lambda^\top \big( g_{t,T} - g_{t,T}^e \big) \big)^2  \right] -  \mathbb{E} \left[  \lambda^\top \big( g_{t,T} - g_{t,T}^e \big)  \right]^2.
		%	%		= \mathbb{E} \left[  \left( \lambda^\top g_{t,T} \right)^2  \right]
		%	%		- \mathbb{E} \left[  \left( \lambda^\top g_{t,T}^e \right)^2  \right].
		%	\end{align}
	%	The second term vanishes as $\mathbb{E} \left[  \lambda^\top \big( g_{t,T} - g_{t,T}^e \big)  \right] = \mathbb{E} \left[  \lambda^\top \big( \mathbb{E}_t \left[g_{t,T} \right] - g_{t,T}^e \big)  \right] = 0$.
	Hence,
	\begin{align}
		\label{eqn:VarianceSplit}
		\Var \left(  z_{t,T}  \right)
		= \mathbb{E} \left[  \left( \lambda^\top g_{t,T} \right)^2  \right]
		- \mathbb{E} \left[  \left( \lambda^\top g_{t,T}^e \right)^2  \right],
	\end{align}
	as $\mathbb{E} \left[  \big( \lambda^\top g_{t,T}^e \big) \cdot \big( \lambda^\top g_{t,T} \big)   \right] = 	\mathbb{E} \left[  \big( \lambda^\top g_{t,T}^e \big) \cdot \mathbb{E}_t \big[ \lambda^\top g_{t,T} \big]   \right]  = \mathbb{E} \left[  \big( \lambda^\top g_{t,T}^e \big)^2  \right]$. 
	%	
	%	
	%	\begin{align}
		%	\mathbb{E} \left[  \big( \lambda^\top \big( g_{t,T} - g_{t,T}^e \big) \big)^2  \right] 
		%	&= \mathbb{E} \left[  \big( \lambda^\top g_{t,T} \big)^2  \right] 
		%	+ \mathbb{E} \left[  \big( \lambda^\top g_{t,T}^e \big)^2  \right] 
		%	- 2 \mathbb{E} \left[  \big( \lambda^\top g_{t,T}^e \big) \cdot \big( \lambda^\top g_{t,T} \big)   \right] \\
		%	&= \mathbb{E} \left[  \big( \lambda^\top g_{t,T} \big)^2  \right] 
		%	- \mathbb{E} \left[  \big( \lambda^\top g_{t,T}^e \big)^2  \right], 
		%	\end{align}
	%	as 
	%	\begin{align}
		%	\mathbb{E} \left[  \big( \lambda^\top g_{t,T}^e \big) \cdot \big( \lambda^\top g_{t,T} \big)   \right] 
		%	= 	\mathbb{E} \left[  \big( \lambda^\top g_{t,T}^e \big) \cdot \mathbb{E}_t \big[ \lambda^\top g_{t,T} \big]   \right] 
		%	= \mathbb{E} \left[  \big( \lambda^\top g_{t,T}^e \big)^2  \right].
		%	\end{align}
	%	Thus, we get that
	%	\begin{align}
		%	\label{eqn:VarianceSplit}
		%	\Var\left[  z_{t,T}  \right]
		%	= \mathbb{E} \left[  \left( \lambda^\top g_{t,T} \right)^2  \right]
		%	- \mathbb{E} \left[  \left( \lambda^\top g_{t,T}^e \right)^2  \right].
		%	\end{align}
	For the first term in (\ref{eqn:VarianceSplit}), we obtain
	\begin{align}
		\mathbb{E} \left[  \left( \lambda^\top g_{t,T} \right)^2  \right]
		&=  \mathbb{E} \left[ (T \delta_T)^{-1}  (\lambda^\top \bh)^2 \mathbb{E}_t \left[ K'\left( \frac{X_t-Y_{t+1}}{\delta_T} \right)^2 \right] \right] \\
		&=  \mathbb{E} \left[ (T \delta_T)^{-1}  (\lambda^\top \bh)^2 \int K'\left( \frac{e}{\delta_T} \right)^2  f_t(e) \, \mathrm{d}e \right] \\ 
		&= \frac{1}{T} \, \mathbb{E} \left[  (\lambda^\top \bh)^2  \int K' \left( u \right)^2 f_{t} (\delta_T u) \,  \mathrm{d}u  \right].
		%		
		%		&=  \frac{1}{T} \int \int \delta_T^{-1}  (\lambda^\top h)^2 K' \left( \frac{ e }{\delta_T} \right)^2 \mathrm{d} P_{t} (e) \, \mathrm{d} P_{\bh}(h) \\
		%		&= \frac{1}{T} \int \int \delta_T^{-1}  (\lambda^\top h)^2 K' \left( \frac{ e }{\delta_T} \right)^2  f_{t} (e) \, \mathrm{d}e \, \mathrm{d} P_{\bh}(h) \\
		%		&= \frac{1}{T} \int \int (\lambda^\top h)^2 K' \left( u \right)^2 f_{t} (\delta_T u) \,  \mathrm{d}u \, \mathrm{d} P_{\bh}(w) \\
		%		&= \frac{1}{T} \int  (\lambda^\top h)^2 \left( \int K' \left( u \right)^2 f_{t} (\delta_T u) \,  \mathrm{d}u \right) \mathrm{d} P_{\bh}(h)).
	\end{align}
	%	As the distribution of $\varepsilon_t$ given $\bh$ is time-invariant as we assume stationarity, 
	%	As $\delta_T \to 0$ when $T \to \infty$, we obtain that $\int K' \left( u \right)^2 f_{t} (\delta_T u) \,  \mathrm{d}u \to  f_{t} (0) \int K' \left( u \right)^2  \,  \mathrm{d}u$ and thus,
	As $\delta_T \to 0$ when $T \to \infty$ and as  $\Omega_{T,\mathrm{Mode}} \to  \Omega_{\mathrm{Mode}}$ from Assumption \ref{assu:CovConvergence}, we get
	%	\begin{align}
		%		\label{eqn:EgConvToNorm}
		%		\sum_{t=1}^T \mathbb{E} \left[  \left( \lambda^\top g_{t,T} \right)^2  \right] 
		%		\to  \mathbb{E} \left[ (\lambda^\top \bh )^2  f_{t} (0) \right] \int K'\left( u \right)^2 \mathrm{d}u =  \lambda^\top  \mathbb{E} \left[ f_{t} (0) \bh \bh^\top \right] \lambda \int K'\left( u \right)^2 \mathrm{d}u .
		%	\end{align}	
	\begin{align}
		\begin{aligned}
			\label{eqn:EgConvToNorm}
			&\sum_{t=1}^T \mathbb{E} \left[  \left( \lambda^\top g_{t,T} \right)^2  \right] -  \lambda^\top \Omega_{\mathrm{Mode}} \lambda \\
			=\,&\sum_{t=1}^T \mathbb{E} \left[  \left( \lambda^\top g_{t,T} \right)^2  \right] - \frac{1}{T}  \sum_{t=1}^T \mathbb{E} \left[ (\lambda^\top \bh )^2  f_{t} (0) \right] \int K'\left( u \right)^2 \mathrm{d}u \\
			&\qquad + \lambda^\top \Omega_{T,\mathrm{Mode}} \lambda
			-  \lambda^\top \Omega_{\mathrm{Mode}} \lambda	\\
			= \;&  \frac{1}{T}  \sum_{t=1}^T \left(\mathbb{E} \left[  (\lambda^\top \bh)^2  \int K' \left( u \right)^2 f_{t} (\delta_T u) \,  \mathrm{d}u  \right] - \lambda^\top  \mathbb{E} \left[ f_{t} (0) \bh \bh^\top \right] \lambda \int K'\left( u \right)^2 \mathrm{d}u  \right) + o(1) \to 0.
		\end{aligned}
	\end{align}

	For the second term in (\ref{eqn:VarianceSplit}), inserting the equality in (\ref{eqn:PsiTransformationOfVariables}) yields
	\begin{align*}
		\left( \lambda^\top g_{t,T}^e \right)^2
		= \left( \delta_T^{3/2} T^{-1/2} (\lambda^\top \bh) \int K'\left( u \right) f'_{t} (\delta_T u)  \mathrm{d}u \right)^2 
		\le \delta_T^3 T^{-1} ||\lambda||^2 ||\bh||^2  \left| \int K'\left( u \right) f'_{t} (u \delta_T)  \mathrm{d}u \right|^2.
	\end{align*}
	As $\sup_x | f'_{t} (x) | \le c$ for all $t \in \mathbb{N}$ and $\left| \int K'\left( u \right) \mathrm{d}u \right| \le c$ by assumption, it holds that 
	\begin{align}
		\label{eqn:geConvToZero}
		\sum_{t=1}^T \mathbb{E}  \left[ \left( \lambda^\top g_{t,T}^e\right)^2 \right]
		&\le \delta_T^3 c^2 ||\lambda||^2  \left( \frac{1}{T} \sum_{t=1}^T \mathbb{E} \left[ ||\bh||^2 \right] \right)
		\to 0,
	\end{align}
	as $\delta_T^3 \to 0 $ as $T \to \infty$.
	The result of the lemma follows by combining (\ref{eqn:EgConvToNorm}) and (\ref{eqn:geConvToZero}).
\end{proof}

\begin{lemma}
	\label{lemma:LLNzSquared}
	Given Assumption \ref{assu:ModeRationality}, %and under the null hypothesis in (\ref{eqn:NullHypothesisMode}), 
	it holds that  $\sum_{t=1}^T z_{t,T}^2 \toP \omega^2 = \lambda^\top \Omega_{\mathrm{Mode}} \lambda$.
\end{lemma}

\begin{proof}
	We define
	\begin{align}
		h_{1,T} := \sum_{t=1}^T \left( z_{t,T}^2  - \mathbb{E}_t \left[ z_{t,T}^2  \right] \right) \qquad \text{ and } \qquad 
		h_{2,T} := \sum_{t=1}^T \mathbb{E}_t \left[ z_{t,T}^2  \right] - \omega^2,
	\end{align}
	such that $\sum_{t=1}^T z_{t,T}^2 - \bar \omega^2  = h_{1,T} + h_{2,T}$.
	We first show that $h_{1,T} \stackrel{L_p}{\longrightarrow} 0$ for some $ 1 < p < 2$ sufficiently small enough and thus $h_{1,T} \toP 0$.
	For this, first notice that $z_{t,T}^2  - \mathbb{E}_t \left[ z_{t,T}^2 \right]$ is a $\mathcal{F}_{t+1}$-MDA by definition.
	Thus, we can apply the \cite{BahrEsseen1965}-inequality for some $p \in (1,2)$ for MDA (in the first line) in order to conclude that
	\begin{align}
		\mathbb{E} \left[ \left| h_{1,T} \right|^p \right]
		&= \mathbb{E} \left[ \left| \sum_{t=1}^T z_{t,T}^2  - \mathbb{E}_t \left[ z_{t,T}^2  \right] \right|^p \right] 
		\le 2 \sum_{t=1}^T \mathbb{E} \left[ \left| z_{t,T}^2  - \mathbb{E}_t \left[ z_{t,T}^2 \right]  \right|^p \right] \\
		&\le 2 \sum_{t=1}^T 2^{p-1} \left(  \mathbb{E} \left[ \left| z_{t,T}^2   \right|^p \right] + \mathbb{E} \left[ \left| \mathbb{E}_t \left[ z_{t,T}^2 \right]  \right|^p \right] \right) 
		%	&\le 2^{p} \sum_{t=1}^T  \left( \mathbb{E} \left[ \left| z_{t,T}   \right|^{2p} \right] + \mathbb{E} \left[  \mathbb{E}_t \left[ \left| z_{t,T} \right|^{2p} \right] \right] \right) \\
		\le 2^{p+1} \sum_{t=1}^T \mathbb{E} \left[ \left| z_{t,T}   \right|^{2p} \right],
	\end{align}
	where we use in the second inequality that $|a+b|^p \le 2^{p-1} \big( |a|^p + |b|^p \big)$ for any $a,b \in \mathbb{R}$.
	Using the same inequality again yields
	\begin{align}
		\label{eqn:Ezt2p}
		\mathbb{E} \left[ \left| z_{t,T}  \right|^{2p} \right]
		= \mathbb{E} \left[ \left| \lambda^\top g_{t,T} - \lambda^\top g_{t,T}^e  \right|^{2p} \right]
		\le 2^{2p-1} \left( \mathbb{E} \left[ \left| \lambda^\top g_{t,T} \right|^{2p} \right] + \mathbb{E} \left[ \left| \lambda^\top g_{t,T}^e \right|^{2p} \right] \right).
	\end{align}
	Thus,
	\begin{align}
		&\sum_{t=1}^T \mathbb{E} \left[ \left| \lambda^\top g_{t,T} \right|^{2p} \right]
		%		&= (T \delta_T)^{-p} \, \mathbb{E} \left[  \left| \lambda^\top \bh  K'\left( \frac{ \varepsilon_t }{\delta_T} \right)  \right|^{2p} \right] \\
		%		&= (T \delta_T)^{-p}  \mathbb{E} \left[  \left| \lambda^\top \bh \right|^{2p}  \mathbb{E}_t \left[ \left| K'\left( \frac{ \varepsilon_t }{\delta_T} \right)  \right|^{2p} \right] \right] \\
		%		&= (T \delta_T)^{-p}  \mathbb{E} \left[  \left| \lambda^\top \bh \right|^{2p}  \int \left| K'\left( \frac{ e }{\delta_T} \right) \right|^{2p}  f_{t}(e) \mathrm{d}e \right] \\
		= (T \delta_T)^{1-p}  \frac{1}{T} \sum_{t=1}^T \mathbb{E} \left[  \left| \lambda^\top \bh \right|^{2p}  \int \left| K'\left( u \right) \right|^{2p}  f_{t}(\delta_T u) \mathrm{d}u \right] \to 0,
		%		\\
		%		&\le (T \delta_T)^{1-p}  ||\lambda||^{2p} \left( \frac{1}{T} \sum_{t=1}^T \mathbb{E} \left[  || \bh ||^{2p} \right] \right) \left( \int \left| K'\left( u \right) \right|^{2p}  f_{t}(\delta_T u) \mathrm{d}u \right) 
	\end{align}
	as $(T \delta_T)^{1-p} \to 0$ for all $p \in (1,2)$, $\mathbb{E} \left[  \left| \lambda^\top \bh \right|^{2p} \right] < \infty$ for $p > 1$ sufficiently small such that $2p \le 2 + \delta$ (for the $\delta > 0$ from Assumption \ref{assu:Moments}), and $\int \left| K'\left( u \right) \right|^{2p}  f_{t}(\delta_T u) \mathrm{d}u \le c c^{2p-1} \int \left| K'\left( u \right) \right| \mathrm{d}u  < \infty$ as $\int \left| K'\left( u \right) \right| \mathrm{d}u  < \infty$, $\sup_{u} \left| K'\left( u \right) \right|  \le c$ and $\sup_{x} f_{t}(x)  \le c$ a.s.\ by assumption.
	Similarly, we obtain
	\begin{align}
		\label{eqn:DetailsSecondTermBahrEsseenIneq}
		&\sum_{t=1}^T \mathbb{E} \left[ \left| \lambda^\top g_{t,T}^e \right|^{2p} \right]
		= \delta_T^{2p-1}  (T \delta_T)^{1-p} \frac{1}{T} \sum_{t=1}^T \mathbb{E} \left[  \left| \lambda^\top \bh \right|^{2p} \left|  \int K'\left( u \right) f_{t}(\delta_T u) \mathrm{d}u \right|^{2p} \right] \to 0,
	\end{align}
	which shows that $h_{1,T} \stackrel{L_p}{\longrightarrow} 0$ for some $p > 1$ sufficiently small which implies that $h_{1,T} \toP 0$.

	We continue by showing that $h_{2,T} \toP 0$.
	Using the same argument as in (\ref{eqn:VarianceSplit}), we split
	\begin{align}
		h_{2,T}
		= \sum_{t=1}^T \mathbb{E}_t \left[ z_{t,T}^2  \right] -  \omega^2
		= \sum_{t=1}^T \mathbb{E}_t \left[ (\lambda^\top g_{t,T})^2 \right] - \sum_{t=1}^T  (\lambda^\top g_{t,T}^e)^2 -  \omega^2.
	\end{align}
	Applying a transformation of variables yields
	\begin{align}
		\sum_{t=1}^T  (\lambda^\top g_{t,T}^e)^2 
		&= \delta_T  \frac{1}{T}  \sum_{t=1}^T  (\lambda^\top \bh)^2 \left( \int  K'(u) f_{t}(\delta_T u) \mathrm{d}u \right)^2 \\
		&\le \delta_T  \left( \frac{1}{T} \sum_{t=1}^T  (\lambda^\top \bh)^2 \right) \left( c \int |K'(u)| \mathrm{d}u \right)^2 
		\toP 0,
	\end{align}
	as $\delta_T \to 0$, $ \frac{1}{T} \sum_{t=1}^T  (\lambda^\top \bh)^2  = \mathbb{E} \left[ (\lambda^\top \bh)^2 \right] + o_P(1)$ and $\left(\int |K'(u)| \mathrm{d}u \right)^2 \le \int |K'(u)|^2 \mathrm{d}u < \infty$ by assumption. 
	Furthermore,
	%	\begin{align}
		%		&\sum_{t=1}^T \mathbb{E}_t \left[ \big( \lambda^\top g_{t,T} \big)^2 \right]
		%		= (T \delta_T)^{-1} \sum_{t=1}^T  (\lambda^\top \bh)^2 \mathbb{E}_t \left[ K'\left( \frac{ \varepsilon_t }{\delta_T} \right)^2 \right]  \\
		%		%	&= (T \delta_T)^{-1} \sum_{t=1}^T  (\lambda^\top \bh)^2  \int K'\left( \frac{ e }{\delta_T} \right)^2 f_{t}(e) \mathrm{d}e \\
		%		%	&= (T \delta_T)^{-1} \sum_{t=1}^T  (\lambda^\top \bh)^2  \int \delta_T K'(u)^2 f_{t}(\delta_T u) \mathrm{d}u \\
		%		= \, &\left( \frac{1}{T} \sum_{t=1}^T  (\lambda^\top \bh)^2 \right) \int  K'(u)^2 f_t(\delta_T u) \mathrm{d}u 
		%		\toP \mathbb{E} \left[  f_{t}(0)  (\lambda^\top \bh)^2 \right]  \int  K'(u)^2 \mathrm{d}u 
		%		%	= \lambda^\top  \mathbb{E} \left[ f_{t}(0)  \bh \bh^\top \right] \lambda  \int  K'(u)^2 \mathrm{d}u 
		%		= \bar \omega^2,
		%	\end{align}
	\begin{align}
		&\sum_{t=1}^T \mathbb{E}_t \left[ \big( \lambda^\top g_{t,T} \big)^2 \right] - 	 \omega^2 \\
		= \, &\left( \frac{1}{T} \sum_{t=1}^T  (\lambda^\top \bh)^2 \right) \int  K'(u)^2 f_t(\delta_T u) \mathrm{d}u  - \frac{1}{T}  \sum_{t=1}^T \mathbb{E} \left[  f_{t}(0)  (\lambda^\top \bh)^2 \right]  \int  K'(u)^2 \mathrm{d}u  
		+ \bar \omega_T^2 - \omega^2 
		\toP 0,
	\end{align}
	as for all $u \in \mathbb{R}$, $\frac{1}{T} \sum_{t=1}^T  (\lambda^\top \bh)^2 f_t(\delta_T u) - \frac{1}{T} \sum_{t=1}^T \mathbb{E} \left[  f_{t}(0)  (\lambda^\top \bh)^2 \right] \toP 0$.
	Thus, we find $h_{2,T} \toP 0$ and consequently $\sum_{t=1}^T z_{t,T}^2 -  \omega^2 \toP 0$, which concludes this proof.
\end{proof}

\begin{lemma}
	\label{lemma:MaxConvZero}
	Given Assumption \ref{assu:ModeRationality}, %and the null hypothesis in (\ref{eqn:NullHypothesisMode}), 
	it holds that  $\max_{1\le t \le T} | h_{t,T} | \toP 0$.
\end{lemma}

\begin{proof}
	Let $\zeta > 0$ and $\delta > 0$ (sufficiently small such that $\mathbb{E} \left[ ||\bh||^{2+\delta}  \right] < \infty$).
	Then,
	\begin{align}
		\begin{aligned}
			\label{eqn:MaxMarkovTrick}
			\mathbb{P} \left( \max_{1\le t \le T} | h_{t,T} | > \zeta \right)
			&= \mathbb{P} \left( \max_{1\le t \le T} | h_{t,T} |^{2+\delta} > \zeta^{2+\delta}  \right)
			\le \sum_{t=1}^T \mathbb{P} \left(  |h_{t,T}|^{2+\delta} > \zeta^{2+\delta}  \right) \\
			&\le \zeta^{-2-\delta}  \sum_{t=1}^T  \mathbb{E} \left[ |h_{t,T}|^{2+\delta} \right]
			= \zeta^{-2-\delta}   \bar \omega_T^{-2-\delta}  \sum_{t=1}^T \mathbb{E} \left[ |z_{t,T}|^{2+\delta} \right],
		\end{aligned}
	\end{align}
	by Markov's inequality.
	Employing the same steps as in the proof of Lemma \ref{lemma:LLNzSquared} following Equation (\ref{eqn:Ezt2p}) and replacing the exponent  ``$2p$''  by ``$2+ \delta$'' yields that $\sum_{t=1}^T \mathbb{E} \left[ |z_{t,T}|^{2+\delta} \right] \to 0$.
	As $\bar \omega_T \to \omega^2 > 0$, this directly implies that $\mathbb{P} \left( \max_{1\le t \le T} | h_{t,T} | > \zeta \right) \to 0$.

\end{proof}

\begin{lemma}
	\label{lemma:WeakIdGMM_SumVariance}
	Given Assumption \ref{assu:ModeRationality} and Assumption \ref{assu:GMMWeakIDRegCond}, for all $\lambda \in \mathbb{R}^k$ such that $||\lambda||_2 = 1$, it holds that $\sum_{t=1}^T  \Var \big( T^{-1/2} \phi_{t,T}^\ast (\theta_0)^\top \lambda \big) - \sigma_T^{2} \to 0$.
\end{lemma}

\begin{proof}
	As $T^{-1/2} \phi_{t,T}^\ast(\theta_0)$ is a $\mathcal{F}_{t+1}$-MDA, it holds that $\mathbb{E} \big[ T^{-1/2} \phi_{t,T}^\ast (\theta_0)^\top \lambda \big] = 0$ and thus, $\Var \big( T^{-1/2} \phi_{t,T}^\ast (\theta_0)^\top \lambda \big) = \mathbb{E} \left[ \big( T^{-1/2} \phi_{t,T}^\ast (\theta_0)^\top \lambda \big)^2 \right]$.
	We further have
	\begin{align}
		\label{eqn:SplitphisquaredGMMWeakID}
		& \quad \sum_{t=1}^T \mathbb{E} \left[ \big( T^{-1/2} \phi_{t,T}^\ast (\theta_0)^\top \lambda \big)^2 \right] \\
		%	&= \sum_{t=1}^T \mathbb{E} \left[ \left\{  \theta_{10}  T^{-1/2} \left(  \lambda^\top  \WMean \bh  \right) \varepsilon_t 
		%	+ \theta_{20} T^{-1/2} \left(  \lambda^\top \WMed \bh \right) \big( \mathds{1}_{\{\varepsilon_t > 0\}} - \mathds{1}_{\{\varepsilon_t < 0 \}} \big) \right. \right.  \\
		%	&\qquad \qquad \qquad + \left. \left. \theta_{30} T^{-1/2} \left( \lambda^\top \WMode \bh  \right) \delta_T^{-1/2} K' \left( \frac{-\varepsilon_t}{\delta_T} \right) + u_{t,T}(\theta_0) \lambda \right\}^{2} \right] \\
		&=  \frac{1}{T} \sum_{t=1}^T \mathbb{E} \left[ \theta_{10}^2 \left(  \lambda^\top  \WMean \bh  \right)^2 \varepsilon_t^2 \right] 
		+  \frac{1}{T} \sum_{t=1}^T \mathbb{E} \left[ \theta_{20}^2 \left(  \lambda^\top  \WMed \bh  \right)^2  \big( \mathds{1}_{\{\varepsilon_t > 0\}} - \mathds{1}_{\{\varepsilon_t < 0 \}} \big)^2 \right]   \nonumber \\
		&+  \frac{1}{T} \sum_{t=1}^T \mathbb{E} \left[ \theta_{30}^2 \left(  \lambda^\top  \WMode \bh  \right)^2   \delta_T^{-1} K' \left( \frac{- \varepsilon_t}{\delta_T} \right)^2  \right] \nonumber	\\
		&+ \frac{2}{T} \sum_{t=1}^T \mathbb{E} \left[ \theta_{10} \theta_{20} \left(  \lambda^\top \WMean \bh  \right) \left(  \lambda^\top  \WMed \bh  \right) \varepsilon_t \big( \mathds{1}_{\{\varepsilon_t > 0\}} - \mathds{1}_{\{\varepsilon_t < 0 \}} \big)  \right]  \nonumber \\
		&+ \frac{2}{T} \sum_{t=1}^T \mathbb{E} \left[ \theta_{10} \theta_{30} \left(  \lambda^\top  \WMean \bh  \right) \left(  \lambda^\top  \WMode \bh  \right) \varepsilon_t  \delta_T^{-1/2} K' \left( \frac{- \varepsilon_t}{\delta_T} \right) \right] \nonumber \\
		&+ \frac{2}{T} \sum_{t=1}^T \mathbb{E} \left[ \theta_{20} \theta_{30} \left(  \lambda^\top \WMed\bh  \right) \left(  \lambda^\top  \WMode \bh  \right) \big( \mathds{1}_{\{\varepsilon_t > 0\}} - \mathds{1}_{\{\varepsilon_t < 0 \}} \big)   \delta_T^{-1/2} K' \left( \frac{- \varepsilon_t}{\delta_T} \right) \right] \nonumber \\
		&+\frac{1}{T}  \sum_{t=1}^T \mathbb{E} \left[  (\lambda^\top u_{t,T}(\theta_0))^2  \right] 			
		- \frac{2}{T}  \sum_{t=1}^T \mathbb{E} \left[ \big( \lambda^\top u_{t,T}(\theta_0)\big) \big(  \lambda^\top \phi_{t,T}(\theta_0) \big) \right]. \nonumber
		%		&+ \frac{2}{T} \sum_{t=1}^T \mathbb{E} \left[ \theta_{10} \left(  \lambda^\top  \Omega_{T,\mathrm{Mean}}^{-1/2} \bh  \right) \varepsilon_t  u_{t,T}(\theta_0)  \right] \\
		%		&+ \frac{2}{T} \sum_{t=1}^T \mathbb{E} \left[ \theta_{20} \left(  \lambda^\top  \Omega_{T,\mathrm{Med}}^{-1/2} \bh  \right) (1- 2 \mathds{1}_{ \{ \varepsilon_t \le 0\}}) u_{t,T}(\theta_0)  \right] \\
		%		&+ \frac{2}{T} \sum_{t=1}^T \mathbb{E} \left[ \theta_{30} \left(  \lambda^\top  \Omega_{T,\mathrm{Mode}}^{-1/2} \bh  \right) \delta_T^{-1/2} K' \left( \frac{\varepsilon_t}{\delta_T} \right)  u_{t,T}(\theta_0)  \right],
	\end{align}
	For the last two terms, we have $T^{-1} \sum_{t=1}^T \mathbb{E} \left[  \big( \lambda^\top u_{t,T}(\theta_0) \big)^2  \right] \to 0$ and $T^{-1} \sum_{t=1}^T \mathbb{E} \left[ \big( \lambda^\top u_{t,T}(\theta_0) \big) \big( \lambda^\top \phi_{t,T}(\theta_0) \big) \right] \to 0$
	% COMMENT TIMO: L_r convergence implies L_s convergence for s < r!!!!!
	%	\begin{align}
		%		\sum_{t=1}^T \mathbb{E} \left[  \big( u_{t,T}(\theta_0) \lambda \big)^2  \right] \to 0, \qquad \text{ and } \qquad 
		%		\sum_{t=1}^T \mathbb{E} \left[ \big( u_{t,T}(\theta_0)\lambda \big) \big( \tilde  \phi_{t,T}(\theta_0) \lambda \big) \right] \to 0,
		%	\end{align}
	by Assumption \ref{assu:GMMWeakIDRegCond} (D).
	For the fifth  term,
	\begin{align}
		&\frac{2}{T} \sum_{t=1}^T \mathbb{E} \left[ \theta_{10} \theta_{30} \left(  \lambda^\top  \WMean \bh  \right) \left(  \lambda^\top  \WMode \bh  \right) \delta_T^{-1/2} \mathbb{E}_t \left[ \varepsilon_t  K' \left( \frac{- \varepsilon_t}{\delta_T} \right) \right] \right] \\
		= \, &- \frac{2}{T} \sum_{t=1}^T \mathbb{E} \left[ \theta_{10} \theta_{30} \left(  \lambda^\top  \WMean \bh  \right) \left(  \lambda^\top  \WMode \bh  \right)  \delta_T^{3/2} \int u K' \left( u \right)f_t(\delta_T u) \mathrm{d}u  \right] \to 0,
	\end{align}
	as $\delta_T^{3/2} \to 0$, $\int u K' \left( u \right) \mathrm{d}u < \infty$ and the respective moments are finite.
	The sixth term converges to zero by a similar argument by bounding $\big| \mathds{1}_{\{\varepsilon_t > 0\}} - \mathds{1}_{\{\varepsilon_t < 0 \}} \big| \le 1$.
	For the third term, it holds that
	\begin{align*}
		&\frac{1}{T} \sum_{t=1}^T \mathbb{E} \left[ \theta_{30}^2 \left(  \lambda^\top  \WMode \bh  \right)^2   \delta_T^{-1} K' \left( \frac{- \varepsilon_t}{\delta_T} \right)^2  \right]  - \frac{1}{T} \sum_{t=1}^T \mathbb{E} \left[ \theta_{30}^2 \left(  \lambda^\top  \WMode \bh  \right)^2 f_t(0)  \int K' \left( u \right)^2  \mathrm{d}u \right] \\
		= \, &\frac{1}{T} \sum_{t=1}^T \mathbb{E} \left[ \theta_{30}^2 \left(  \lambda^\top  \WMode \bh  \right)^2 \int  K' \left( u \right)^2 \big( f_t(\delta_T u) - f_t(0) \big) \mathrm{d}u  \right] \to 0
	\end{align*}
	The remaining first, second and fourth terms in \eqref{eqn:SplitphisquaredGMMWeakID} appear equivalently in $\sigma_T^2$, which concludes the proof of this lemma.
	%	The difference of the remaining first, second and fourth terms in \eqref{eqn:SplitphisquaredGMMWeakID} and the respective terms in $\sigma_T^2$ obviously converge to zero, which concludes the proof of this lemma.
\end{proof}

\begin{lemma}
	\label{lemma:WeakIdGMM_phito1}
	Given Assumption \ref{assu:ModeRationality} and Assumption \ref{assu:GMMWeakIDRegCond}, for all $\lambda \in \mathbb{R}^k$ such that $||\lambda||_2 = 1$, it holds that
	$T^{-1} \sum_{t=1}^T \big( \phi_{t,T}^\ast (\theta_0)^\top \lambda \big)^2  - \sigma_T^{2} \toP 0$.
\end{lemma}

\begin{proof}
	We apply the same factorization as in (\ref{eqn:SplitphisquaredGMMWeakID}) (however without the expectation operator).
	%	By applying a law of large numbers for stationary and ergodic sequences (Theorem 3.34 in \cite{White2001}), we obtain that
	By applying a law of large numbers for mixing sequences (Corollary 3.48 in \cite{White2001}), we obtain that
	\begin{align*}
		&\frac{1}{T} \sum_{t=1}^T  \left( \theta_{10} \left(   \bh^\top  \WMean \lambda  \right) \varepsilon_t \right)^2 
		- \frac{1}{T} \sum_{t=1}^T \mathbb{E} \left[ \left(  \theta_{10} \left(  \bh^\top  \WMean \lambda \right) \varepsilon_t \right)^2 \right] \toP 0,  \\
		&\frac{1}{T} \sum_{t=1}^T \left( \theta_{20} \left(  \bh^\top  \WMed  \lambda \right)  \big( \mathds{1}_{\{\varepsilon_t > 0\}} - \mathds{1}_{\{\varepsilon_t < 0 \}} \big) \right)^2 
		- \frac{1}{T} \sum_{t=1}^T \mathbb{E} \left[ \left( \theta_{20} \left(  \bh^\top  \WMed  \lambda \right)  \big( \mathds{1}_{\{\varepsilon_t > 0\}} - \mathds{1}_{\{\varepsilon_t < 0 \}} \big) \right)^2  \right] \toP 0, \\
		&\frac{2}{T} \sum_{t=1}^T \left(  \theta_{10} \theta_{20} \left(  \bh^\top  \WMean \lambda  \right) \left(  \bh^\top  \WMed \lambda   \right) \varepsilon_t \big( \mathds{1}_{\{\varepsilon_t > 0\}} - \mathds{1}_{\{\varepsilon_t < 0 \}} \big)  \right)^2  \\
		&\qquad - \frac{2}{T} \sum_{t=1}^T \mathbb{E} \left[ \left( \theta_{10} \theta_{20} \left(  \bh^\top  \WMean \lambda   \right) \left(   \bh^\top  \WMed  \lambda \right) \varepsilon_t \big( \mathds{1}_{\{\varepsilon_t > 0\}} - \mathds{1}_{\{\varepsilon_t < 0 \}} \big)  \right)^2 \right] \toP 0.
	\end{align*}
	Furthermore, a slight modification of Lemma \ref{lemma:LLNzSquared}  (multiplying with $\theta_{30}^2 \big( \bh^\top  \WMode \lambda  \big)^2$ instead of $ \big( \bh^\top \lambda  \big)^2$) yields that
	\begin{align}
		\frac{1}{T} \sum_{t=1}^T \theta_{30}^2 \left( \bh^\top  \WMode \lambda  \right)^2   \delta_T^{-1} K' \left( \frac{-\varepsilon_t}{\delta_T} \right)^2  
		- \frac{1}{T} \sum_{t=1}^T \mathbb{E} \left[ \theta_{30}^2 \left( \bh^\top  \WMode \lambda \right)^2  f_t(0) \int K' \left( u \right)^2 \mathrm{d} u  \right] \toP 0.
	\end{align}

	We now show that the remaining four terms vanish asymptotically (in probability).
	For the mixed mean/mode term, we apply a similar addition of a zero (adding and subtracting $\mathbb{E}_t[\dots]$) as in the proof of Lemma \ref{lemma:LLNzSquared}.
	For this, we first note that
	\begin{align}
		&\frac{2}{T} \sum_{t=1}^T\theta_{10} \theta_{30} \left( \bh^\top  \WMean  \lambda \right) \left(   \bh^\top  \WMode \lambda \right)  \delta_T^{-1/2} \mathbb{E}_t \left[ \varepsilon_t  K' \left( \frac{-\varepsilon_t}{\delta_T} \right) \right] \\
		%	= \, &-\frac{2}{T} \sum_{t=1}^T\theta_{10} \theta_{30} \left( \bh^\top  \WMean  \lambda \right) \left(   \bh^\top  \WMode  \lambda \right)  \delta_T^{-1/2} \int e K' \left( \frac{e}{\delta_T} \right) f_t(e) \mathrm{d}e \\
		= -\,&\frac{2}{T} \sum_{t=1}^T\theta_{10} \theta_{30} \left( \bh^\top  \WMean  \lambda \right) \left(   \bh^\top  \WMode  \lambda \right)  \delta_T^{3/2} \int u K' \left( u \right) f_t(\delta_T u) \mathrm{d}u \toP 0,
	\end{align}
	as $\delta_T^{3/2}  \to 0$.
	In the following, we further show that
	\begin{align*}
		\frac{2}{T} \sum_{t=1}^T\theta_{10} \theta_{30} \left( \bh^\top  \WMean  \lambda \right) \left(   \bh^\top  \WMode \lambda \right)  \delta_T^{-1/2} \left\{  \varepsilon_t  K' \left( \frac{-\varepsilon_t}{\delta_T} \right) - \mathbb{E}_t \left[ \varepsilon_t  K' \left( \frac{-\varepsilon_t}{\delta_T} \right) \right] \right\} \stackrel{L_p}{\longrightarrow} 0,
	\end{align*}
	for any $p \in (1,2)$ small enough.
	As in the proof of Lemma \ref{lemma:LLNzSquared}, we apply the \cite{BahrEsseen1965} inequality and Minkowski's inequality in order to conclude that
	\begin{align*}
		&\mathbb{E} \left[ \left| \frac{2}{T} \sum_{t=1}^T\theta_{10} \theta_{30} \left( \bh^\top  \WMean  \lambda \right) \left(   \bh^\top  \WMode \lambda \right)  \delta_T^{-1/2} \left\{  \varepsilon_t  K' \left( \frac{\varepsilon_t}{\delta_T} \right) - \mathbb{E}_t \left[ \varepsilon_t  K' \left( \frac{\varepsilon_t}{\delta_T} \right) \right] \right\} \right|^p \right] \\
		\le \, & \frac{2^{p+2}}{T^p} \sum_{t=1}^T \left\{ \mathbb{E} \left[ \left| \theta_{10} \theta_{30} \left( \bh^\top  \WMean \lambda \right) \left(   \bh^\top  \WMode  \lambda \right)  \delta_T^{-1/2}  \varepsilon_t  K' \left( \frac{\varepsilon_t}{\delta_T} \right)  \right|^p \right] \right. \\
		&\qquad \qquad \; + \left. \mathbb{E} \left[ \left| \theta_{10} \theta_{30} \left( \bh^\top  \WMean \lambda \right) \left(   \bh^\top  \WMode  \lambda \right)  \delta_T^{-1/2} \mathbb{E}_t \left[  \varepsilon_t  K' \left( \frac{\varepsilon_t}{\delta_T} \right)  \right] \right|^p  \right] \right\},
	\end{align*}
	where the first term is bounded from above by
	\begin{align*}
		& \frac{2^{p+2}}{T^p} \sum_{t=1}^T  \mathbb{E} \left[ \left| \theta_{10} \theta_{30} \left( \bh^\top  \WMean \lambda \right) \left(   \bh^\top  \WMode \lambda \right) \right|^p  \delta_T^{-p/2} \int \left| e  K' \left( \frac{e}{\delta_T} \right)  \right|^p f_t(e) \mathrm{d}e  \right] \\
		= \, & 2^{p+2}  \delta_T^{1+p/2} T^{1-p} \frac{1}{T} \sum_{t=1}^T  \mathbb{E} \left[ \left| \theta_{10} \theta_{30} \left( \bh^\top  \WMean \lambda \right) \left(   \bh^\top  \WMode  \lambda \right) \right|^p   \int \left| u  K' \left( u \right)  \right|^p f_t(\delta_T u)  \mathrm{d}u  \right] \to 0,
	\end{align*}
	as $\delta_T^{1+p/2}\to 0 $, $T^{1-p} \to 0$ for any $p > 1$ and as the respective moments are bounded by assumption.
	Similar arguments also yield that the second term converges to zero (compare to (\ref{eqn:DetailsSecondTermBahrEsseenIneq})).
	Applying the same line of reasoning for the mixed median/mode terms shows that
	\begin{align}
		\frac{2}{T} \sum_{t=1}^T\theta_{20} \theta_{30} \left( \bh^\top  \WMed  \lambda \right) \left(   \bh^\top  \WMode \lambda \right)  \delta_T^{-1/2} \mathbb{E}_t \left[ \big( \mathds{1}_{\{\varepsilon_t > 0\}} - \mathds{1}_{\{\varepsilon_t < 0 \}} \big)  K' \left( \frac{-\varepsilon_t}{\delta_T} \right) \right] \toP 0.
	\end{align}
	For the fourth and last term,
	$T^{-1} \sum_{t=1}^T \big( u_{t,T}(\theta_0)^\top \lambda \big)^2  \toP 0$ and 
	$T^{-1} \sum_{t=1}^T \big( u_{t,T}(\theta_0)^\top \lambda \big) \big(  \phi_{t,T}(\theta_0)^\top \lambda \big) \toP 0$
	%	\begin{align}
		%		\sum_{t=1}^T \big( u_{t,T}(\theta_0) \lambda \big)^2  \toP 0, \qquad \text{ and } \qquad
		%		\sum_{t=1}^T \big( u_{t,T}(\theta_0)\lambda \big) \big( \tilde  \phi_{t,T}(\theta_0) \lambda \big) \toP 0,
		%	\end{align}
	by Assumption \ref{assu:GMMWeakIDRegCond} (D), which concludes this proof.
\end{proof}

\begin{lemma}
	\label{lemma:WeakIdGMM_maxto0}
	Given Assumption \ref{assu:ModeRationality} and Assumption \ref{assu:GMMWeakIDRegCond}, for all $\lambda \in \mathbb{R}^k$ such that $||\lambda||_2 = 1$, it holds that
	$\max_{1 \le t \le T} \big| \sigma_T^{-1} T^{-1/2} \phi_{t,T}^\ast (\theta_0)^\top \lambda \big|  \toP 0$.
\end{lemma}

\begin{proof}
	Let $\zeta > 0$ and $\delta > 0$ (sufficiently small such that $\mathbb{E} \left[ ||\bh||^{2+\delta}  \right] < \infty$ holds).
	Then, as in (\ref{eqn:MaxMarkovTrick}) in the proof of Lemma \ref{lemma:MaxConvZero}, we get that
	\begin{align}
		\mathbb{P} \left( \max_{1\le t \le T}  \big| \sigma_T^{-1} T^{-1/2} \phi_{t,T}^\ast (\theta_0)^\top \lambda \big| > \zeta \right)
		%		= \mathbb{P} \left( \max_{1\le t \le T}  \big| \sigma_T^{-1} \phi_{t,T}^\ast (\theta_0) \lambda \big|^{2+\delta} > \zeta^{2+\delta}  \right) 
		%	\le \, &\sum_{t=1}^T \mathbb{P} \left(   \big| \sigma_T^{-1} \phi_{t,T}^\ast (\theta_0) \lambda \big|{2+\delta} > \zeta^{2+\delta}  \right) 
		\le \zeta^{-2-\delta} \sigma_T^{-2-\delta}  \sum_{t=1}^T  \mathbb{E} \left[ \big| T^{-1/2} \phi_{t,T}^\ast (\theta_0)^\top \lambda \big|^{2+\delta} \right],
	\end{align}
	by Markov's inequality.
	Furthermore, we get that
	\begin{align}
		4^{-2-\delta} \sum_{t=1}^T \mathbb{E} \left[ \big| T^{-1/2} \phi_{t,T}^\ast (\theta_0)^\top \lambda  \big| ^{2+\delta} \right] 
		%	= \,  & \sum_{t=1}^T \mathbb{E} \left[ \left| T^{-\frac{1}{2}}
		%	\left(  \bh^\top  \WMean \lambda \right)  \theta_{10} \varepsilon_t 
		%	+ T^{-\frac{1}{2}} \left(  \bh^\top  \WMed \lambda \right) \theta_{20} \big( \mathds{1}_{\{\varepsilon_t > 0\}} - \mathds{1}_{\{\varepsilon_t < 0 \}} \big) \right. \right. \\
		%	&\qquad \qquad + \left. \left. T^{-\frac{1}{2}} \left(  \bh^\top  \WMode \lambda   \right) \theta_{30} \delta_T^{-1/2} K' \left( \frac{\varepsilon_t}{\delta_T} \right) + u_{t,T}(\theta_0) \right|^{2+\delta} \right] \\
		& \le \, \sum_{t=1}^T \mathbb{E} \left[ \left| T^{-1/2} u_{t,T}(\theta_0)^\top \lambda \right|^{2+\delta} \right] \\
		& + \theta_{10}^{2+\delta} T^{-\frac{\delta}{2}}  \frac{1}{T}   \sum_{t=1}^T \mathbb{E} \left[ \left| \bh^\top  \WMean \lambda  \right|^{2+\delta}  |\varepsilon_t|^{2+\delta} \right] \\
		& + \theta_{20}^{2+\delta} T^{-\frac{\delta}{2}}  \frac{1}{T}   \sum_{t=1}^T \mathbb{E} \left[ \left| \bh^\top  \WMed \lambda \right|^{2+\delta}  \big| \mathds{1}_{\{\varepsilon_t > 0\}} - \mathds{1}_{\{\varepsilon_t < 0 \}} \big|^{2+\delta}  \right] \\
		& + \theta_{30}^{2+\delta} T^{-\frac{\delta}{2}}  \frac{1}{T}   \sum_{t=1}^T \mathbb{E} \left[ \left| \bh^\top  \WMode \lambda   \right|^{2+\delta} \delta_T^{-\frac{2+\delta}{2}}  \left|  K' \left( \frac{-\varepsilon_t}{\delta_T} \right) \right|^{2+\delta} \right]. 
	\end{align}
	The first term converges to zero by Assumption \ref{assu:GMMWeakIDRegCond}.
	%	\begin{align}
		%		= \, & \sigma_T^{-2 - \delta} T^{-\frac{2+\delta}{2}}  \sum_{t=1}^T \mathbb{E} \left[ (\lambda^\top \bh)^2  \left\{ \theta_{10}^{2+\delta} |\varepsilon_t|^{2+\delta} +  \theta_{20}^{2+\delta} |1 - 2 \mathds{1}_{ \{ \varepsilon_t \le 0\}}|^{2+\delta} +  \theta_{30}^{2+\delta} \delta_T^{-2-\delta}  \left|  K' \left( \frac{\varepsilon_t}{\delta_T} \right) - \mathbb{E}_t \left[ K' \left( \frac{\varepsilon_t}{\delta_T} \right) \right] \right|^{2+\delta}  \right\}  \right] \\
		%		= \, & \sigma_T^{-2 - \delta} T^{-\frac{2+\delta}{2}}  \sum_{t=1}^T \mathbb{E} \left[ (\lambda^\top \bh)^2 \theta_{10}^{2+\delta} |\varepsilon_t|^{2+\delta} \right] \\
		%		+\, & \sigma_T^{-2 - \delta} T^{-\frac{2+\delta}{2}}  \sum_{t=1}^T \mathbb{E} \left[ (\lambda^\top \bh)^2  \theta_{20}^{2+\delta} |1 - 2 \mathds{1}_{ \{ \varepsilon_t \le 0\}}|^{2+\delta}  \right] \\
		%		+ \, & \sigma_T^{-2 - \delta} T^{-\frac{2+\delta}{2}}  \sum_{t=1}^T \mathbb{E} \left[ (\lambda^\top \bh)^2 \theta_{30}^{2+\delta} \delta_T^{-2-\delta}  \left|  K' \left( \frac{\varepsilon_t}{\delta_T} \right) - \mathbb{E}_t \left[ K' \left( \frac{\varepsilon_t}{\delta_T} \right) \right] \right|^{2+\delta} \right].
		%	\end{align}
	The second and third term converge to zero as $T^{-\frac{\delta}{2}} \to 0$ and the respective moments are bounded by assumption.
	%	\begin{align}
		%	\theta_{10}^{2+\delta} T^{-\frac{\delta}{2}} \frac{1}{T}  \sum_{t=1}^T \mathbb{E} \left[ \left| \bh^\top  \WMean \lambda  \right|^{2+\delta}  |\varepsilon_t|^{2+\delta} \right] \to 0
		%	\end{align}
	%	as $T^{-\frac{\delta}{2}} \to 0$ and the respective moment is bounded by assumption.
	For the last term, we obtain convergence equivalently to the proof of Lemma \ref{lemma:MaxConvZero},
	\begin{align}
		&\theta_{30}^{2+\delta}  T^{-\frac{\delta}{2}}  \frac{1}{T}   \sum_{t=1}^T \mathbb{E} \left[ \left| \bh^\top  \WMode \lambda   \right|^{2+\delta} \delta_T^{-\frac{2+\delta}{2}}  \left|  K' \left( \frac{-\varepsilon_t}{\delta_T} \right) \right|^{2+\delta} \right] \\
		%		\le \,  &\theta_{30}^{2+\delta} T^{-\frac{2+\delta}{2}}  \sum_{t=1}^T \mathbb{E} \left[ \left| \bh^\top  \WMode \lambda   \right|^{2+\delta} \delta_T^{-\frac{2+\delta}{2}} \int \left|  K' \left( \frac{e}{\delta_T} \right) \right|^{2+\delta} f_t(e) \mathrm{d}e \right] \\	
		%		\le \,  &\theta_{30}^{2+\delta} T^{-\frac{\delta}{2}} \frac{1}{T} \sum_{t=1}^T \mathbb{E} \left[ \left| \bh^\top  \WMode  \lambda \right|^{2+\delta} \delta_T^{-\frac{\delta}{2}} \int \left|  K' \left( u \right) \right|^{2+\delta} f_t(\delta_T u) \mathrm{d}u \right] \\
		\le \,  &\theta_{30}^{2+\delta} (T \delta_T)^{-\frac{\delta}{2}} \frac{1}{T} \sum_{t=1}^T \mathbb{E} \left[ \left| \bh^\top  \WMode \lambda   \right|^{2+\delta} \int \left|  K' \left( u \right) \right|^{2+\delta}f_t(\delta_T u) \mathrm{d}u \right],
	\end{align}
	that converges to zero as $(T \delta_T)^{-\frac{\delta}{2}} \to 0$ and the respective moments are bounded by assumption.
\end{proof}

\begin{lemma}
	\label{lem:Ass_TrueMode} 
	If $X_t$ is the mode of $F_t$ for all $t \in \mathbb{N}$, the choice of $u_{t,T}(\theta_0)$ in \eqref{eqn:u_Mode} satisfies Assumption~\ref{assu:GMMWeakIDRegCond}~(D).
\end{lemma}

\begin{proof}
	If $X_t$ is the mode of $F_t$, we set $\theta_0 = (0,0,1)$ and thus, $\phi_{t,T}(\theta_0) = - \omega_\text{Mode} \delta_T^{-1/2} K'(\varepsilon_t / \delta_T) \bh$ and we get $\phi_{t,T}(\theta_0) = \phi_{t,T}^\ast(\theta_0) +  u_{t,T}(\theta_0)$ by setting ~
	\begin{align*}
		T^{-1/2} \phi_{t,T}^\ast(\theta_0) = \omega_\text{Mode} \, g_{t,T}^\ast
		\qquad \text{ and } \qquad
		T^{-1/2} u_{t,T}(\theta_0) = \omega_\text{Mode} \, g_{t,T}^e,
	\end{align*}
	as in the proof of Theorem \ref{thm:ModeRationality}.
	Thus, $T^{-1/2} \phi_{t,T}^\ast(\theta_0) = \omega_\text{Mode} \, g_{t,T}^\ast$ is a MDA satisfying condition (a).
	
	For the remaining conditions (b) and (c) of Assumption \ref{assu:GMMWeakIDRegCond} (D), as in the proof of Lemma \ref{lemma:ExpectationPsiTo0}, we get
	\begin{align*}
		g_{t,T}^e =  \frac{1}{2} T^{-1/2}  \delta_T^{7/2} \bh \int  u^2 K \left( u \right) f'''_{t} (\zeta \delta_T u) \, \mathrm{d}u.
	\end{align*}
	As $\delta_T \to 0$ and $\int  u^2 K \left( u \right) f'''_{t} (\zeta \delta_T u)$ is bounded as argued in the proof of Lemma \ref{lemma:ExpectationPsiTo0}, we get
	\begin{align*}
		T^{-1}  \sum_{t=1}^T || u_{t,T}(\theta_0) ||^2
		=  \sum_{t=1}^T ||  \omega_\text{Mode} \, g_{t,T}^e ||^2
		= \delta_T^{7} \frac{\omega_\text{Mode}^2}{4} \, \frac{1}{T} \sum_{t=1}^T \left| \left|   \bh \int  u^2 K \left( u \right) f'''_{t} (\zeta \delta_T u) \, \mathrm{d}u \right|\right|^2 \toP 0.
	\end{align*}
	Furthermore, by using similar arguments as above,
	%	$T^{-1} \sum_{t=1}^T  u_{t,T}(\theta_0) \phi_{t,T}(\theta_0)^\top
	%	= T^{-1} \sum_{t=1}^T  u_{t,T}(\theta_0) u_{t,T} (\theta_0)^\top +  T^{-1} \sum_{t=1}^T  u_{t,T}(\theta_0)  \phi_{t,T}^\ast(\theta_0)^\top$, and for the latter term, 
	we get that
	\begin{align*}
		T^{-1} \sum_{t=1}^T  u_{t,T}(\theta_0)  \phi_{t,T}(\theta_0)^\top
		&= - \omega_\text{Mode}^2 \sum_{t=1}^T g_{t,T}^e \cdot (T \delta_T)^{-1/2}  K'(\varepsilon_t / \delta_T) \bh^\top \\
		&= - \omega_\text{Mode}^2 \sum_{t=1}^T  \left( \frac{1}{2} T^{-1/2}  \delta_T^{7/2} \bh \int  u^2 K \left( u \right) f'''_{t} (\zeta \delta_T u) \, \mathrm{d}u \right) (T \delta_T)^{-1/2}  K'(\varepsilon_t / \delta_T) \bh^\top \\
		&=  \delta_T^3 \frac{\omega_\text{Mode}^2}{2} \frac{1}{T} \sum_{t=1}^T  K'(\varepsilon_t / \delta_T)   \int  u^2 K \left( u \right) f'''_{t} (\zeta \delta_T u) \, \mathrm{d}u \; \bh\bh^\top \toP 0.
	\end{align*}
	The condition $T^{-1} \sum_{t=1}^T \mathbb{E} \left[ u_{t,T}(\theta_0)  \phi_{t,T} (\theta_0)^\top \right] \to 0$ follows by exactly the same arguments, but working under expectation.
	Eventually, we get that
	%	We finally show that $\sum_{t=1}^T \mathbb{E} \left[ || T^{-1/2} u_{t,T}(\theta_0)||^{2+\delta} \right] \to 0$. 
	%	For this, we get that
	\begin{align*}
		\sum_{t=1}^T \mathbb{E} \left[ || T^{-1/2} u_{t,T}(\theta_0)||^{2+\delta} \right]
		= T^{-\delta/2} \delta_T^{(14 + 7 \delta)/2}  \; \frac{\omega_\text{Mode}^{2+\delta}}{2^{2+\delta}} \frac{1}{T} \sum_{t=1}^T \mathbb{E} \left[ \left| \left|   \bh \int  u^2 K \left( u \right) f'''_{t} (\zeta \delta_T u) \, \mathrm{d}u \right|\right|^{2+\delta}  \right] \to 0,
	\end{align*}
	as $T^{-\delta/2} \delta_T^{(14 + 7 \delta)/2} \to 0$ and the remaining term is bounded, which concludes the proof.
\end{proof}

\section{\edit{Bandwidth Choice}}
\label{sec:BandwidthChoice}

As discussed at the end of Section \ref{sec:ForecastRatioanlityTestMode}, for the Gaussian kernel we choose the bandwidth as
\begin{align}
		\delta_T &= k_1 \cdot k_2 \cdot T^{-0.143}, 
		\qquad \text{with} \qquad 
		k_1 = 2.4 \, \widehat{\operatorname{Med}} \big( \big| \varepsilon_{1:T} - \widehat{\operatorname{Med}} [ \varepsilon_{1:T} ] \big| \big),
		\qquad \text{and}  \\
		k_2 &= \exp(-3 \left| \hat \gamma \right|),
		\qquad \text{where } \qquad 
		\hat \gamma =  \frac{3 \, \big( \frac{1}{T} \sum_t \varepsilon_t - \widehat{\operatorname{Med}} [\varepsilon_{1:T}] \big)}{\widehat{\sigma}(\varepsilon_{1:T})}, \nonumber
\end{align}
which extends the rule-of-thumb proposed by \cite{Kemp2012} and \cite{Kemp2019} by additionally taking into account the empirical skewness of the data.
While the convergence rate of $T^{-0.143}$, which is almost $T^{-1/7}$, follows from the theoretical considerations in Section \ref{sec:ForecastRatioanlityTestMode}, we assess the performance of the data-driven constants $k_1$ and $k_2$ in the following through simulations.

In particular, we consider the simulation setup of Section \ref{sec:SimModeRat} and use the modified bandwidth choices given by
\begin{align}
	\label{eqn:BandwidthFactor}
	\delta_{T,c} = c \cdot k_1 \cdot k_2 \cdot T^{-0.143},
\end{align}
with the additional scaling factors $c \in \{0.5, 0.75, 1, 1.5, 2\}$.

\begin{table}[tb]
	\centering
	\scriptsize
	\begin{tabular}{l lrrrrr lrrrrr lrrrrr rrrrr}
		\midrule 
		%		\cmidrule(lr){2-6} \cmidrule(lr){7-11} \cmidrule(lr){12-16}
		& & \multicolumn{5}{c}{Skewness $\gamma =  0$}  & & \multicolumn{5}{c}{Skewness $\gamma =  0.25$}   & & \multicolumn{5}{c}{Skewness $\gamma =  0.5$} \\
		\cmidrule(lr){3-7} \cmidrule(lr){9-13} \cmidrule(lr){15-19} 
		%		\cmidrule(lr){2-16}
		Bandwidth factor & & 0.5 & 0.75& 1 & 1.5  & 2  & & 0.5 & 0.75& 1 & 1.5  & 2  & & 0.5 & 0.75& 1 & 1.5  & 2 \\
		\midrule 
		\multicolumn{15}{l}{\textbf{(1) iid DGP}} \\
		\\
		$100$ &   & 4.2 & 4.7 & 4.9 & 6.1 & 6.3 &   &  4.2 &  4.4 &  6.2 &  8.5 & 10.5 &   &  4.9 &  5.7 &  7.4 & 14.6 & 21.2\\
		$500$ &   & 4.8 & 4.9 & 5.7 & 6.2 & 6.2 &   &  5.6 &  6.2 &  8.2 & 12.8 & 20.8 &   &  5.8 &  5.7 &  7.2 & 14.8 & 27.9\\
		$2000$ &   & 5.4 & 5.9 & 6.2 & 6.0 & 6.0 &   &  5.0 &  5.9 &  7.6 & 14.6 & 31.1 &   &  4.4 &  5.5 &  6.3 & 13.7 & 30.2\\
		$5000$ &   & 5.1 & 5.1 & 5.2 & 5.7 & 5.0 &   &  5.7 &  5.7 &  6.8 & 14.8 & 36.0 &   &  5.3 &  6.0 &  6.9 & 12.3 & 33.6 \\
		\midrule 
		\multicolumn{15}{l}{\textbf{(2) AR-GARCH DGP}} \\
		\\
		$100$ &   & 4.5 & 5.2 & 5.9 & 6.5 & 7.1 &   &  4.9 &  5.5 &  6.6 &  8.2 & 10.5 &   &  5.1 &  6.6 &  8.6 & 14.9 & 22.9\\
		$500$ &   & 3.8 & 4.3 & 4.9 & 5.1 & 5.1 &   &  4.7 &  5.7 &  6.8 & 12.1 & 19.3 &   &  5.7 &  5.3 &  7.3 & 14.4 & 27.1\\
		$2000$ &   & 4.4 & 4.7 & 5.1 & 6.3 & 6.2 &   &  4.9 &  6.3 &  7.0 & 13.2 & 27.6 &   &  4.6 &  5.3 &  6.2 & 13.9 & 32.1\\
		$5000$ &   & 5.4 & 4.9 & 4.9 & 5.8 & 5.5 &   &  4.5 &  4.7 &  6.4 & 15.2 & 37.2 &   &  4.5 &  4.7 &  5.7 & 12.7 & 33.8 \\
		\midrule 		
		\multicolumn{15}{l}{\textbf{(3) Cross-Sectional Heteroskedastic DGP}} \\
		\\
		$100$ &   & 4.9 & 5.7 & 5.5 & 6.2 & 6.7 &   &  4.9 &  5.1 &  5.8 &  7.0 &  8.9 &   &  4.8 &  5.2 &  7.3 & 11.5 & 15.8\\
		$500$ &   & 4.6 & 4.6 & 5.1 & 6.7 & 6.3 &   &  5.9 &  6.9 &  7.5 &  9.3 & 14.7 &   &  5.3 &  5.4 &  7.3 & 12.1 & 21.2\\
		$2000$ &   & 5.1 & 5.1 & 5.8 & 5.1 & 5.1 &   &  5.5 &  6.2 &  7.2 & 14.3 & 26.2 &   &  5.1 &  5.5 &  6.6 & 13.8 & 29.1\\
		$5000$ &   & 4.9 & 4.3 & 4.2 & 5.0 & 5.6 &   &  4.7 &  5.5 &  6.9 & 16.7 & 39.7 &   &  5.5 &  5.5 &  6.2 & 12.9 & 33.3 \\
		\midrule 		
		\multicolumn{15}{l}{\textbf{(4) AR DGP}} \\
		\\
		$100$ &   & 3.8 & 4.9 & 5.6 & 7.8 & 7.8 &   &  4.9 &  5.3 &  6.5 &  9.4 & 11.4 &   &  5.5 &  7.3 &  8.9 & 14.5 & 20.4\\
		$500$ &   & 5.4 & 5.9 & 6.2 & 6.4 & 6.1 &   &  5.1 &  6.2 &  7.8 & 12.7 & 19.5 &   &  5.3 &  6.2 &  8.0 & 16.9 & 29.6\\
		$2000$ &   & 3.9 & 4.7 & 5.1 & 5.2 & 5.3 &   &  4.9 &  5.0 &  5.9 & 13.2 & 30.3 &   &  5.1 &  6.2 &  7.2 & 14.3 & 30.4\\
		$5000$ &   & 5.2 & 5.2 & 5.1 & 5.1 & 5.0 &   &  4.4 &  5.1 &  6.5 & 15.3 & 37.6 &   &  4.7 &  5.0 &  5.5 & 11.8 & 34.5 \\
		\bottomrule 
		%		\multicolumn{\textit{Notes:} This table presents the empirical rejection rates (in percent) of the mode rationality test using various sample sizes, various levels of skewness in the residual distribution, different choices of instruments, and the four DGPs described in equation (\ref{eqn:GeneralDGP}). The nominal significance level is $5\%$.}
	\end{tabular}
	\caption{\textbf{Mode rationality test sizes for different bandwidth choices.} 
		This table displays the mode rationality test sizes for five bandwidth correction factors $c \in \{0.5, 0.75, 1, 1.5, 2\}$ as described in \eqref{eqn:BandwidthFactor} for the four DGPs from \eqref{eqn:GeneralDGP} and footnote \ref{fn:DGPs} together with three choices of the skewness parameter $\gamma \in \{0, 0.25, 0.5\}$ and four sample sizes.}
	\label{tab:SizeBandwidth}
\end{table}

Table \ref{tab:SizeBandwidth} reports the size of our mode rationality test for the adjustments with  $c \in \{0.5, 0.75, 1, 1.5, 2\}$, for three skewness parameters of $\gamma \in \{0, 0.25, 0.5\}$ and the four considered sample sizes for the four considered DGPs in \eqref{eqn:GeneralDGP}. Besides the homoskedastic cross-sectional and the AR-GARCH DGPs described below  \eqref{eqn:GeneralDGP}, we also use a heteroskedastic cross-sectional DGP that is as the homoskedastic one, but with $\sigma_{t+1} = 0.5 + 1.5(t + 1)/T$, and a simple AR(1) process with $Z_t = Y_{t}$, $\zeta = 0.5$ and $\sigma_{t+1} = 1$.

For symmetric data with $\gamma = 0$, all bandwidth choices result in a similar behavior under the null with approximately correct rejection rates.
This easily follows from the fact that \emph{any} modal midpoint  (with any value of $\delta > 0$) coincides with the mode for symmetric data such that any bandwidth choice works theoretically.
For the skewed data, however, an increasing bandwidth through a higher value of $c$ results in inflated test sizes, where our choice of $c=1$ with values below $8\%$ in reasonable sample sizes behaves satisfactory. 

\begin{figure}[tb]
	\centering
	\includegraphics[width=\linewidth]{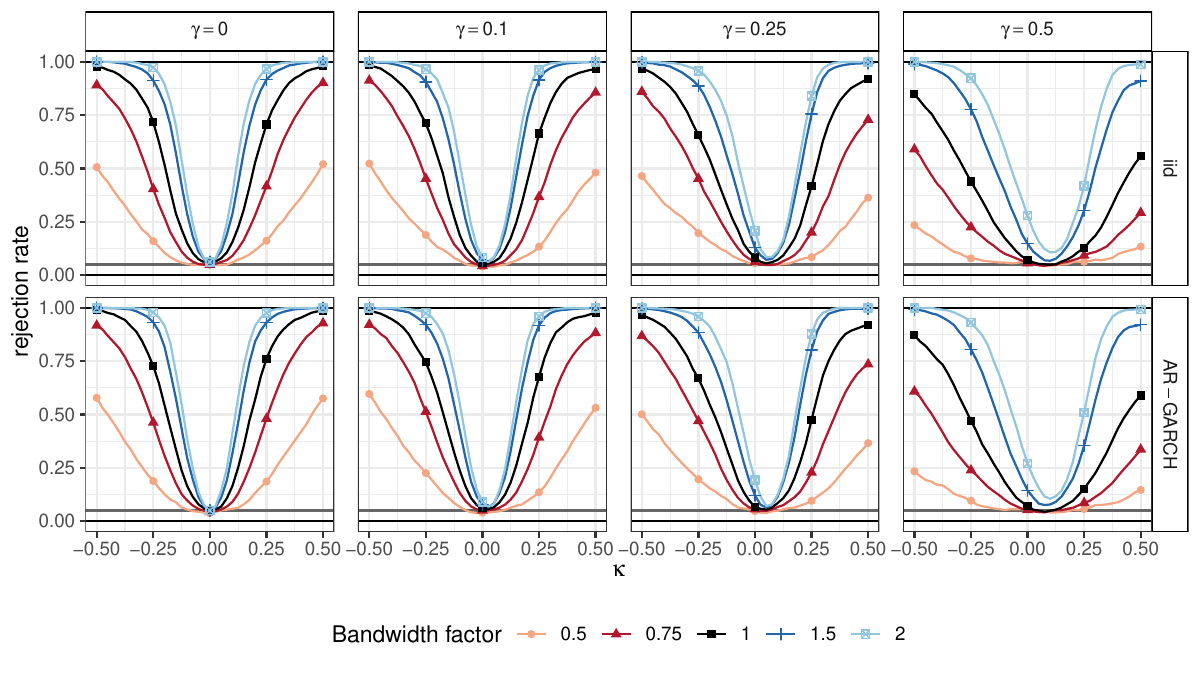}
	\caption{\textbf{Mode rationality test power for different bandwidth choices.} 
	This figure plots the empirical rejection frequencies for different bandwidth factors $c \in \{0.5, 0.75, 1, 1.5, 2\}$ as given in \eqref{eqn:BandwidthFactor} against the degrees of misspecification $\kappa$ for two DGPs in the vertical panels and for four skewness levels in the horizontal panels.
	We fix the sample size to $T=500$, the misspecification follows the \textit{bias} setup and we utilize the instrument vector $\bh= (1,X_t)$ and a nominal significance level of $5\%$.}
	\label{fig:PowerBandwidth}
\end{figure}

Figure \ref{fig:PowerBandwidth} continues to illustrate the test's power by plotting the rejection rates against the degree of misspecification $\kappa$ for the homoskedastic and AR-GARCH DGPs, the four skewness values and a fixed sample size of $T=500$.
Notice that the rejection rates at $\kappa = 0$ display a subset of the values shown in Table \ref{tab:SizeBandwidth}.
We see a uniformly increasing power for increasing values of $c$ that can be explained by the fact that larger bandwidth choices effectively utilize more data.
However, as already seen in Table \ref{tab:SizeBandwidth}, larger values of $c$ inflate the rejection rates under the null hypothesis with $\kappa = 0$, hence leading to a severely oversized test.

Overall, as common in non-parametric statistics, the finite sample bandwidth choice is a delicate balancing of test size and power (or in the classical estimation world, bias and variance), where our choice of $c=1$ achieves a relatively balanced behavior of the test.

\section{Kernel Choice}
\label{sec:KernelChoice}

The asymptotic results presented in Section \ref{sec:ForecastRatioanlityTestMode} rely on the chosen kernel $K$ satisfying Assumption \ref{assu:Kernel}. Besides the normalization $\int K(u) \mathrm{d}u = 1$ and boundedness assumptions, we impose the \textit{first-order kernel} condition $\int u K(u) \mathrm{d}u = 0$ (and $\int u^2 K(u) \mathrm{d}u >0$ follows from the non-negativity of $K$).
As discussed in \cite{LiRacine2006}, higher-order kernels allow one to apply a Taylor expansion of higher order and can thereby obtain a faster rate of convergence, which could in theory be made arbitrarily close to $\sqrt{T}$, at the cost of stronger smoothness assumptions on the underlying density function. However, in our application of kernel functions to the generalized modal midpoint in Definition \ref{def:GeneralizedModalMidpoint}, we need to ensure that the limit of this quantity is well-defined and unique, and that the identification is \textit{strict}. For this, we assume in Theorem \ref{thm:GenModalMidpointProperties} that the kernel function is log-concave %. % and has infinite support. 
which is automatically violated for higher-order kernels. Consequently, we do not consider higher-order kernels in this work. 

It is also well-known in the literature on nonparametric statistics that kernels with bounded support can be more efficient \citep{LiRacine2006}. 
%The efficiency in terms of the MSE in nonparametric (density) estimation is governed by the quantity $\big(\int u^2 K(u) \mathrm{d}u\big)^2 \big/ \int K(u)^2 \mathrm{d}u$.
However, note that the kernel choice enters the asymptotic variance of nonparametric density estimation through the quantity $\int K(u)^2 \mathrm{d}u$, while the covariance $\Omega_{\mathrm{Mode}}$ in Theorem \ref{thm:ModeRationality}  depends upon $\int K'(u)^2 \mathrm{d}u$, revealing that the efficiency of our mode rationality tests depends upon a different quantity.

\begin{figure}[tb]
	\includegraphics[width=\linewidth]{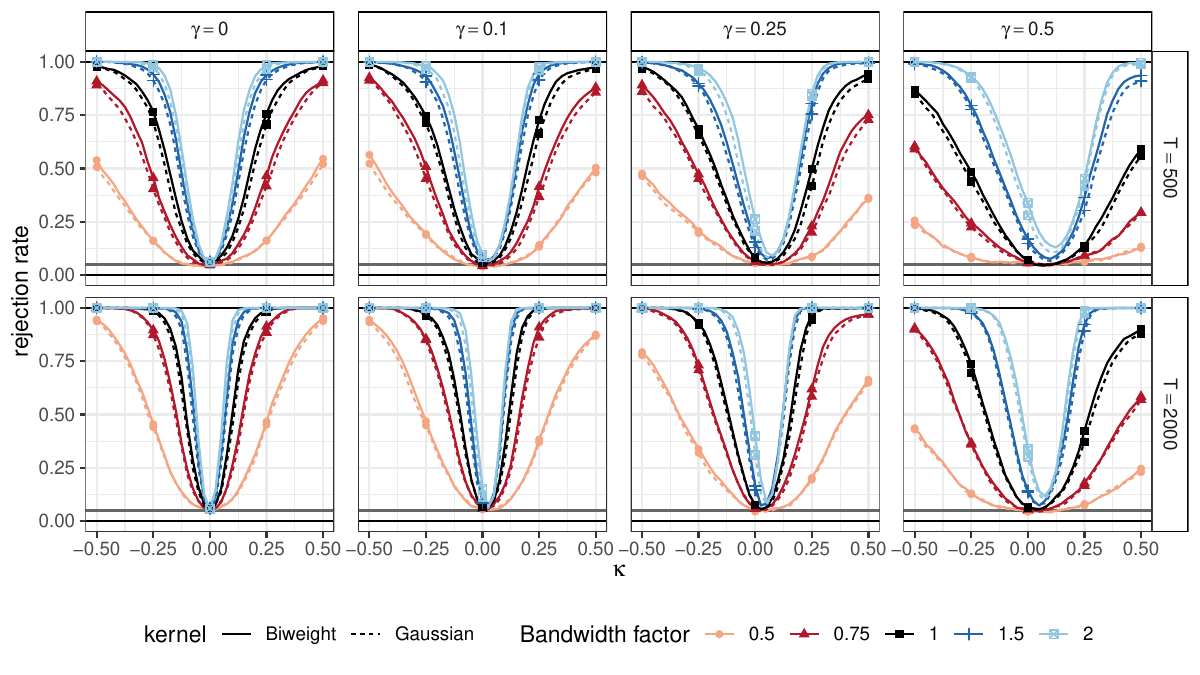}
	\caption{\textbf{Test power for different kernel functions.}
		This figure plots the empirical rejection frequencies for the Gaussian and the biweight kernels against the degrees of misspecification $\kappa$ for different sample sizes in the vertical panels and for four skewness levels in the horizontal panels.
		We simulate data from the homoskedastic process, the misspecification follows the \textit{bias} setup and we utilize the instrument vector $\bh= (1,X_t)$ and a nominal significance level of $5\%$.
		The bandwidth factor refers to the bandwidth choice explained in \eqref{eqn:BandwidthFactor}.}
	\label{fig:PowerBiasKernels}
\end{figure}

\edit{Figure \ref{fig:PowerBiasKernels} compares the Gaussian kernel (that we use in our main analyses) with a biweight kernel within the simulation setup of Section \ref{sec:BandwidthChoice}.
For the biweight kernel, we adjust the bandwidth choice given in Section~\ref{sec:BandwidthChoice} by multiplying with a factor of 2.5 in order to account for the different kernel shape.}
Figure \ref{fig:PowerBiasKernels} illustrates that the test power does not increase by employing a biweight kernel, which has bounded support, and is usually found to be relatively efficient in nonparametric estimation. Strict identifiability of the generalized modal midpoint---and hence, asymptotic identifiability of the mode---only holds for kernel functions with unbounded support, which is satisfied by the Gaussian kernel, but not so for the biweight kernel.

\section{Convex combination of functional values}
\label{sec:convexfunctionals}

Here, we illustrate that a convex combination of functionals is generally neither elicitable nor identifiable.
This result shows that---as stated in Remark \ref{rem:ConvexFunctionalCombinations}---testing forecast rationality directly for convex functional combinations is impossible.
For this, we adapt the simplified notation of Section \ref{sec:ModeFunctional}.

\begin{proposition}
	\label{prop:convexvalues}
	Let $\mathcal{P}$ be a convex class of distributions and let $\Gamma_\beta(P)= \beta \Gamma^l(P) + (1-\beta) \Gamma^n(P)$, $\beta \in [0,1]$ be the convex combination of a \emph{linear} functional $\Gamma^l: \mathcal{P} \to \mathbb{R}$ and a non-linear functional $\Gamma^n: \mathcal{P} \to \mathbb{R}$, which are both continuous (in the distribution $P$) and translation equivariant, i.e., if $P\in \mathcal{P}$, then for $c \in \mathbb{R}$ the shifted $P+c \in \mathcal{P}$ and $\Gamma^l(P+c)=\Gamma^l(P)+c$ and $\Gamma^n(P+c)=\Gamma^n(P)+c$.
	Then, the functional  $\Gamma_\beta$ is neither elicitable nor identifiable.	
\end{proposition}

\begin{proof}[Proof of Proposition \ref{prop:convexvalues}]
	 Theorem 6 of \cite{Gneiting2011} and Proposition 3.11 of \cite{FisslerHoga2021} show that convex level sets, i.e., 
	\begin{align}
		\text{for } P_1,P_2 \in \mathcal{P} \text{ with } \Gamma_\beta(P_1)=\Gamma_\beta(P_2) \quad \Longrightarrow \quad \Gamma_\beta(\alpha P_1 + (1-\alpha) P_2) = \Gamma_\beta(P_1),
	\end{align}
	for $\alpha \in (0,1)$ are a necessary condition for elicitability and identifiability of a functional.
	We show that the functional $\Gamma_\beta$ does not have convex level sets.

	As $\Gamma^n$ is not linear, there exists $\alpha \in (0,1)$, and $P_1$, $P_2 \in \mathcal{P}$ such that 
	$$\Gamma^n(\alpha P_1 + (1-\alpha) P_2) \neq \alpha \Gamma^n(P_1) + (1-\alpha) \Gamma^n(P_2).$$
	Define $P_2' = P_2 - \Gamma_\beta(P_2) + \Gamma_\beta(P_1)$ such that $\Gamma_\beta(P_2')=\Gamma_\beta(P_1)$ and
	$$\Gamma^n(\alpha P_1 + (1-\alpha) P_2') \neq \alpha \Gamma^n(P_1) + (1-\alpha) \Gamma^n(P_2')$$
	as $\Gamma^n(P+c)=\Gamma^n(P)+c$.
	It follows that
	\begin{align*}
		\Gamma_\beta(\alpha P_1 + (1-\alpha) P_2') &= 
		\beta \Gamma^l(\alpha P_1 + (1-\alpha) P_2') + (1-\beta) \Gamma^n(\alpha P_1 + (1-\alpha) P_2') \\ &\neq
		\beta \alpha \Gamma^l(P_1) + \beta (1-\alpha) \Gamma^l(P_2') +
		(1-\beta) \alpha \Gamma^n(P_1) + (1-\beta) (1-\alpha) \Gamma^n(P_2')\\ &=
		\alpha \Gamma_\beta(P_1) + (1-\alpha) \Gamma_\beta(P_2')\\&=
		\Gamma_\beta(P_1)
	\end{align*}
	and hence $\Gamma_\beta$ does not have convex level sets.
\end{proof}

As functionals are linear if and only if they are expectations \citep{abernethy2012characterization}, the mean is linear and the median is non-linear for classes $\mathcal{P}$ sufficiently rich enough such that it contains distributions for which the median does not equal the mean (i.e., asymmetric distributions). Hence, Proposition \ref{prop:convexvalues} shows that a convex combination of the mean and median is generally neither elicitable nor identifiable. Eliciting a convex combination that may further include the mode (which is itself only asymptotically elicitable) is only possible in the unusual case where a convex (sub-)combination of the nonlinear median and mode functionals becomes linear, an outcome that does not generally hold.

Fortunately, testing rationality for functionals elicited through a convex combination of loss functions is feasible and its interpretability is supported by the following result that convexity of the combination weights is preserved when moving from a functional elicited by a convex combination of loss (identification) functions to a convex combination of functional values. 

%For this, let
%\begin{align}
%	V_\textrm{Mean}(x,Y) &= x-Y \\
%	V_\textrm{Med}(x,Y) &= \mathds{1}(Y < x) - \mathds{1}(Y > x) \\
%	V_{\textrm{MMP}, \delta} (x,Y) &= - \frac{1}{\delta^2} K'\left( \frac{x-Y}{\delta} \right),
%\end{align}
%with expectations
%\begin{align}
%	\overline{V}_\mu(x,P) &= x- \mu\\
%	\overline{V}_m(x,P) &= 2 F(x) - 1 \\
%	\overline{V}_{\text{MMP}_\delta} (x,P) &= (\tilde K'_\delta \ast f)(x),
%\end{align}
%See the proof of Theorem \ref{thm:GenModalMidpointProperties} for the last equality. 

The loss functions  $L_\mathrm{Mean}$,  $L_\mathrm{Med}$ and $L_\mathrm{Mode, \delta} = \delta^{3/2} L_\delta^K$ are defined just before equation \eqref{eqn:CentralityAssumptionLoss} and the identification functions $V_\mathrm{Mean}$,  $V_\mathrm{Med}$ and $V_\mathrm{Mode, \delta} = \delta^{3/2}  V_\delta^K$ at equations \eqref{eq:MeanIF}, \eqref{eq:MedIF} and \eqref{eq:ModeIF}.
We use the scaling by $\delta^{3/2}$ for the asymptotic mode loss and identification functions to be consistent with Section \ref{sec:FunctionalCentrality}.

\begin{proposition}
	 \label{prop:BetaThetaMapping}
	 Let $\mathcal{P}$ be some class of distributions such that the mean, $\mu$, the median, $m$, and the generalized modal midpoint with parameter $\delta$, $\Gamma^K_\delta$, exist and are elicited by their loss functions $L_\mathrm{Mean}$,  $L_\mathrm{Med}$, and $L_\mathrm{Mode, \delta}$.
	 Let $x$ be the functional defined by
	 \begin{align}
	 	x(P) = \underset{\tilde x \in \mathbb{R}}{\argmin} \; \mathbb{E}_{Y \sim P} \left[ \theta_0^\top \Big( L_\mathrm{Mean}(\tilde x,Y),  L_\mathrm{Med}(\tilde x,Y),   L_\mathrm{Mode, \delta}(\tilde x,Y) \Big)^\top \right] 
	 \end{align}
 	 for some $\theta_0 \in \Theta$.
	 Then, for every $P \in \mathcal{P}$ there exists some $\beta_0 \in \Theta$, such that $x(P) = \beta_0^\top \big(\mu(P), m(P) , \Gamma^K_\delta(P) \big)^\top$.
	 % where $\mu(P)$ denotes the mean, $m(P)$ the median, and  $\Gamma^K_\delta(P)$ the generalized modal midpoint with parameter $\delta$ of the distribution $P$.
\end{proposition}

\begin{proof}[Proof of Proposition \ref{prop:BetaThetaMapping}]
	Let $P \in \mathcal{P}$, where we assume without loss of generality that the three functionals are not all equal. For notational convenience we drop $P$ when denoting functional values, e.g., we write $\mu$ instead of $\mu(P)$.
	For the elicited forecast $x$, it holds that
	\begin{align}
		\overline{V}(x) := \theta_0^\top \Big( \overline{V}_\mathrm{Mean}(x,P),  \overline{V}_\mathrm{Med}(x,P),  \overline{V}_\mathrm{Mode, \delta}(x,P) \Big)^\top = 0.
	\end{align}
	We define $L := \min(\mu, m, \Gamma^K_\delta )$ and $U := \max(\mu, m, \Gamma^K_\delta )$ as the lower and upper functional values where it holds that $L<U$. 
	Further let $\bar V_L( x,P)$ and $\bar V_U( x,P)$ denote the corresponding expected identification functions for the distribution $P$.
	Suppose that $x < L$. Then, it must hold that $\overline{V}(x) > 0$ as all three expected identification functions have the same sign as they are oriented in the sense of \cite{steinwart}.
	Similarly, if $x > U$, it must hold that $\overline{V}(x) < 0$.
	Hence, we can conclude that $x \in [L,U]$, which implies that there exists $\zeta \in [0,1]$ such that $ x = \zeta L + (1-\zeta) U$. Thus, $ x$ can be constructed as a convex combination of the functional values, i.e., there exists a $\beta_0 \in \Theta$ such that $x = \beta_0^\top \big(\mu, m ,  \Gamma^K_\delta \big)^\top$.
\end{proof}

%Things to discuss:
%\begin{itemize}
%	\item This shows that convexity is preserved when moving from a loss (identification) function combination to a functional combination. However, the $\beta_0$ and $\theta_0$ in the proof of Proposition \ref{prop:BetaThetaMapping} are in general not equal.
%	\item Also note that there exists more than one combination for any fixed distribution $P$ if the weight of the third functional is not set to zero. This illustrates the partial identification again.
%	\item Note that the existence holds only for a fixed distribution $P$. (Shall we really?)
%	\item Discuss what happens to the mode as MMP is not mode... say it holds asymptotically?
%\end{itemize}

\section{\edit{Heterogeneity of Forecasters}}
\label{sec:HeterogeneityForecasters}

In this section, we discuss heterogeneity in the measure of centrality adopted by survey respondents. For this discussion we adopt the simplified decision-theoretic notation of Section \ref{sec:ModeFunctional}.
In particular, consider a forecaster who issues a (centrality) forecast $x$ for her outcome variable $Y \sim P$, whose distribution $P \in \mathcal{P}$ is in some suitable class of distributions $\mathcal{P}$.

Formally, we represent the mixture of forecaster types through a probabilistic mixture of mean, median, and the generalized modal midpoint (with bandwidth $\delta$) forecasts with combination weights (probabilities) $\gamma = (\gamma_\mathrm{Mean},\gamma_\mathrm{Med},\gamma_\mathrm{Mode, \delta}) \in \Theta$.
Given a random variable $Z \sim C_\gamma$ that is categorically distributed with values $\{(1,0,0), \, (0,1,0), \, (0,0,1)\}$ and corresponding probabilities $\gamma$, we define the probabilistic mixture (with the symbol $ \otimes_\gamma$) as
\begin{equation}
	\label{eq:gammafc}
	x_\gamma^\dagger
	= {x_\gamma^\dagger(P)}  
	= \otimes_\gamma \big( \mu(P),m(P),\Gamma^K_\delta(P) \big) 
	=  \big( \mu(P),m(P),\Gamma^K_\delta(P) \big) Z^\top.
\end{equation}
Notice that opposed to the \emph{deterministic} forecasts $x$ considered in Section \ref{sec:ModeFunctional}, the probabilistic mixture $x_\gamma^\dagger$ is \emph{random} through its definition relying on the random variable $Z$.

%In Section \ref{sec:FunctionalCentrality}, we consider forecasts with $P=Y_{t+1} | \mathcal{F}_t$. Unlike the optimal forecast $X_t^*(\theta)$ from Equation (\ref{eqn:CentralityAssumptionLoss}), the mixture $X_\gamma^\dagger(Y_{t+1}|\mathcal{F}_t)$ is not $\mathcal{F}_t$-measurable. 
%However, the key moment condition from Section \ref{sec:FunctionalCentrality} holds unconditionally as the following Proposition shows.

\begin{proposition}
	\label{prop:GammaThetaMapping}
	Let $\mathcal{P}$ be some class of distributions such that the mean, $\mu = \mu(P)$, the median, $m = m(P)$, and the generalized modal midpoint with parameter $\delta$, $\Gamma^K_\delta = \Gamma^K_\delta(P)$, exist and are unique for all $P \in \mathcal{P}$.
	Let $x_\gamma^\dagger$ be the functional defined by Equation (\ref{eq:gammafc}) for some $\gamma \in \Theta$.
	Then, for every $P \in \mathcal{P}$ there exists some $\theta_0 \in \Theta$, such that
	$$
	{\mathbb{E}_{Y\sim P, Z \sim C_\gamma}} \left[ \theta_0^\top \Big({V}_\mathrm{Mean}(x_\gamma^\dagger,Y),{V}_\mathrm{Med}(x_\gamma^\dagger,Y),{V}_{\mathrm{Mode, \delta}}(x_\gamma^\dagger,Y)\Big)^\top \right] = 0
	$$
\end{proposition}

\begin{proof}
	Using the notation $V_\theta(x,Y)=\theta^\top \Big({V}_\mathrm{Mean}(x,Y),{V}_\mathrm{Med}(x,Y),{V}_{\mathrm{Mode, \delta}}(x,Y)\Big)^\top$ and \\  $\overline{V}(x,P)= \mathbb{E}_{Y\sim P, Z \sim C_\gamma}[V(x,Y)]$, the linearity of the expectation operator implies
	\begin{align}
		\label{eq:ExpectedGammaMixture}
		\overline{V}_\theta(x,P) = \theta^\top \Big( \overline{V}_\mathrm{Mean}(x,P),  \overline{V}_\mathrm{Med}(x,P),  \overline{V}_{\mathrm{Mode, \delta}}(x,P) \Big)^\top.
	\end{align}
	Using  $x= \otimes_\gamma (\mu(P),m(P),\Gamma^K_\delta(P))$ and the law of total expectations, we get for the first component that
	\begin{align}
		\label{eq:GammaVMean}
		\overline{V}_\mathrm{Mean}(x,P)=\gamma^\top\Big( \overline{V}_{\mathrm{Mean}}(\mu,P), \overline{V}_{\mathrm{Mean}}(m,P), \overline{V}_{\mathrm{Mean}}(\Gamma^K_\delta,P)\Big)^\top
	\end{align}
	where it naturally holds that $\overline{V}_{\mathrm{Mean}}(\mu,P)=0$.
	
	Analogous formulas to \eqref{eq:GammaVMean} hold for $\overline{V}_\mathrm{Med}(x,P)$ and $\overline{V}_{\mathrm{Mode, \delta}}(x,P)$.
	Plugging these into \eqref{eq:ExpectedGammaMixture} gives
	\begin{align}
		\label{eqn:V_in_het}
		\overline{V}_\theta(x,P)=\theta_\mathrm{Mean} \Big(\gamma_\mathrm{Med} \overline{V}_{\mathrm{Mean}}(m,P)+ \gamma_{\mathrm{Mode, \delta}} \overline{V}_{\mathrm{Mean}}(\Gamma^K_\delta,P)\Big)+ \nonumber \\
		\theta_\mathrm{Med} \Big(\gamma_\mathrm{Mean} \overline{V}_{\mathrm{Med}}(\mu,P)+ \gamma_{\mathrm{Mode, \delta}} \overline{V}_{\mathrm{Med}}(\Gamma^K_\delta,P)\Big)+\\
		\theta_{\mathrm{Mode, \delta}} \Big(\gamma_\mathrm{Mean} \overline{V}_{\mathrm{Mode, \delta}}(\mu,P)+ \gamma_\mathrm{Med} \overline{V}_{\mathrm{Mode, \delta}}(m,P)\Big). \nonumber 
	\end{align}
	Let $L = L(P) = \min(\mu(P),m(P),\Gamma^K_\delta(P))$ denote the ``lowest'' functional, $U$ the ``highest'' functional, and $M$ the third functional in the ``middle''. 
	As the three identification functions are oriented in the sense of \cite{steinwart}, it holds that $\overline{V}_{L}(U,P)\ge 0 $ and $\overline{V}_{L}(M,P) \ge 0$ such that
	\[
	\gamma_U \overline{V}_{L}(U,P)+\gamma_M \overline{V}_{L}(M,P) \ge 0.
	\]
	Analogously, it follows that
	\[
	\gamma_L \overline{V}_{U}(L,P)+\gamma_M \overline{V}_{U}(M,P) \le 0.
	\]
	Thus, there is at least one non-negative and one non-positive term in Equation (\ref{eqn:V_in_het}) and for any $P$ and any $\gamma$ there exits a $\theta$ such that $\overline{V}_\theta(x,P)=0$. (If two functionals are the same, a similar argument applies.)
\end{proof}

Proposition \ref{prop:GammaThetaMapping} illustrates that in a setting of multiple forecasters, which have identical forecast distributions $P$, but are heterogeneous in either forecasting the mean, median or mode (modal midpoint), the asymptotic moment condition that there exists a $\theta_0 \in \Theta$ such that $\frac{1}{T} \sum_{t=1}^T \mathbb{E} \left[ \phi_{t,T} (\theta_0)  \right] \to 0$, which is closely related to Assumption \ref{assu:GMMWeakIDRegCond} (D), holds for the choice $\bh = 1$. 
Moreover, the value of $\theta_0 \in \Theta$ does not change for any location shift $P+c$ of the distribution $P$, as all three identification functions fulfill the \emph{translation invariance} property $V(x+c,y+c)=V(x,y)$. Thus, Proposition \ref{prop:GammaThetaMapping} implies the existence of a unique $\theta_0 \in \Theta$ for which the moment condition is satisfied for $\bh = 1$, even if different forecasters (that report heterogeneous functionals) have shifted distributions. 
For $\bh = 1$, Proposition \ref{prop:GammaThetaMapping} can also be extended to the more general assumption that the heterogeneous forecast distributions only have \emph{identical order} of mean, median, and modal midpoint.\footnote{
In that case, we define $\mathcal{P}^O \subseteq \mathcal{P}$ as the subclass of distributions that have identical order of mean, median, and modal midpoint and consider a probability distribution $Q$ on $\mathcal{P}^O$ (that informally represents the different distributions of the different forecasters).
The proof of  Proposition \ref{prop:GammaThetaMapping} then works equivalently by considering the expected identification functions $\overline{V}(x,Q)$ to also capture the expectation over $Q$, which represents (the distribution over) the different subjective distributions.}
However, multiple instruments (as e.g., our baseline case of $\bh = (1,X_t)$) generally imply different values of $\theta_0$ for each component of the instrument vector, such that there exists no $\theta_0$ that ensures the (vector-valued) validity of Assumption \ref{assu:GMMWeakIDRegCond} (D).

Also note that in the case of a probabilistic mixture, the parameter $\theta_0$ contains information on the ratio of forecasters $\gamma$, but their values are not necessarily equal. 
Using Equation (\ref{eqn:V_in_het}) and (for simplicity) considering the situation without mode (modal midpoint) forecasters, we observe that
$$\frac{\theta_{\mathrm{Mean}}}{\theta_\mathrm{Med}}=\frac{\gamma_\mathrm{Mean}}{\gamma_\mathrm{Med}} \frac{\overline{V}_{\mathrm{Med}}(\mu,P)}{\overline{V}_{\mathrm{Mean}}(m,P)},$$
such that $\theta_0$ reflects the proportion of forecasters if the identification functions are suitably standardized.

% Old:
%Proposition \ref{prop:GammaThetaMapping} illustrates that in a setting of multiple forecasters, which have identical forecast distributions $P$ but are heterogeneous in either forecasting the mean, median or mode (modal midpoint), the moment condition in (\ref{eqn:WeakIDExpectationAssumption}) for the choice $\bh = 1$ holds for some $P$-dependent $\theta_0(P) \in \Theta$.
%However, both, heterogeneous (conditional) forecast distributions among forecasters or the use of more informative instruments (as e.g., our baseline case of $\bh = (1,X_t)$) imply heterogeneous distributions $P$.
%Then, each line (or subgroup) of (\ref{eqn:WeakIDExpectationAssumption}) is valid for a different $\theta_0(P)$ such that the vector-valued statement in (\ref{eqn:WeakIDExpectationAssumption}) does not necessarily hold.

\begin{figure}[ptb]
	\centering
	\includegraphics[width=\linewidth]{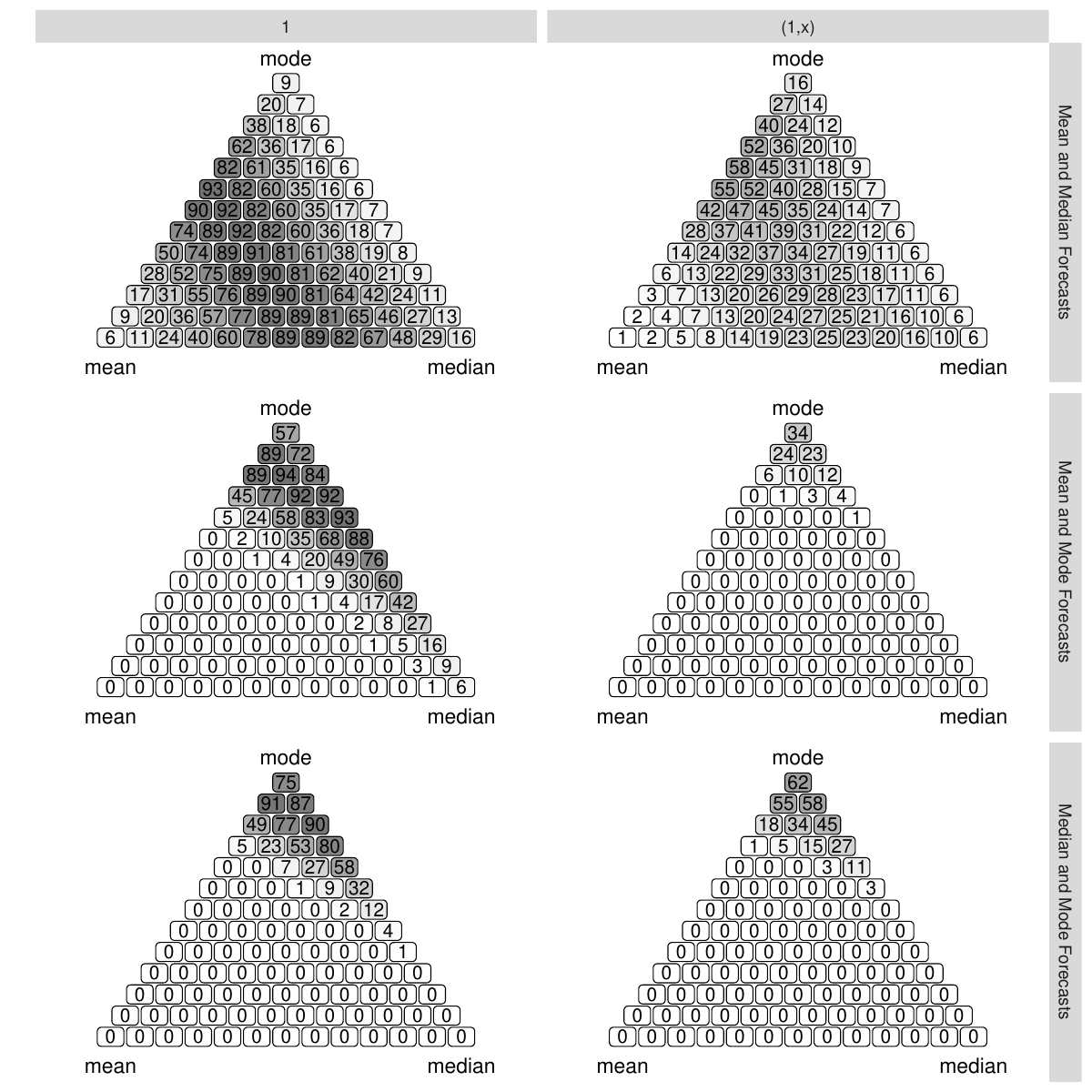}
	\caption{\textbf{Coverage rates in simulations with heterogeneous forecasters.} 
		This figure shows the percentage values in how many of the simulation runs the respective confidence set includes a given point in the triangle in the iid DGP from \eqref{eqn:GeneralDGP}.
		The left column of plots uses the instruments  $\bh = 1$ and the right column $\bh= (1,X_t)$.
		The three plot lines use equally weighted probabilistic mixtures (see \eqref{eq:gammafc}) of the respective functional values.}
	\label{fig:sim_convex_gammacomb}
\end{figure}

We next illustrate via simulations that our confidence sets from \eqref{eqn:GMMWeakIdConfidenceRegion} develop power against probabilistic combinations of functional values as in \eqref{eq:gammafc} for sufficiently informative instruments.
For this, we utilize the iid DGP from \eqref{eqn:GeneralDGP} and simulate according to probabilistic combinations  in \eqref{eq:gammafc} using mean-median combinations with $\gamma = (0.5, 0.5,0)$, mean-mode combinations with $\gamma = (0.5, 0, 0.5)$, and median-mode combinations with $\gamma = (0, 0.5, 0.5)$.
We use the two instrument choices $\bh = 1$ and $\bh= (1,X_t)$.

Figure \ref{fig:sim_convex_gammacomb} depicts the confidence set coverage rates for these six settings.
As illustrated by Proposition \ref{prop:GammaThetaMapping}, when using a constant instrument $\bh = 1$, we can always find a $\theta_0 \in \Theta$ (even partially identified lines in $\Theta$) with approximately nominal coverage rates of $90\%$.
In contrast, when using the more informative instruments $\bh= (1,X_t)$, we see rejection rates clearly below the nominal level of $90\%$ showing that our confidence sets develop power against the entire unit simplex $\Theta$.

In summary, this section illustrates that our confidence sets for convex combinations of centrality measures contain probabilistic mixtures of forecaster types only for unconditional tests using the constant as the only instrument.
In the baseline case of our empirical application, $\bh= (1,X_t)$, our procedure develops power against probabilistic mixtures.

\pagebreak
\section{\edit{Additional Simulation Results}}
\label{sec:AdditionalPlotsTables}

\begin{table}[h!]
	\centering
	\scriptsize
	\begin{tabular}{l lrrrr lrrrr lrrrr lrrrr lrrrr}
		\toprule 
		& & \multicolumn{9}{c}{(1) iid DGP}  & & \multicolumn{9}{c}{(2) AR-GARCH DGP}  \\
		\cmidrule(lr){3-11} \cmidrule(lr){13-21} 
		&&  \multicolumn{4}{c}{Instrument $\bh =  1$} & & \multicolumn{4}{c}{Instrument $\bh =  (1,X_t)$}   & & \multicolumn{4}{c}{Instrument $\bh =  1$} & & \multicolumn{4}{c}{Instrument $\bh =  (1,X_t)$} \\
		\cmidrule(lr){3-6} \cmidrule(lr){8-11} \cmidrule(lr){13-16} \cmidrule(lr){18-21}
		Skewness & & 0 & 0.1 & 0.25 & 0.5 & & 0 & 0.1 & 0.25 & 0.5  && 0 & 0.1 & 0.25 & 0.5 && 0 & 0.1 & 0.25 & 0.5\\ 
		\midrule 
		$100$ &   & 0.7 & 1.1 & 1.2 & 2.5 &   & 1.1 & 1.3 & 1.2 & 1.8 &   & 0.9 & 1.2 & 1.7 & 2.9 &   & 0.8 & 1.2 & 1.6 & 2.3\\
		$500$ &   & 1.1 & 1.4 & 2.4 & 2.3 &   & 1.3 & 1.0 & 2.2 & 1.8 &   & 0.8 & 1.5 & 2.2 & 1.8 &   & 0.8 & 1.0 & 1.5 & 1.7\\
		$2000$ &   & 1.4 & 1.2 & 2.2 & 1.4 &   & 1.2 & 1.5 & 1.7 & 1.3 &   & 1.7 & 0.8 & 1.8 & 2.1 &   & 0.9 & 1.1 & 1.8 & 1.7\\
		$5000$ &   & 0.9 & 1.1 & 1.7 & 2.1 &   & 0.9 & 1.1 & 1.4 & 1.7 &   & 0.9 & 1.1 & 2.3 & 1.2 &   & 0.9 & 1.1 & 1.6 & 1.4 \\
		\midrule
		\\ 
		\midrule
		& & \multicolumn{9}{c}{(3) Cross-Sectional Heteroskedastic DGP}  & & \multicolumn{9}{c}{(4) AR DGP}  \\
		\cmidrule(lr){3-11} \cmidrule(lr){13-21} 
		&&  \multicolumn{4}{c}{Instrument $\bh =  1$} & & \multicolumn{4}{c}{Instrument $\bh =  (1,X_t)$}   & & \multicolumn{4}{c}{Instrument $\bh =  1$} & & \multicolumn{4}{c}{Instrument $\bh =  (1,X_t)$} \\
		\cmidrule(lr){3-6} \cmidrule(lr){8-11} \cmidrule(lr){13-16} \cmidrule(lr){18-21}
		Skewness & & 0 & 0.1 & 0.25 & 0.5 & & 0 & 0.1 & 0.25 & 0.5  && 0 & 0.1 & 0.25 & 0.5 && 0 & 0.1 & 0.25 & 0.5\\ 
		\midrule
		$100$ &   & 0.7 & 1.3 & 1.1 & 2.6 &   & 0.9 & 0.9 & 1.0 & 2.0 &   & 1.1 & 1.4 & 1.4 & 2.5 &   & 1.1 & 0.9 & 1.1 & 2.4\\
		$500$ &   & 1.0 & 1.8 & 2.5 & 2.3 &   & 1.5 & 2.1 & 1.8 & 1.8 &   & 1.2 & 1.0 & 2.5 & 2.6 &   & 1.4 & 1.2 & 1.8 & 2.0\\
		$2000$ &   & 1.2 & 1.7 & 2.6 & 1.8 &   & 0.9 & 1.7 & 1.8 & 1.4 &   & 1.3 & 1.1 & 2.1 & 1.7 &   & 0.9 & 1.1 & 1.5 & 2.2\\
		$5000$ &   & 0.8 & 1.1 & 2.2 & 1.8 &   & 1.0 & 1.2 & 1.5 & 1.6 &   & 0.8 & 1.8 & 2.0 & 1.7 &   & 1.0 & 1.4 & 1.8 & 1.3\\
		\bottomrule
	\end{tabular}
	\caption{This table presents the empirical size of the mode rationality test for a Gaussian kernel, varying sample sizes, varying levels of skewness in the residual distribution and different instrument choices for a nominal significance level of $1\%$. The DGPs are described in \eqref{eqn:GeneralDGP} and in footnote \ref{fn:DGPs}.}	
	\label{tab:SizeGaussian1}
\end{table}

\begin{table}[h!]
	\centering
	\scriptsize
	\resizebox{\columnwidth}{!}{
	\begin{tabular}{l lrrrr lrrrr lrrrr lrrrr lrrrr}
		\toprule 
		& & \multicolumn{9}{c}{(1) iid DGP}  & & \multicolumn{9}{c}{(2) AR-GARCH DGP}  \\
		\cmidrule(lr){3-11} \cmidrule(lr){13-21} 
		&&  \multicolumn{4}{c}{Instrument $\bh =  1$} & & \multicolumn{4}{c}{Instrument $\bh =  (1,X_t)$}   & & \multicolumn{4}{c}{Instrument $\bh =  1$} & & \multicolumn{4}{c}{Instrument $\bh =  (1,X_t)$} \\
		\cmidrule(lr){3-6} \cmidrule(lr){8-11} \cmidrule(lr){13-16} \cmidrule(lr){18-21}
		Skewness & & 0 & 0.1 & 0.25 & 0.5 & & 0 & 0.1 & 0.25 & 0.5  && 0 & 0.1 & 0.25 & 0.5 && 0 & 0.1 & 0.25 & 0.5\\ 
		\midrule 
		$100$ &   &  8.6 & 10.3 & 11.0 & 15.4 &   &  9.8 &  9.9 & 11.3 & 13.3 &   &  9.6 & 11.5 & 11.8 & 17.2 &   & 10.3 & 11.1 & 11.7 & 15.7\\
		$500$ &   & 11.2 & 12.1 & 15.2 & 14.3 &   & 10.8 & 10.9 & 13.9 & 13.0 &   &  9.2 & 11.5 & 14.0 & 14.0 &   & 10.7 & 11.7 & 13.2 & 12.7\\
		$2000$ &   & 11.0 & 12.2 & 15.3 & 12.8 &   & 11.6 & 12.6 & 13.8 & 12.2 &   & 11.5 & 10.2 & 13.1 & 12.7 &   & 11.0 & 10.2 & 13.1 & 12.2\\
		$5000$ &   & 10.7 & 11.2 & 14.2 & 13.1 &   & 10.4 & 11.0 & 11.9 & 12.5 &   & 10.7 & 11.2 & 14.4 & 11.6 &   & 10.9 &  9.7 & 12.0 & 11.9 \\
		\midrule
		\\ 
		\midrule
		& & \multicolumn{9}{c}{(3) Cross-Sectional Heteroskedastic DGP}  & & \multicolumn{9}{c}{(4) AR DGP}  \\
		\cmidrule(lr){3-11} \cmidrule(lr){13-21} 
		&&  \multicolumn{4}{c}{Instrument $\bh =  1$} & & \multicolumn{4}{c}{Instrument $\bh =  (1,X_t)$}   & & \multicolumn{4}{c}{Instrument $\bh =  1$} & & \multicolumn{4}{c}{Instrument $\bh =  (1,X_t)$} \\
		\cmidrule(lr){3-6} \cmidrule(lr){8-11} \cmidrule(lr){13-16} \cmidrule(lr){18-21}
		Skewness & & 0 & 0.1 & 0.25 & 0.5 & & 0 & 0.1 & 0.25 & 0.5  && 0 & 0.1 & 0.25 & 0.5 && 0 & 0.1 & 0.25 & 0.5\\ 
		\midrule
		$100$ &   &  9.5 & 10.5 & 11.2 & 15.2 &   & 10.8 & 11.2 & 11.2 & 13.7 &   &  9.8 & 11.0 & 12.7 & 16.1 &   & 11.7 & 11.8 & 12.2 & 16.3\\
		$500$ &   & 10.7 & 13.2 & 15.5 & 14.6 &   & 10.2 & 11.8 & 14.1 & 13.7 &   & 11.2 & 12.1 & 14.8 & 15.8 &   & 10.5 & 10.0 & 14.6 & 14.3\\
		$2000$ &   & 10.1 & 12.3 & 14.7 & 13.0 &   & 10.4 & 11.8 & 13.5 & 12.3 &   &  9.2 & 10.9 & 13.0 & 13.4 &   &  9.8 & 11.3 & 11.2 & 12.7\\
		$5000$ &   & 10.3 & 12.6 & 13.6 & 12.6 &   &  9.3 & 11.6 & 13.6 & 11.5 &   & 10.5 & 11.8 & 13.8 & 11.2 &   & 10.0 & 11.7 & 12.7 & 10.5 \\
		\bottomrule
	\end{tabular}
	}
	\caption{This table presents the empirical size of the mode rationality test for a Gaussian kernel, varying sample sizes, varying levels of skewness in the residual distribution and different instrument choices for a nominal significance level of $10\%$. The DGPs are described in \eqref{eqn:GeneralDGP} and in footnote \ref{fn:DGPs}.}		
	\label{tab:SizeGaussian10}
\end{table}

\begin{figure}[tbp]
	\centering
	\includegraphics[width=\linewidth]{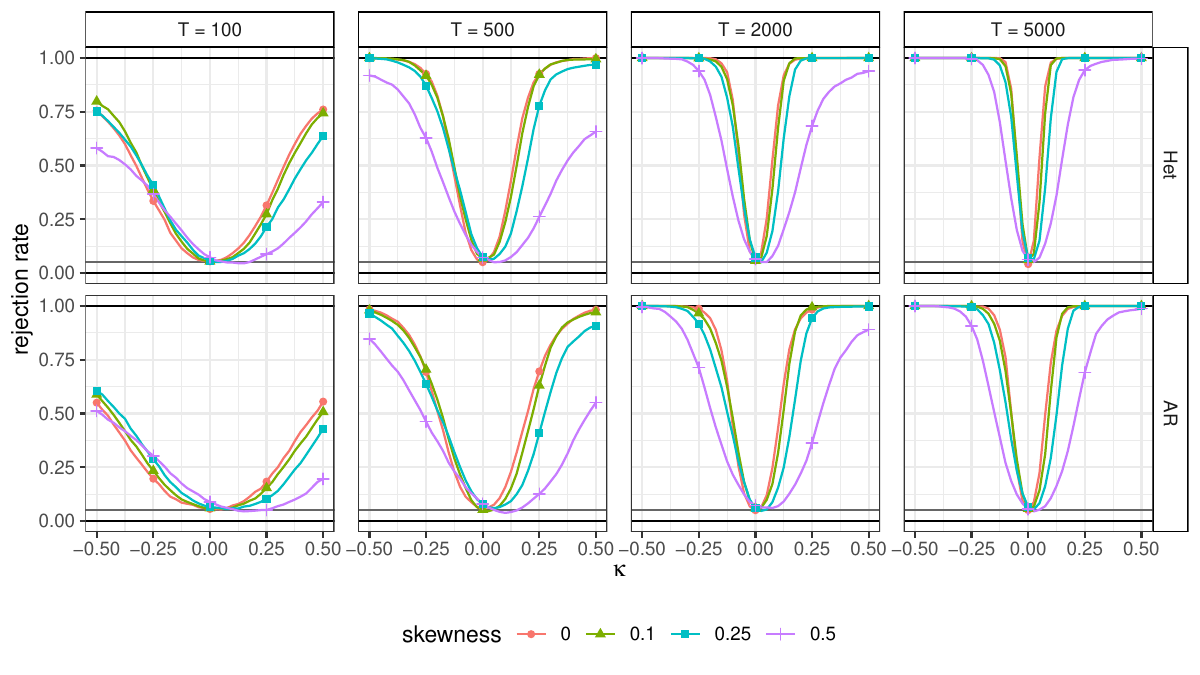}
	\caption{\textbf{Power study for the ``bias'' simulation.} This figure plots the empirical rejection frequencies against the degrees of misspecification $\kappa$ for different sample sizes in the vertical panels and for the Heteroskedastic and the AR DGPs (described in footnote \ref{fn:DGPs}) in the horizontal panels. The misspecification follows the bias design described in the main text and we use the instrument vector $(1,X_t)$ and a nominal significance level of $5\%$.}
	\label{fig:PowerBiasHetAR}
\end{figure}

\begin{figure}[tbp]
	\centering
	\includegraphics[width=\linewidth]{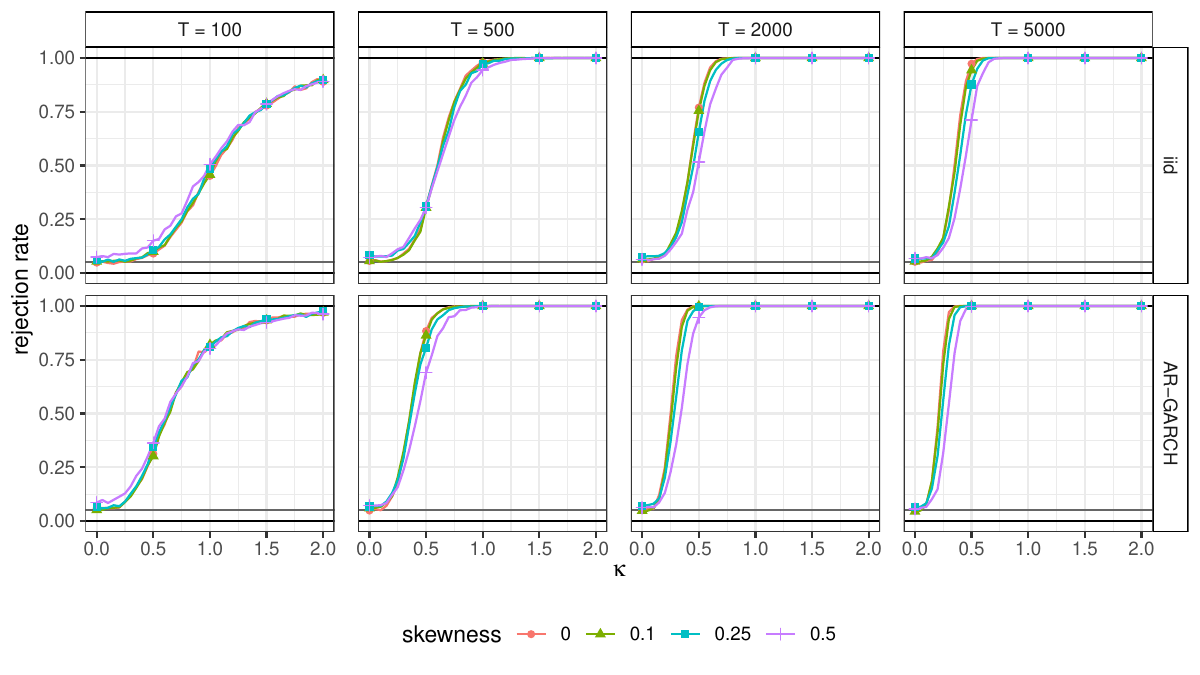}
	\caption{\textbf{Power study for the ``noise'' simulation.} This figure plots the empirical rejection frequencies against the degrees of misspecification $\kappa$ for different sample sizes in the vertical panels and for the iid and the AR-GARCH DGPs in the horizontal panels. The misspecification follows the noise design described in footnote \ref{fn:NoiseMisspec} and we use the instrument vector $(1,X_t)$ and a nominal significance level of $5\%$.}
	\label{fig:PowerNoiseHomGARCH}
\end{figure}

\begin{figure}[tbp]
	\centering
	\includegraphics[width=\linewidth]{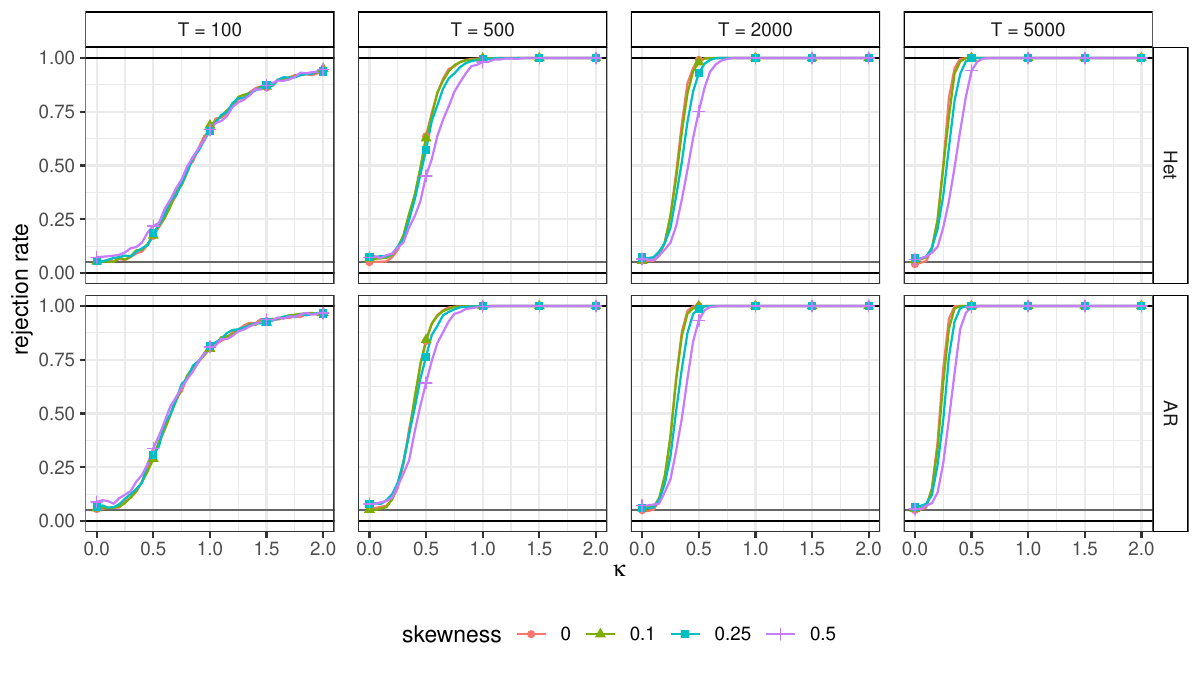}
	\caption{\textbf{Power study for the ``noise'' simulation.} This figure plots the empirical rejection frequencies against the degrees of misspecification $\kappa$ for different sample sizes in the vertical panels and for the Heteroskedastic and the AR DGPs (described in footnote \ref{fn:DGPs}) in the horizontal panels. The misspecification follows the noise design described in footnote \ref{fn:NoiseMisspec}  and we use the instrument vector $(1,X_t)$ and a nominal significance level of $5\%$.}
	\label{fig:PowerNoiseHetAR}
\end{figure}

\begin{figure}[p]
	\centering
	\includegraphics[width=\linewidth]{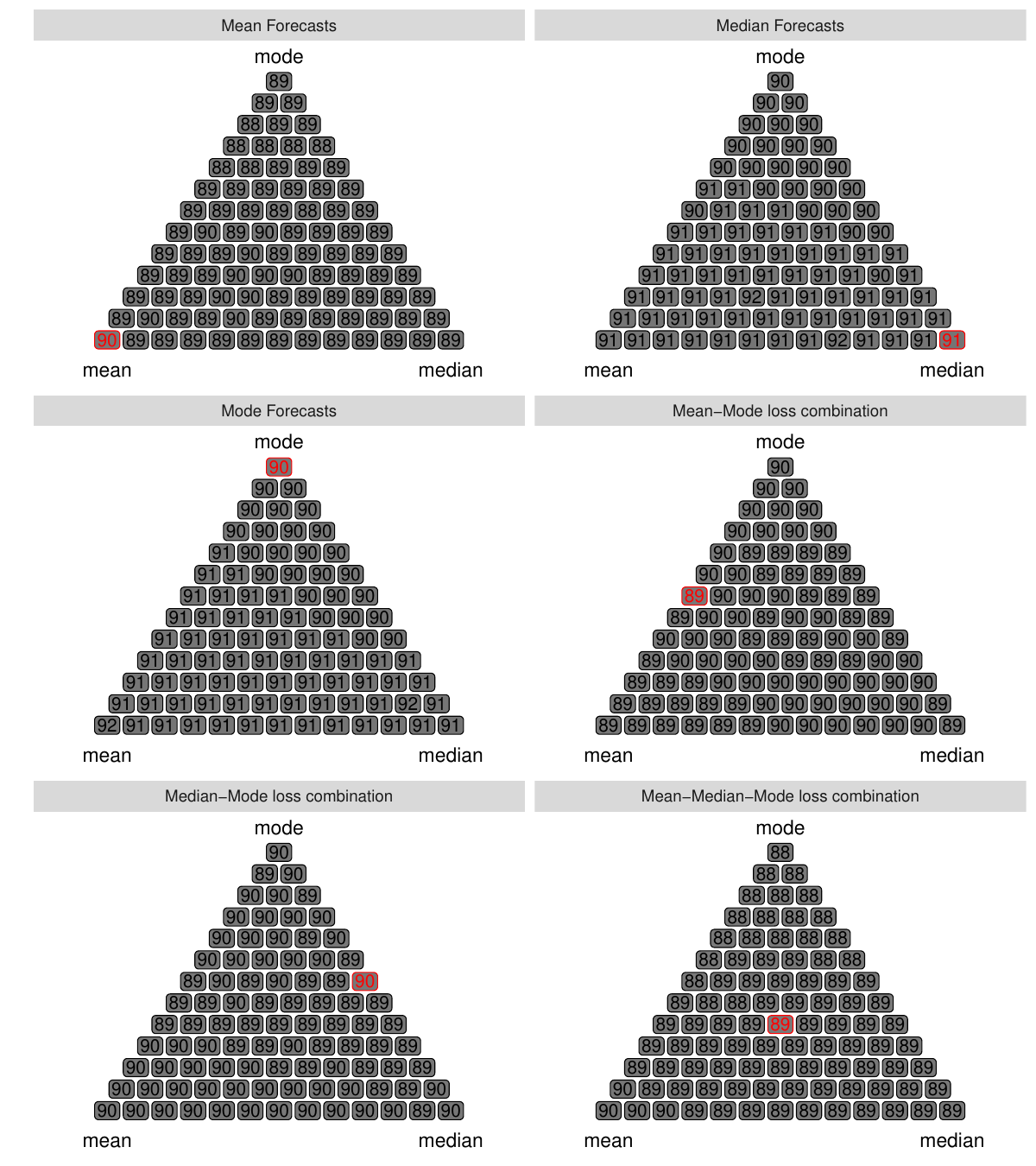}
	\caption{Confidence set coverage rates as in Figure \ref{fig:sim_convex_thetacomb}, but with symmetric data, $\gamma=0$.}
\end{figure}

\begin{figure}[p]
	\centering
	\includegraphics[width=\linewidth]{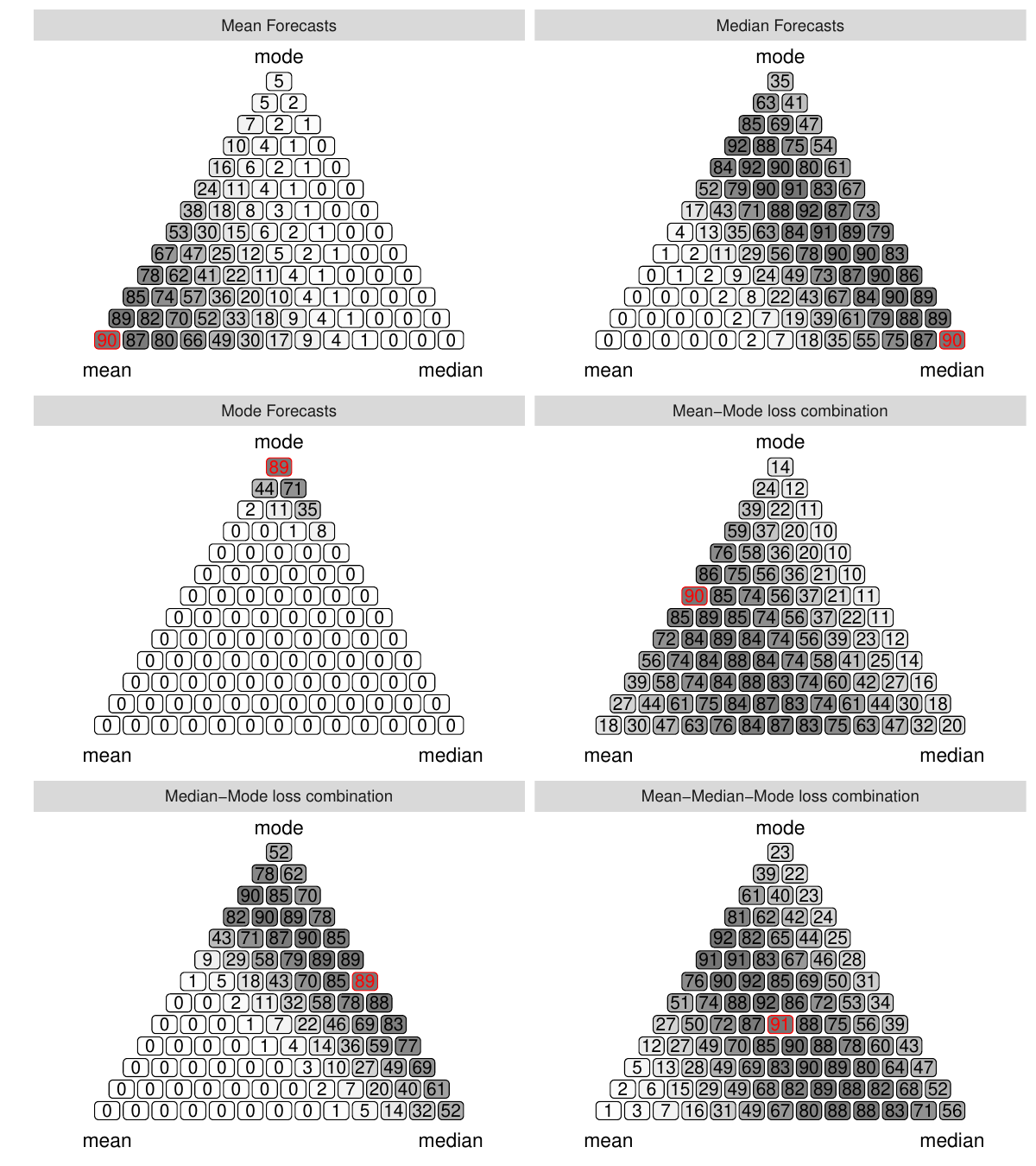}
	\caption{Confidence set coverage rates as in Figure \ref{fig:sim_convex_thetacomb}, but with the AR-GARCH DGP.}
\end{figure}

\begin{figure}[p]
	\centering
	\includegraphics[width=\linewidth]{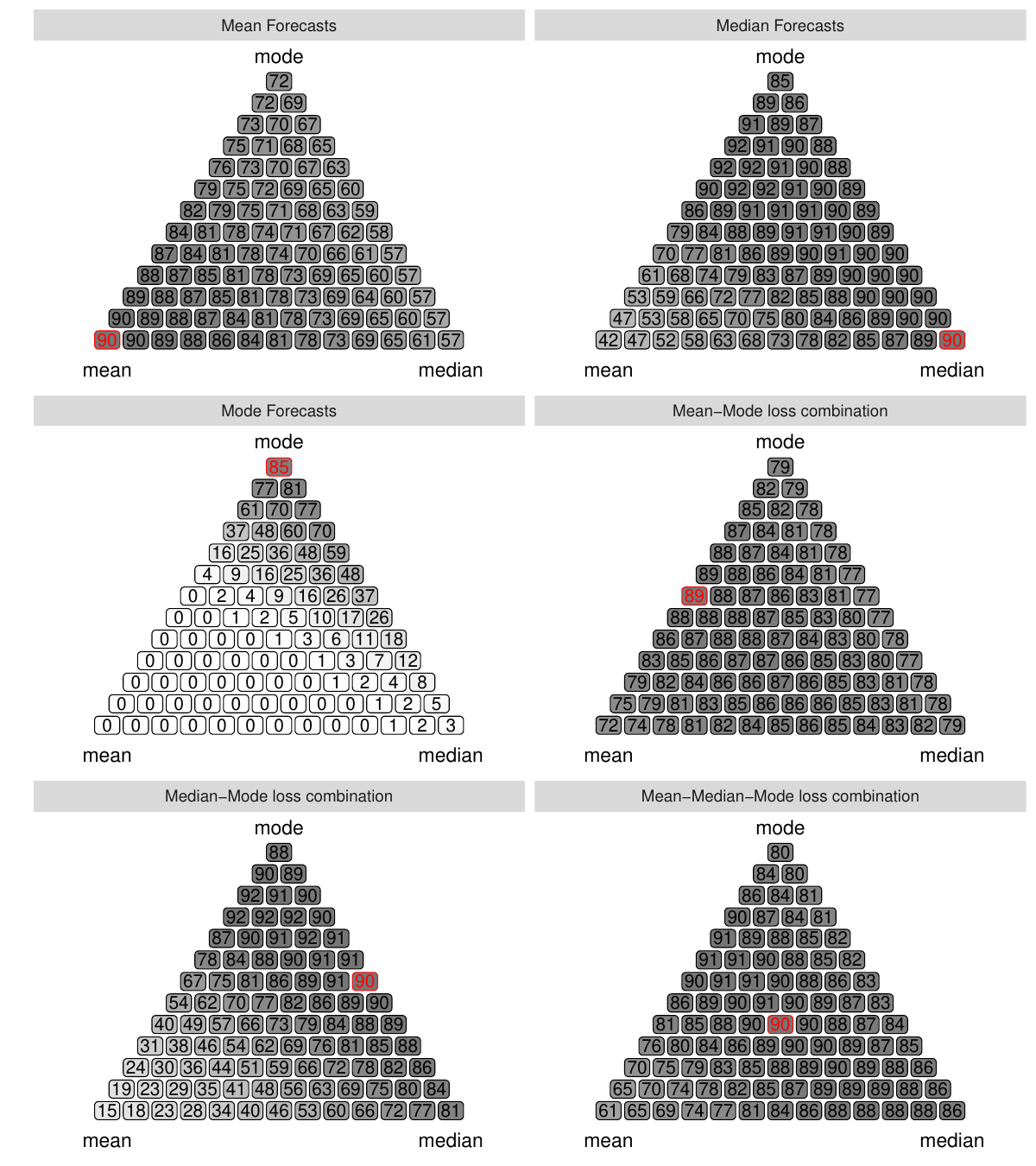}
	\caption{Confidence set coverage rates as in Figure \ref{fig:sim_convex_thetacomb}, but with $T=500$.}
\end{figure}

\section{Additional results for the Survey of Consumer Expectations}
\label{sec:SuppSCErest}

\subsection{Clustered Covariance Estimator}
\label{sec:clusterSCE}

Figures \ref{fig:ResultsSCEIncomeC} to \ref{fig:ResultsSCEIncomeRoundC} below are equivalent to Figures \ref{fig:ResultsSCEIncome} to \ref{fig:ResultsSCEIncomeRound} after substituting the covariance estimator of Theorem \ref{thm:GMMWeakIDCovarianceConsistency} with a clustered covariance estimator $\widehat{\Sigma}^{CL}_T$. Let $\phi_{i,t,T}(\theta)$ denote the moment function of individual $i$ at time $t$. $\mathcal{T}$ denotes the number of waves and $n_t$ the number of observations within wave such that $T=\sum_{t=1}^\mathcal{T} n_t$.
\begin{align}
	\widehat{\Sigma}^{CL}_T(\theta) = \frac{1}{T} \sum_{t=1}^\mathcal{T} 
	\left( \sum_{i=1}^{n_t} \widehat \phi_{i,t,T}(\theta) \right) 
	\left( \sum_{i=1}^{n_t} \widehat \phi_{i,t,T}(\theta) \right)^\top
\end{align}

Overall, the results are robust to clustering at the time level. While the mean rejection is less pronounced in Figure \ref{fig:ResultsSCEIncomeC}, the confidence sets are sharper for the subpopulations. In Figure \ref{fig:ResultsSCEIncomeTercileC} to \ref{fig:ResultsSCEIncomeRoundC} mean rationality is consistently rejected for lower income individuals at the 5\% level.

\begin{figure}[tbh]
	\centering
	\includegraphics[width=.6\linewidth]{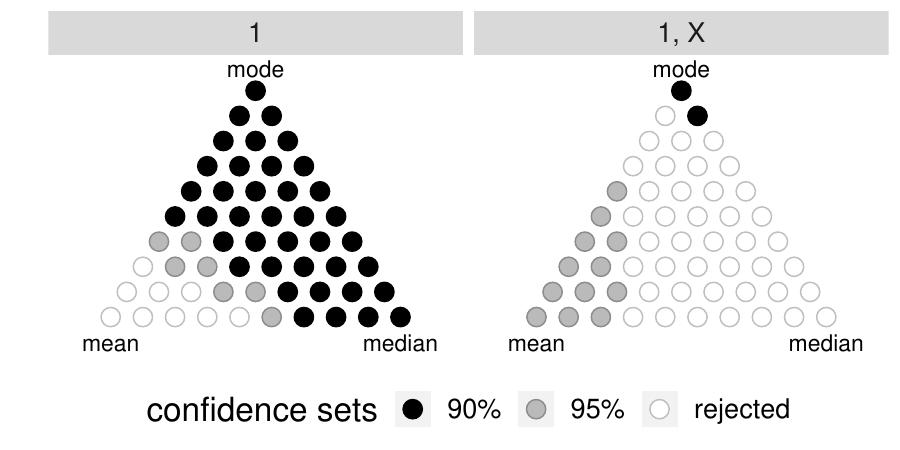}
	\caption{\textbf{Confidence sets for income survey forecasts.} This figure shows the measures of centrality that rationalize the New York Federal Reserve income survey forecasts with clustered covariance estimator.% Black dots indicate that the measure is inside the Stock-Wright 90\% confidence set, dark grey dots indicate that the measure is inside the 95\% confidence set, light grey dots indicate that rationality for that measure can be rejected at the 5\% level. The left panel uses just a constant as the instrument; the right panel uses a constant and the forecast.
	}
	\label{fig:ResultsSCEIncomeC}
\end{figure}

\begin{figure}[tbh]
	\centering
	\includegraphics[width=.9\linewidth]{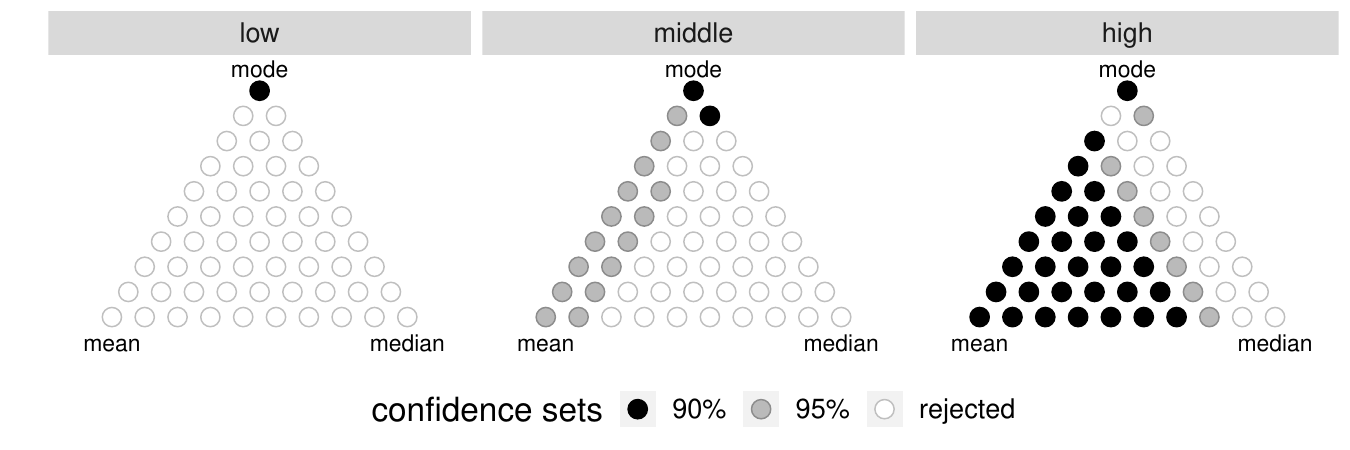}		\caption{\textbf{Confidence sets for income survey forecasts, stratified by income.} This figure shows the measures of centrality %that rationalize the New York Federal Reserve income survey forecasts, 
		for low-, middle- and high-income respondents with clustered covariance estimator. %Groups are formed using terciles of lagged reported income. Black dots indicate that the measure is inside the Stock-Wright 90\% confidence set, dark grey dots indicate that the measure is inside the 95\% confidence set, light grey dots indicate that rationality for that measure can be rejected at the 5\% level. All panels use a constant and the forecast as test instruments.
	}
	\label{fig:ResultsSCEIncomeTercileC}
\end{figure}

\begin{figure}[tbh]
	\centering
	\includegraphics[width=.7\linewidth]{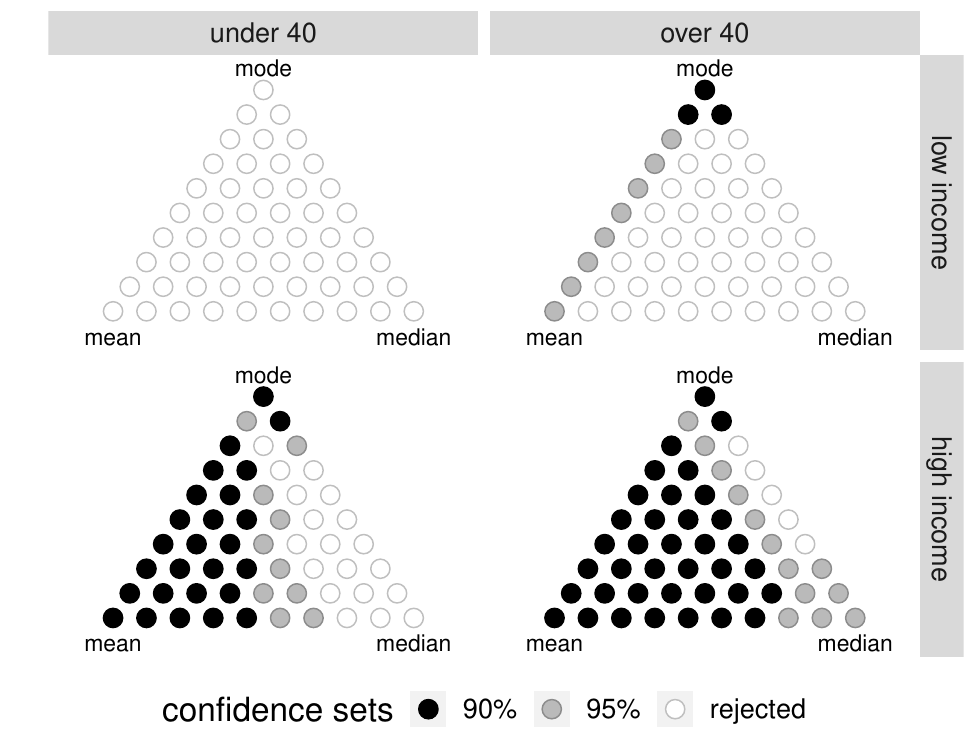}
	\caption{\textbf{Confidence sets for income survey forecasts, stratified by income and age.} This figure shows the measures of centrality that rationalize the New York Federal Reserve income survey forecasts, for low- and high-income respondents who are below or above the age of 40 with clustered covariance estimator.% Income groups are formed using the median lagged reported income. Black dots indicate that the measure is inside the Stock-Wright 90\% confidence set, dark grey dots indicate that the measure is inside the 95\% confidence set, light grey dots indicate that rationality for that measure can be rejected at the 5\% level. All panels use a constant and the forecast as test instruments.
	}
	\label{fig:ResultsSCEIncomeAgeC}
\end{figure}

\begin{figure}[tbh]
	\centering
	\includegraphics[width=0.7\linewidth]{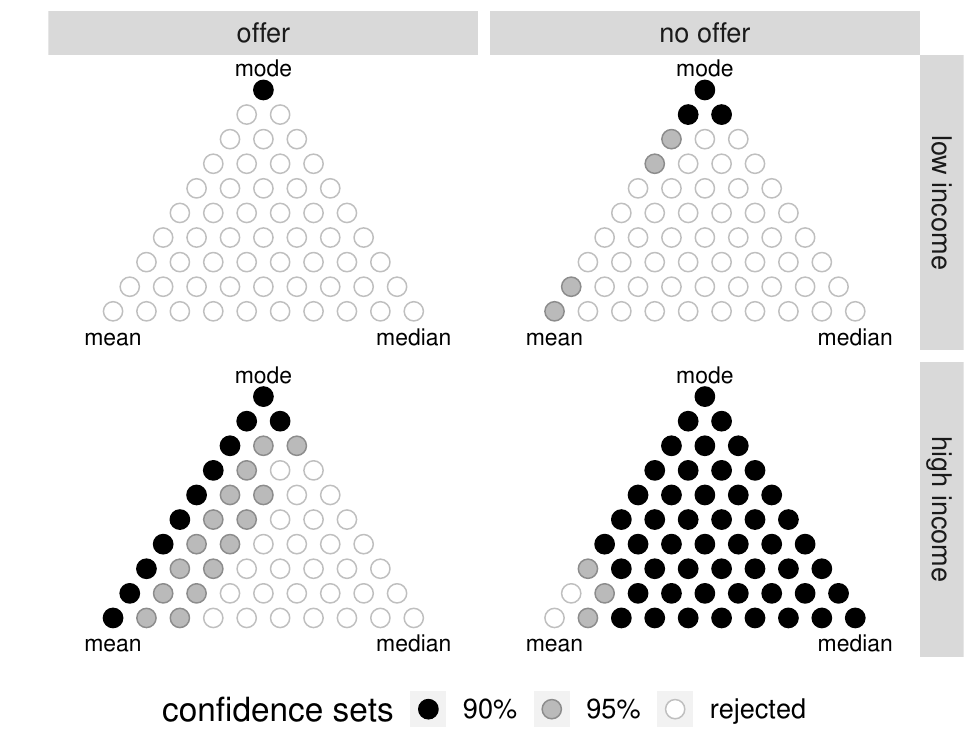}
	\caption{\textbf{Confidence sets for income survey forecasts, stratified by income and job offer.} This figure shows the measures of centrality that rationalize the New York Federal Reserve income survey forecasts, for low- and high-income respondents in the private sector or not with clustered covariance estimator.% Groups are formed using the median lagged reported income, and whether or not the respondent reported receiving at least one job offer in the past four months. Black dots indicate that the measure is inside the Stock-Wright 90\% confidence set, dark grey dots indicate that the measure is inside the 95\% confidence set, light grey dots indicate that rationality for that measure can be rejected at the 5\% level. All panels use a constant and the forecast as test instruments.
	}
	\label{fig:ResultsSCEIncomeOfferC}
\end{figure}

%\begin{figure}[t]
%	\centering
%	\includegraphics[width=0.7\linewidth]{plots/revision/Rsce_by_roundC}
%	\caption{\textbf{Confidence sets for income survey forecasts, stratified by first and second elicitation round.} This figure shows the measures of centrality that rationalize the New York Federal Reserve income survey forecasts, for respondents in their first and second panel round with clustered covariance estimator. 
	%		Black dots indicate that the measure is inside the Stock-Wright 90\% confidence set, dark grey dots indicate that the measure is inside the 95\% confidence set, light grey dots indicate that rationality for that measure can be rejected at the 5\% level. All panels use a constant and the forecast as test instruments.
	%	}
%	\label{fig:ResultsSCERoundC}
%\end{figure}

\begin{figure}[tbh]
	\centering
	\includegraphics[width=0.7\linewidth]{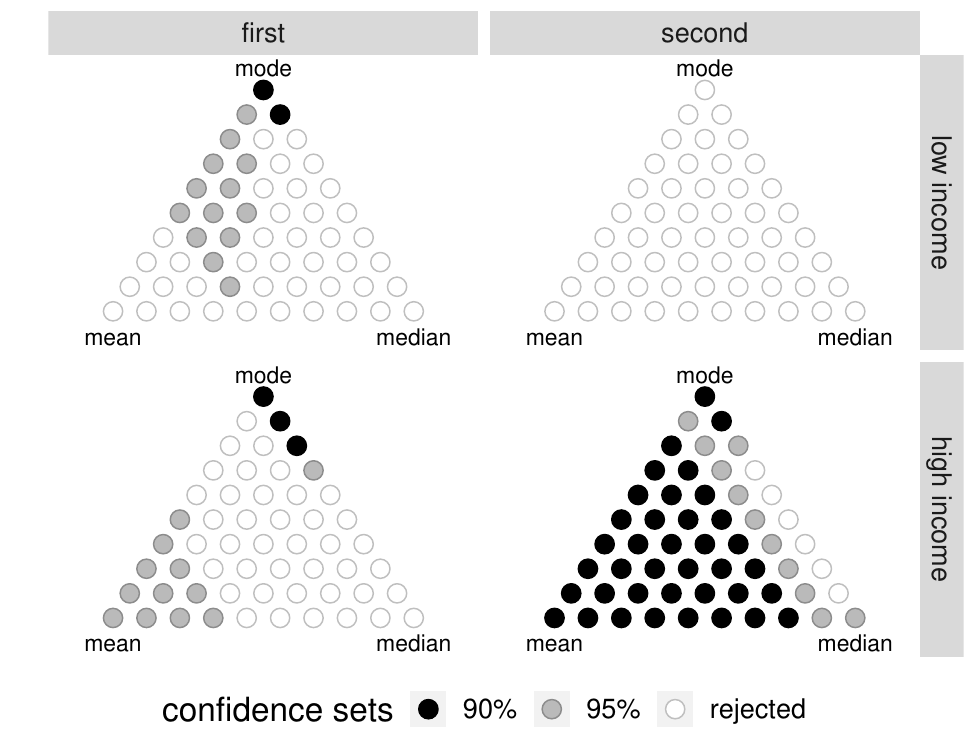}
	\caption{\textbf{Confidence sets for income survey forecasts, stratified by income and elicitation round.} This figure shows the measures of centrality that rationalize the New York Federal Reserve income survey forecasts, for low- and high-income respondents in their first and second panel round with clustered covariance estimator. 
		%		Black dots indicate that the measure is inside the Stock-Wright 90\% confidence set, dark grey dots indicate that the measure is inside the 95\% confidence set, light grey dots indicate that rationality for that measure can be rejected at the 5\% level. All panels use a constant and the forecast as test instruments.
	}
	\label{fig:ResultsSCEIncomeRoundC}
\end{figure}

\FloatBarrier

\subsection{Forecast rationality and past forecast accuracy} 
\label{sec:SupplPastFCAccuracy}

In this section we investigate the relationship between past forecast accuracy and the rationality of subsequent forecasts. \cite{dacunto2019cognitive} find absolute forecast errors to be negatively correlated with cognitive ability measures for inflation expectations, which may suggest that respondents with larger past forecast errors are more likely to issue irrational forecasts. In contrast, \citet{vanNVeldkamp2006learning} propose that economic agents learn more about the economy when more signals are available, and a large forecast error may signal to the respondent that they need to pay more attention to their forecast. We address this question using the set of 1,288 respondents to the FRBNY survey for which we have two matched pairs of predictions and realizations. We classify respondents as having ``small'' or ``large'' past forecast errors using the cross-respondent median of absolute percentage forecast errors in the first round as a cutoff. In addition to the (relative) size of their past forecast errors, we continue to stratify by income. 

\begin{figure}[t]
	\centering
	\includegraphics[width=0.7\linewidth]{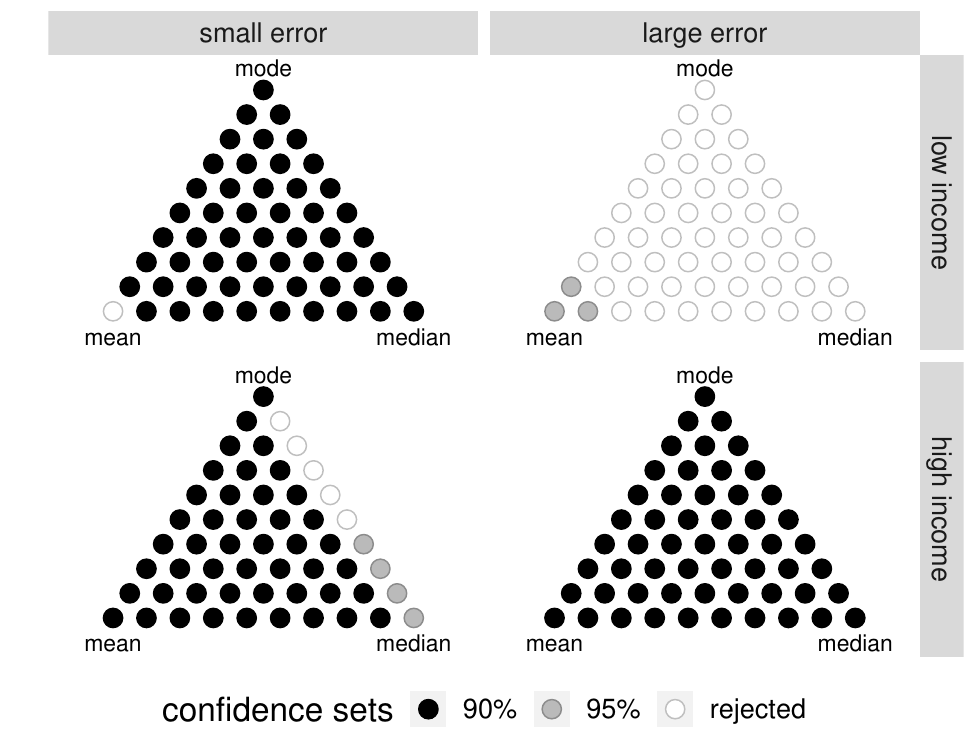}
	\caption{\textbf{Confidence sets for income survey forecasts, stratified by forecast error size and income.} See the notes to Figure \ref{fig:ResultsSCEIncomeTercile} for additional details.}
	\label{fig:ResultsSCEIncomeError}
\end{figure}

Figure \ref{fig:ResultsSCEIncomeError} shows that respondents who had a small forecast error in the first round provide second-round forecasts that can be rationalized using many different measures of centrality, and there are few differences across low- and high-income respondents: for low-income respondents, only the mean leads to a rejection, while for high-income respondents the median and measures between the median and the mode are rejected. For respondents with a large forecast error in the first round, the rationality test results are starkly different for low and high income respondents: high-income respondents can be rationalized using any measure of centrality, while low-income respondents can be rationalized only as mean, or near-mean, forecasters, but only at the 95\% confidence level; at the 90\% confidence level rationality is rejected for all centrality measures. An alternative interpretation of Figure \ref{fig:ResultsSCEIncomeError} is that high-income respondents produce forecasts that are rationalizable for almost all measures of centrality, regardless of past forecast accuracy, while the rationality of low-income respondents' forecasts varies greatly with past forecast errors: those with large past forecast errors are not rationalizable as any measure of centrality at the 90\% confidence level, while those with small forecast errors are rationalizable as almost any measure of centrality. 

Figure~\ref{fig:ResultsSCEIncomeErrorC} reports results corresponding to Figure~\ref{fig:ResultsSCEIncomeError}, but based on the clustered standard errors described in Section~\ref{sec:clusterSCE}.

\begin{figure}[tbh]
	\centering
	\includegraphics[width=0.7\linewidth]{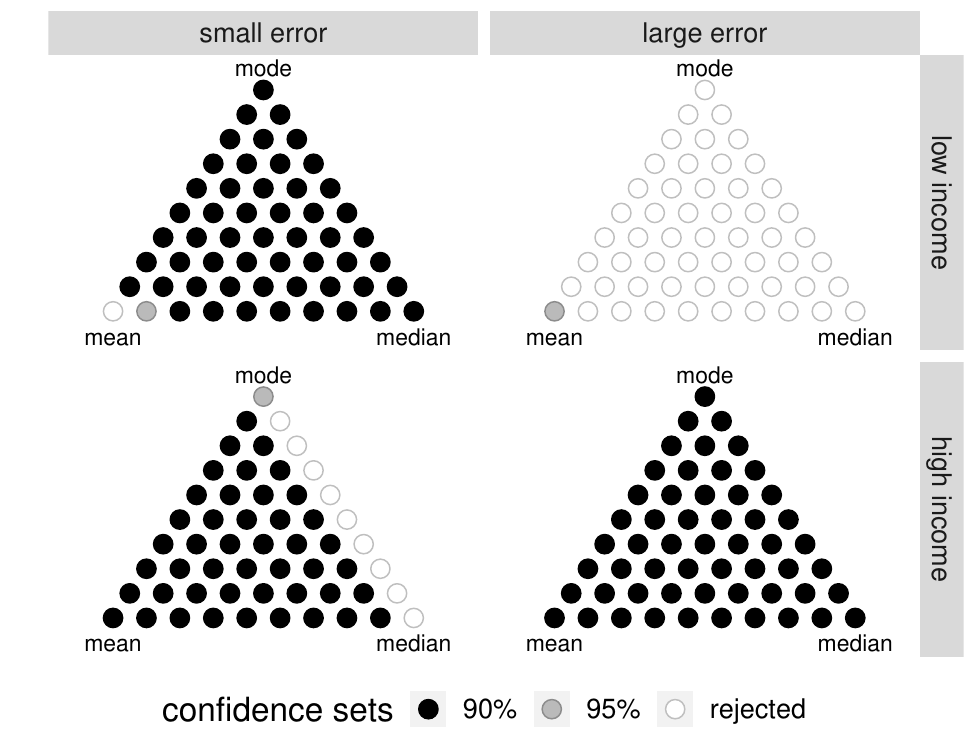}
	\caption{\textbf{Confidence sets for income survey forecasts, stratified by income and shock size.} This figure shows the measures of centrality that rationalize the New York Federal Reserve income survey forecasts, for low- and high-income respondents with small and large income shocks with clustered covariance estimator. 
		%Income shocks are denoted by the absolute forecast error divided by the forecast in the first survey. Black dots indicate that the measure is inside the Stock-Wright 90\% confidence set, dark grey dots indicate that the measure is inside the 95\% confidence set, light grey dots indicate that rationality for that measure can be rejected at the 5\% level. All panels use a constant and the forecast as test instruments.
	}
	\label{fig:ResultsSCEIncomeErrorC}
\end{figure}

\subsection{Further results for the EKT tests on subsamples}
\label{sec:SupplEKTSubsamples}

Table \ref{tab:sce_samples2} shows the test results for the subsamples considered in Section \ref{sec:ApplicationSCE} that were omitted in Table \ref{tab:sce_samples}.

\bigskip

\begin{table}[h!]
	\centering
	\scriptsize
			\begin{tabular}{lcccccc}
				\toprule 
				\addlinespace
				& n & Mean & Median & Mode & Quantiles & Expectiles \\ 
				\cline{2-7}   		\addlinespace
				Offer,~~~~~low income & 429 & 0.01 & 0.00 & 0.64 & 0.00 [0.63, 0.71] & 0.01 [0.50, 0.70] \\ 
				No offer, low income & 1471 & 0.11 & 0.00 & 0.95 & 0.00 [0.57, 0.61] & 0.05 [0.42, 0.53] \\ 
				Offer,~~~~~high income & 328 & 0.20 & 0.00 & 0.70 & 0.03 [0.63, 0.71] & 0.08 [0.40, 0.63] \\ 
				No offer, high income & 1560 & 0.11 & 0.10 & 0.96 & 0.11 [0.62, 0.66] & 0.04 [0.43, 0.54] \\ \midrule \addlinespace
				%				First round & 2628 & 0.01 & 0.00 & 0.88 & 0.00 [0.61, 0.64] & 0.01 [0.50, 0.58] \\
				%				Second round & 1288 & 0.40 & 0.06 & 0.29 & 0.03 [0.60, 0.64] & 0.21 [0.41, 0.54] \\ \midrule \addlinespace
				First round,~~~~low income & 1263 & 0.03 & 0.00 & 0.84 & 0.00 [0.58, 0.62] & 0.03 [0.50, 0.60] \\ 
				Second round, low income & 638 & 0.05 & 0.00 & 0.06 & 0.00 [0.59, 0.65] & 0.03 [0.32, 0.51] \\ 
				First round,~~~~high income & 1238 & 0.09 & 0.00 & 0.42 & 0.03 [0.63, 0.68] & 0.03 [0.45, 0.57] \\ 
				Second round, high income & 650 & 0.52 & 0.20 & 0.36 & 0.29 [0.60, 0.66] & 0.28 [0.40, 0.56] \\ \midrule \addlinespace
				Large error, low income & 376 & 0.08 & 0.00 & 0.00 & 0.00 [0.53, 0.61] & 0.03 [0.34, 0.56] \\ 
				Small error, low income & 262 & 0.04 & 0.30 & 0.76 & 0.50 [0.62, 0.72] & 0.46 [0.17, 0.43] \\ 
				Large error, high income & 267 & 0.87 & 0.72 & 0.49 & 0.95 [0.53, 0.63] & 0.73 [0.36, 0.58] \\ 
				Small error, high income & 383 & 0.56 & 0.08 & 0.10 & 0.03 [0.62, 0.70] & 0.29 [0.41, 0.60] \\
				\addlinespace
				\bottomrule 
				\addlinespace
			\end{tabular}
		\caption{\textbf{Summary of $p$-values for rationality tests in different samples.} The first four columns of this table present the sample size and $p$-values from tests of rationality when interpreting the point forecasts as forecasts of the mean, median, or mode. The last two columns present $p$-values from tests of rationality when interpreting the point forecasts as quantiles or expectiles following \cite{EKT2005}, and 90\% confidence intervals for the asymmetry parameter, given in square brackets.}
		\label{tab:sce_samples2}
\end{table}

\section{Additional Empirical Applications}\label{sec:SuppApplications}

In this section, we present two additional economic applications, to ``Greenbook'' forecasts of US GDP growth, and to random walk forecasts of exchange rates. In some of this applications we find evidence against mode rationality but not against mean rationality. This confirms that our proposed mode forecast rationality test has non-trivial power in relevant applications.

%This illustrates the power and applicability of the developed tests and showcases that the evidence for mode rationality in the income survey forecasts from Section \ref{sec:ApplicationSCE} is by no means universal, and that the mode test has non-trivial power in relevant applications.

\subsection{Greenbook forecasts of U.S.\ GDP growth}\label{sec:SuppGDP}

First, we consider one-quarter-ahead forecasts of U.S.\ GDP growth produced by the staff of the Board of Governors of the Federal Reserve (the so-called ``Greenbook'' forecasts), from 1967Q2 until 2015Q2, a total of $192$ observations.\footnote{Greenbook forecasts are only available to the public after a five-year lag.} These forecasts are prepared in preparation for each meeting of the Federal Open Market Committee, and substantial resources are devoted to them, see e.g.\ \cite{RomerRomer2000}. Greenbook forecasts are available several times each quarter; for this analysis we take the single forecast closest to the middle date in each quarter. Broadly similar results are found when using the first, or last, forecast within each quarter. 

\begin{figure}[t]
	\centering
	\includegraphics[width=.9\linewidth]{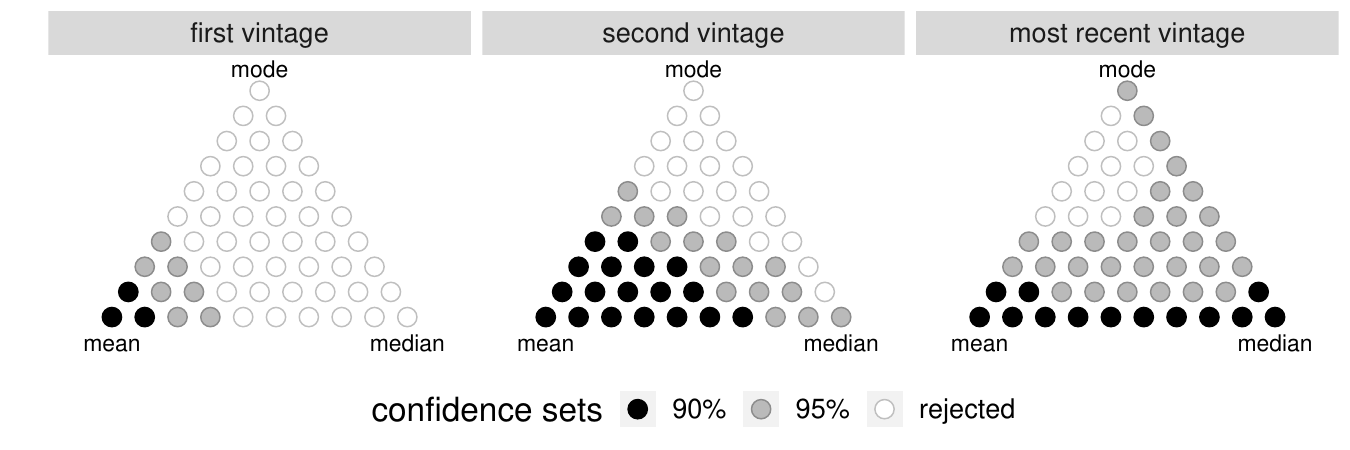}
	\caption{\textbf{Confidence sets for Greenbook GDP forecasts.} This figure shows the measures of centrality that ``rationalize'' the Federal Reserve Board's ``Greenbook'' forecasts of U.S.\ GDP growth. The three panels use three different measures of GDP growth in a given quarter. Black dots indicate that the measure is inside the Stock-Wright 90\% confidence set, grey dots indicate that the measure is inside the 95\% confidence set, and white dots indicate that rationality for that measure of centrality can be rejected at the 5\% level. All panels use a constant and the forecast as test instruments.}
	\label{fig:ResultsGDPall}
\end{figure}

\Cref{fig:ResultsGDPall} presents the confidence set for the measures of centrality that can be rationalized for these forecasts. As GDP growth is measured with error and official values are often revised, we present results for three different ``vintages'' of the realized value: the first, second and most recent release. For the first and second vintages, we see that only measures of centrality ``close to'' the mean can be rationalized as optimal, while the mode, median and similar measures can all be rejected. This is particularly noteworthy given the known lower power at the mode vertex. Using the most recent vintage for GDP growth, both the mean and median, and centrality measures between and near those, are included in the confidence set. That the Greenbook GDP forecasts are rational when interpreted as mean forecasts, but not when taken as mode or median forecasts, is consistent with the Fed staff using econometric models for these forecasts, as such models almost invariably focus on the mean.\footnote{\cite{ReifschneiderTulip2017} discuss the ambiguity in the specific centrality measure reported in the Greenbook forecasts, but write that they are ``typically viewed as modal forecasts'' by the Federal Reserve staff. Our results suggest that they are better interpreted as mean forecasts.}

\subsection{Random walk forecasts of exchange rates}\label{sec:SuppExchange}

For our final empirical application we revisit the famous result of \cite{MeeseRogoff1983}, that exchange rate movements are approximately unpredictable when evaluated by the squared-error loss function, implying that the lagged exchange rate is an optimal mean forecast. See \cite{Rossi2013} for a more recent survey of the literature on forecasting exchange rates. We use daily data from the European Central Bank's ``Statistical Data Warehouse'' on the USD/EUR, JPY/EUR and AUD/EUR exchange rates, over the period January 2000 to July 2020, a total of $5,265$ trading days. Note that our sample period has no overlap with that of \cite{MeeseRogoff1983}, and so their conclusions about the mean-optimality of the random walk forecast need not hold in our data.

\begin{figure}[tb]
	\centering
	\includegraphics[width=.9\linewidth]{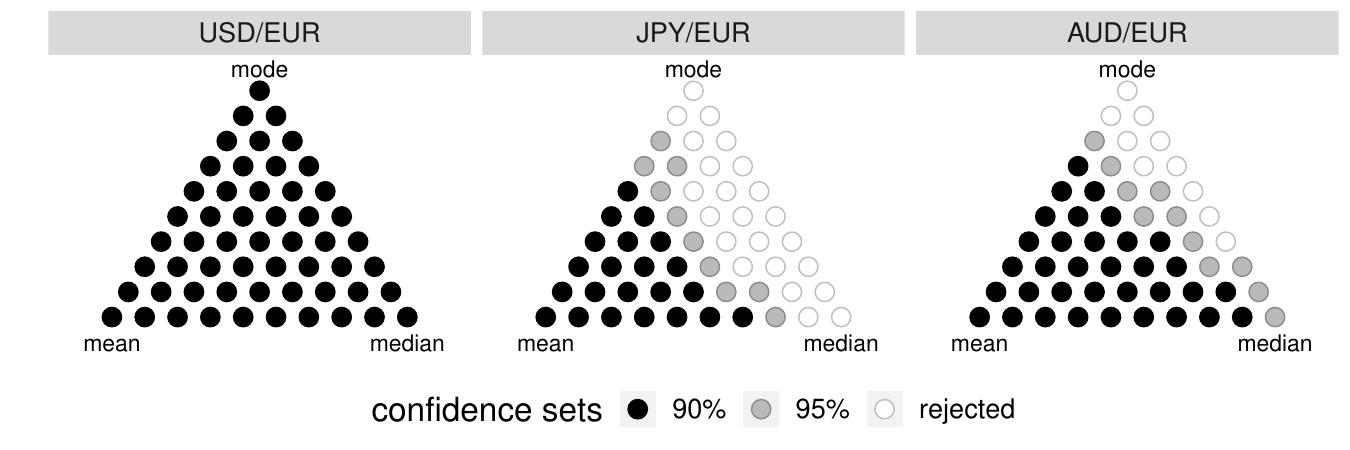}
	\caption{\textbf{Confidence sets for random walk forecasts of exchange rates.} This figure shows the measures of centrality that ``rationalize'' the random walk forecast of daily exchange rates movements. Black dots indicate that the measure is inside the Stock-Wright 90\% confidence set, grey dots indicate that the measure is inside the 95\% confidence set, and white dots indicate that rationality for that measure of centrality can be rejected at the 5\% level. All panels use a constant and the forecast as test instruments.}
	\label{fig:ResultsExchangeRates}
\end{figure}

\Cref{fig:ResultsExchangeRates} presents the results of our tests for rationality, all of which use a constant and the forecast as the instrument set. The middle and right panels reveal that for the JPY/EUR and AUD/EUR exchange rates the lagged exchange rate is not rejected as a mean forecast, while it is rejected when taken as a mode or median forecast. Thus the rationality of the random walk forecast critically depends, for these exchange rates, on whether it is interpreted as a mean, median or mode forecast. %The confidence set is larger for the USD/EUR exchange rate than for the JPY/USD, indicating that more measures of centrality are consistent with rationality for that exchange rate. 
For the USD/EUR exchange rate we cannot reject rationality with respect to \textit{any} of convex combination of these measures of central tendency, implying that the random walk forecast is consistent with rationality under any of these measures.\footnote{Results for the GBP/EUR and CAD/EUR exchange rates are identical to those for the USD/EUR.} The mean vertex being included in the confidence set for all three exchange rates, indicating no evidence against rationality of the random walk model when interpreted as a mean forecast, is consistent with the conclusion of \cite{MeeseRogoff1983}.

%\textcolor{red}{Mean optimality directly reflects the theory of the efficient-market hypothesis, as deviations from the expected exchange rate tomorrow would offer arbitrage opportunities. Interestingly, the rejection of mode and median optimality suggests that the risk profile of the JPY/EUR-exchange rate is skewed. In particular, the results confirm that the random walk's mean forecasts (a no change forecast) is optimal, but challenge the hypothesis that the random walk (as probabilistic forecast with Gaussian innovation) is indeed an optimal probabilistic forecast for exchange rate predictions.}

\subsection{Summary of results across empirical applications}\label{sec:SuppAppliSummary}

Table \ref{tab:summary_samples} summarizes the rationality tests for the mean, median, and mode functionals, as well as the \citet{EKT2005} rationality test that allows for optimism and pessimism. For the Federal Reserve Bank of New York (FRBNY) survey data, mode rationality is not rejected, but all other functionals are rejected. For the Greenbook GDP forecasts, the mode is rejected a the 5\% level, but the mean and expectiles close to the mean are consistent with the data. For the exchange rates random walk forecasts, the mean and central quantiles and expectiles are consistent with the data, while mode rationality is rejected for two of the three exchange rates at the 5\% level.

\begin{table}[tb]
	\centering
	\scriptsize
			\begin{tabular}{lcccccc}
				\toprule 
				\addlinespace
				Data & n & Mean & Median & Mode & Quantiles & Expectiles \\ 
				\midrule   
				\addlinespace
				FRBNY income forecasts & 3916 & 0.01 & 0.00 & 0.77 & 0.00 [0.61, 0.64] & 0.01 [0.49, 0.56] \\ 
				Greenbook GDP forecasts & 192 & 0.35 & 0.06 & 0.03 & 0.02 [0.46, 0.57] & 0.20 [0.45, 0.61] \\ 
				Random walk: USD/EUR & 5264 & 0.94 & 0.32 & 0.17 & 0.35 [0.48, 0.51] & 0.86 [0.48, 0.51] \\ 
				\phantom{Random walk: }JPY/EUR & 5264 & 0.91 & 0.02 & 0.04 & 0.43 [0.47, 0.50] & 0.74 [0.48, 0.51] \\ 
				\phantom{Random walk: }AUD/EUR & 5264 & 0.88 & 0.09 & 0.03 & 0.96 [0.51, 0.53] & 0.62 [0.48, 0.52] \\ 
				\addlinespace
				\bottomrule 
				\addlinespace
			\end{tabular}
		\caption{\textbf{Summary of $p$-values for rationality tests in different samples.} The first four columns of this table present the sample size and $p$-values from tests of rationality when interpreting the point forecasts as forecasts of the mean, median, or mode. The last two columns present $p$-values from tests of rationality when interpreting the point forecasts as quantiles or expectiles following \cite{EKT2005}, and 90\% confidence intervals for the asymmetry parameter are given in square brackets. The first row uses the FRBNY consumer survey data introduced in Section \ref{sec:ApplicationSCE} of the main paper. The second row uses Federal Reserve ``Greenbook'' forecasts of US GDP growth discussed in Section \ref{sec:SuppGDP}. The last three rows uses daily data on exchange rates, discussed in Section \ref{sec:SuppExchange}.}
		\label{tab:summary_samples}
\end{table}

\end{document}